\documentclass[11pt]{article}
\usepackage{microtype}  
\usepackage{array}
\usepackage{multirow}
\usepackage{amsmath} 
\usepackage{amssymb}
\usepackage{amsthm} 
\usepackage{thmtools}

\usepackage{xcolor}  
\usepackage{xspace}
\usepackage{enumitem}    
\usepackage{tikz}   

\usepackage{comment} 

\usepackage{graphicx} 
\usepackage{caption} 
\usepackage{subcaption}  

\usepackage{authblk} 

\usepackage{thmtools} 
\usepackage{thm-restate}

\usepackage{graphicx} 

\definecolor{ForestGreen}{rgb}{0.1333,0.5451,0.1333}
\definecolor{DarkRed}{rgb}{0.8,0,0}
\definecolor{Red}{rgb}{1,0,0}
\usepackage[linktocpage=true,
pagebackref=true,colorlinks,
linkcolor=DarkRed,citecolor=ForestGreen,
bookmarks,bookmarksopen,bookmarksnumbered]
{hyperref}
\graphicspath{{./Figures/}}

\input{commands}
\renewcommand{\paragraph}[1]{\medskip\noindent{\bf #1}\xspace}

\declaretheorem[numberwithin=section]{theorem}
\declaretheorem[numberlike=theorem]{lemma}
\declaretheorem[numberlike=theorem]{proposition}
\declaretheorem[numberlike=theorem]{corollary}
\declaretheorem[numberlike=theorem]{claim}
\declaretheorem[numberlike=theorem]{observation}
\declaretheorem[numberlike=theorem]{invariant}

\theoremstyle{definition}
\declaretheorem[numberlike=theorem]{definition}

\usepackage{cleveref}

\crefname{algorithm}{Algorithm}{Algorithms}
\Crefname{algorithm}{Algorithm}{Algorithms}

\newcommand{\ot}{\tilde{O}}
\newcommand{\oh}{\widehat{O}}

\newcommand{\ignore}[1]{}

\newcommand{\A}{\mathcal{A}\xspace}
\newcommand{\cP}{\mathcal{P}\xspace}

\newcommand{\D}{\mathcal{D}\xspace}
\newcommand{\cF}{\mathcal{F}\xspace}

\newcommand{\eps}{\epsilon}

\newcommand{\poly}{\operatorname{poly}} 
\newcommand{\polylog}{\operatorname{polylog}}

\newcommand{\textlocal}{\operatorname{local}} 
\newcommand{\textbase}{\operatorname{base}} 
\newcommand{\texth}{\operatorname{h}} 
\newcommand{\textsplit}{\operatorname{split}} 
\newcommand{\textcost}{\operatorname{cost}}

\newcommand{\LocalVC}{\operatorname{ LocalVC }}

\newcommand{\SplitVC}{\operatorname{ SplitVC }} 
\newcommand{\baseCaseMainVC}{ \Lambda}

\newcommand{\vol}{\operatorname{vol}}

\newcommand{\minlr}{\min(|L|,|R|)}

\usepackage{bbm}
\usepackage{algorithm,mathtools}
\usepackage[noend]{algpseudocode}

\newcommand{\f}{\frac}
\newcommand{\cd}{\cdot}

\newcommand{\sr}{\sqrt}

\newcommand{\lds}{\ldots}

\newcommand{\s}{\subseteq}

\newcommand{\BE}{\begin{enumerate}}
\newcommand{\EE}{\end{enumerate}}
\newcommand{\im}{\item}
\newcommand{\BI}{\begin{itemize}}
\newcommand{\EI}{\end{itemize}}

\newcommand{\inv}{^{-1}}

\newcommand{\R}{\mathbb R}

\newcommand{\N}{\mathbb N}

\newcommand{\e}{\epsilon}

\newcommand{\De}{\Delta}

\newcommand{\pt}{\partial}
\newcommand{\al}{\alpha}
\newcommand{\be}{\beta}
\newcommand{\om}{\omega}
\newcommand{\Om}{\Omega}
\newcommand{\el}{\ell}

\newcommand{\Th}{\Theta}
\newcommand{\m}{\mathcal}

\newcommand{\lf}{\lfloor}
\newcommand{\rf}{\rfloor}
\newcommand{\lc}{\lceil}
\newcommand{\rc}{\rceil}

\newcommand{\E}{\mathbb E}

\newcommand{\1}{\mathbbm 1}

\newcommand{\lp}{\left(}
\newcommand{\rp}{\right)}

\newcommand{\lmt}{\left[\begin{matrix}}
\newcommand{\rmt}{\end{matrix}\right]}

\newcommand{\BT}{\begin{theorem}}
\newcommand{\ET}{\end{theorem}}
\newcommand{\BL}{\begin{lemma}}
\newcommand{\EL}{\end{lemma}}
\newcommand{\BD}{\begin{definition}}
\newcommand{\ED}{\end{definition}}
\newcommand{\BC}{\begin{corollary}}
\newcommand{\EC}{\end{corollary}}
\newcommand{\BO}{\begin{observation}}
\newcommand{\EO}{\end{observation}}
\newcommand{\BCL}{\begin{claim}}
\newcommand{\ECL}{\end{claim}}
\newcommand{\BP}{\begin{proof}}
\newcommand{\EP}{\end{proof}}
\newcommand{\BPS}{\begin{proof}[Proof (Sketch)]}
\newcommand{\EPS}{\end{proof}}
\Crefname{observation}{Observation}{Observations}
\Crefname{assumption}{Assumption}{Assumptions}
\Crefname{reduction}{Reduction}{Reductions}
\Crefname{claim}{Claim}{Claims}
\Crefname{subclaim}{Subclaim}{Sublaims}

\newcommand{\para}{\paragraph}

\newcommand{\exc}{\textup{excess}}

\newcommand{\tO}{\tilde{O}}
\newcommand{\V}[1]{{V_{\ge {#1}}^p}}
\newcommand{\VV}[1]{{\vol(\V{#1})}}
\newcommand{\bv}{\mathbf v}

\Crefname{invariant}{Invariant}{Invariants}
\newtheorem{remark}[theorem]{Remark}
\newtheorem{fact}[theorem]{Fact}

\newcommand{\sm}{\setminus}

\newcommand{\phihat}{\widehat{\phi}}
\newcommand{\nhat}{\widehat{n}}
\newcommand{\mhat}{\widehat{m}}
\newcommand{\Shat}{\widehat{S}}

\def\richard#1{\marginpar{$\leftarrow$\fbox{R}}\footnote{$\Rightarrow$~{\sf #1 --Richard}}}
\def\danupon#1{\marginpar{$\leftarrow$\fbox{D}}\footnote{$\Rightarrow$~{\sf #1 --Danupon}}}
\def\yu#1{\marginpar{$\leftarrow$\fbox{Y}}\footnote{$\Rightarrow$~{\sf #1 --Yu}}}
\def\sorrachai#1{\marginpar{$\leftarrow$\fbox{S}}\footnote{$\Rightarrow$~{\sf
      #1 --Sorrachai}}}
\def\amm#1{\marginpar{$\leftarrow$\fbox{A}}\footnote{$\Rightarrow$~{\sf #1 --Amm}}}
\def\thatchaphol#1{\marginpar{$\leftarrow$\fbox{TS}}\footnote{$\Rightarrow$~{\sf #1 --Thatchaphol}}}

\def\richard#1{}
\def\danupon#1{}
\def\yu#1{}
\def\sorrachai#1{}
\def\amm#1{}
\def\thatchaphol#1{}

\title{Deterministic Graph Cuts in Subquadratic Time:\\Sparse, Balanced, and $k$-Vertex}

\usepackage{authblk}


\author[1]{Yu Gao}
\author[2]{Jason Li}
\author[3]{Danupon Nanongkai}
\author[1]{Richard Peng}
\author[4]{Thatchaphol Saranurak}
\author[5]{Sorrachai Yingchareonthawornchai}

\affil[1]{Georgia Institute of Technology and Microsoft Research Redmond, USA}
\affil[2]{Carnegie Mellon University, USA}
\affil[3]{KTH Royal Institute of Technology, Sweden}
\affil[4]{Toyota Technological Institute at Chicago, USA}
\affil[5]{Aalto University,  Finland}

\date{}

\begin{document}

\begin{titlepage}
        \maketitle
        \pagenumbering{roman}
        \begin{abstract}
	
	We study deterministic algorithms for computing graph cuts,
	with focus on two fundamental problems: balanced sparse cut 
	and $k$-vertex connectivity for small $k$ ($k=O(\polylog n)$).
	Both problems can be solved in near-linear time with randomized algorithms,
	but their previous deterministic counterparts take at least quadratic time.
	In this paper, we break this bound for both problems. Interestingly, achieving this for one problem crucially relies on doing so for the other.

	In particular, via a divide-and-conquer argument, a variant of the cut-matching game by [Khandekar et al.`07], and the local vertex connectivity algorithm of [Nanongkai et al. STOC'19], we give
	a subquadratic time algorithm for $k$-vertex connectivity using a
	subquadratic time algorithm for computing balanced sparse cuts on sparse graphs. 
	To achieve the latter, we improve the previously best $mn$ bound for approximating balanced sparse cut for the whole range of $m$.
	This starts from (1) breaking the $n^3$ barrier on dense graphs to $n^{\omega + o(1)}$ (where $\omega < 2.372$) using the the PageRank matrix, but without explicitly sweeping to find sparse cuts; to (2) getting the $\tilde O(m^{1.58})$ bound by combining the $J$-trees by [Madry FOCS `10] with the $n^{\omega + o(1)}$ bound above, and finally; to (3) getting the $m^{1.5 + o(1)}$ bound by recursively invoking the second bound in conjunction with expander-based graph sparsification. Interestingly, our final  $m^{1.5 + o(1)}$  bound lands at a natural stopping point in the sense that polynomially breaking it would lead to a breakthrough for the dynamic connectivity problem. 
\end{abstract}

        \setcounter{tocdepth}{2}
        \newpage      
        \tableofcontents
\end{titlepage}

\newpage
\pagenumbering{arabic}

\section{Introduction}
\label{sec:intro}

Graph cuts, or sets of vertices or edges whose removal disconnects graphs,
are fundamental objects in graph theory and graph algorithms.
Efficiently computing graph cuts have a wide range of applications
that include planarity testing~\cite{HopcroftT73}, image processing~\cite{BoykovVZ99},
and high performance/parallel computing~\cite{BulucMSSS16}.
While a very large gap remains between theoretical and practical performances,
graph partitioning algorithms have also proven to be a powerful theoretical tool,
with prominent applications including linear systems solving~\cite{SpielmanT11},
approximation algorithms for unique games~\cite{Trevisan08,AroraBS15},
and dynamic graph data structures~\cite{NanongkaiS17,Wulff-Nilsen17,NanongkaiSW17}.

Due to the central role of cuts in graph algorithms,
they are a natural focus in the study of more efficient graph algorithms.
Many improvements in the running times of cut-related graph
algorithms~\cite{KargerKT95,Karger00:journal,
        SpielmanTengSolver:journal,Sherman13,KelnerLOS14,Madry10}
stem from better understandings of
randomized algorithmic primitives:
there is a polynomial factor separation between the best randomized
and the best deterministic algorithms for many problems on graphs.

On the other hand, deterministic algorithms have a multitude of advantages
over randomized ones.
Theoretically, this is perhaps most evident in data structures,
where an adaptive adversary can choose the next operation based
on the previous output of the data structure.
This resulting dependency is not handled by the analysis of many randomized data
structures~\cite{KelnerL13,FahrbachMPSWX18}, and is only
fixable in isolated situations using more intricate tools for analyzing probabilistic processes~\cite{CohenMP16,KyngS16,KyngPPS17}.
The highly efficient performance of the randomized algorithms also make it
difficult to apply more general purpose derandomization tools (e.g.  \cite{nisan1994hardness,Umans2003pseudo,CarmosinoIS18}),
as many such tools could
potentially incur overheads of polynomial factors.
Historically, the derandomization efforts also led to powerful tools that are useful beyond the derandomization itself. For the case of cuts, a recent example is  
Thorup and Kawarabayashi's edge connectivity algorithm \cite{KawarabayashiT15} (see also \cite{HenzingerRW17,LoST18}), which is a deterministic counterpart of Karger's prominent randomized near-linear time algorithm \cite{Karger00:journal}.
Techniques from \cite{KawarabayashiT15} have later found applications in, e.g., distributed algorithms \cite{DagaHNS19}, dynamic algorithms \cite{GoranciHT18}, and querying algorithms \cite{RubinsteinSW18} (some of these algorithms are randomized).

\paragraph{Near-linear Randomized vs Quadratic Deterministic.}
While deterministic algorithms for edge connectivity are well understood, there remain big gaps between deterministic and randomized algorithms for many other graph cut problems.
Among such problems, two well-known ones are {\em $k$-vertex connectivity}, and {\em approximate sparsest cut} along with its generalization to {\em balanced sparse cut}. On sparse graphs, these problems can be solved in {\em near-linear} time with randomized algorithms, but their previous deterministic counterparts take at least {\em quadratic time}:

(I) The first problem asks whether $k$ vertices can be removed to disconnect the graph for a given parameter $k$. Note that on sparse graphs (when $m=\tilde O(n)$), it can be assumed that $k=O(\polylog(n))$ \cite{NagamochiI92}.\footnote{We use $n$ and $m$ to denote the number of nodes and edges respectively, and use $\tilde O$ to suppress $\polylog(n)$.} 
There has been a long line of work on this problem  (e.g. \cite{Kleitman1969methods,Podderyugin1973algorithm,EvenT75,Even75,Galil80,EsfahanianH84,Matula87,BeckerDDHKKMNRW82,LinialLW88,CheriyanT91,NagamochiI92,CheriyanR94,Henzinger97,HenzingerRG00,Gabow06,NanongkaiSY19}). The problem was recently shown to admit a randomized algorithm that takes $\tilde O(m+ nk^3)$ time \cite{NanongkaiSY19,NanongkaiSY19_linear,ForsterY19}, thus {\em near-linear} time on sparse graphs.
In contrast, the best deterministic algorithm, due to Gabow~\cite{Gabow06}, takes $O(m\cdot(n+\min(k^{5/2},kn^{3/4})))$ time, which is $\tilde O(n^2)$ on sparse graphs.
Note, though, that for such $k$ the $\tilde O(n^2)$ bound dates five decades back to the result of Kleitman \cite{Kleitman1969methods}. Gabow's and the preceding deterministic algorithms (e.g. \cite{HenzingerRG00,Even75,FederM95}) improved over Kleitman's bound only when the input graph is dense enough.  
In fact, no improvement over Kleitman was known even for $k=4$ (the case of $k<4$ was long known to admit near-linear time \cite{Tarjan72,HopcroftT73}). 
See, e.g., \cite[Chapter~15]{Schrijver-book} and \cite{Gabow06} for further surveys.

(II) The second problem asks for a cut that approximately minimizes the {\em conductance}, which is the ratio between the number of cut edges and the volume (the sum of degrees) of the smaller side of the cut (we will make this, as well as other high level definitions, more precise in \Cref{sec:overview}).
A harder version of this problem, the balanced sparse cut problem (\Cref{def:ExpansionApprox}), additionally requires that the two sides of the cut are (approximately) as equal as possible in terms of volumes. 
Our focus is on algorithms for these problems with small approximation
factors in both the cut size and the balance.
With randomization, both versions admit such algorithms with near-linear time complexity
(e.g.~by flow-based algorithms \cite{KhandekarRV09,KelnerLOS14,Sherman13,DBLP:conf/soda/Peng16} or spectral-based algorithms \cite{SpielmanT04,AndersenCL06,OrecchiaV11,OrecchiaSV12}), and such
routines are widely used primitives in efficient graph algorithms.

However, the previous best deterministic algorithms require $\Omega(mn)$
time (e.g. by computing PageRank vectors \cite{AndersenCL06}). 
On dense graphs this bound can be improved to $O(n^\omega)$ for sparsest cut \cite{Alon86,PanC99}, where $\omega< 2.372$ is the matrix multiplication exponent, but not for the balanced version that
underlies most uses of sparse cuts in efficient algorithms.
\footnote{While there is a reduction from sparsest cut to balanced sparsest cut in the approximation algorithms literature, it may need to iterate up to $\Theta(n)$ times, giving a worse overall total than the $O(nm)$ bound.}

Furthermore, approximating balanced sparse cuts on sparse graphs is already understood to be an extremely important graph theoretic primitive~\cite{NanongkaiS17,Wulff-Nilsen17,NanongkaiSW17}.
For any constant $\theta>0$, achieving an $O(n^{1.5-\theta})$-time deterministic algorithm for the balanced sparse cut problem on sparse graphs would imply a major breakthrough in dynamic graph algorithms, namely a polynomial improvement over the classic deterministic algorithm for the dynamic graph connectivity problem \cite{Frederickson85,EppsteinGIN97}.
On sparse graphs, no time bound better than $n^2$ was known for balanced sparse cut or even the easier sparsest cut problem.

\subsection{Our Contributions}\label{sec:intro:contributions}

In this paper, we present the first sub-quadratic time algorithms on sparse graphs for all the above problems. 
Our result for the $k$-vertex connectivity problem is as follows.

\begin{theorem}
        \label{thm:VertexCutMain}
        There is an algorithm that takes an undirected unweighted graph
        with $n$ vertices and $m$ edges, along with a threshold $k$, where $k < n^{1/8}$,
        and outputs a subset $S$ of size less than $k$
        whose removal from $G$ disconnects it into at least two components,
        or that no such subsets exist,
        in time $\oh(m
        + \min\{ n^{1.75} k^{1 + k/2}, n^{1.9} k^{2.5}\})$.%
        \footnote{We use the $\oh(\cdot)$ notation to hide sup-polynomial
                lower order terms.
                Formally $\oh(f(n)) = O(f(n)^{1 + o(1)})$, or that for
                any constant $\theta > 0$, we have $\oh(f(n)) \leq O(f(n)^{1 + \theta})$.
                It can be viewed as a direct generalization of the $\ot(\cdot)$ notation
                for hiding logarithmic factors, and behaves in the same manner.}
\end{theorem}

\begin{table}[t]
        \begin{tabular}{>{\centering}p{0.35\columnwidth}|>{\centering}p{0.15\textwidth}|>{\centering}p{0.05\textwidth}|>{\centering}p{0.17\textwidth}|>{\centering}p{0.16\textwidth}}
                Method  & $f(\phi)$  & $\beta$  & Rand. Runtime  & Det. Runtime\tabularnewline
                \hline 
                Spectral / Cheeger Cut\\ \cite{Alon86}  & $O(\phi^{1/2})$  & $n$  & $O(m\phi^{-1/2}\log{n})$  & $O(n^{\omega})$\tabularnewline
                LP/SDP rounding \cite{LeightonR99,AroraRV09,AroraK16}  & $O(\phi\sqrt{\log{n}})$  & $O(1)$  & $O(n^{2})$  & $\poly(n,m)$\tabularnewline
                Local PageRank \cite{AndersenCL06,AndersenCL08,AndersenL08,AndersenP09}  & $O(\phi^{1/2}\log{n})$  & $O(1)$  & $\oh(m\phi^{-1})$  & $O(nm)$\tabularnewline
                Single Commodity Flows \cite{KhandekarRV09,KhandekarKOV07,Sherman09,DBLP:conf/soda/Peng16}  & $O(\phi\sqrt{\log{n}})$  & $O(1)$  & $\oh(m)$  & \textemdash{}\tabularnewline
                \hline 
                (A) This paper [\Cref{thm:pagerank-main}]  & $O(\phi^{1/2}\log n)$  & $O(1)$  & \textemdash{}  & $O(n^{\omega})$\tabularnewline
                (B) This paper [\Cref{cor:OneShotMadry}] & $O(\phi^{1/2}\log^{2.5} n)$ & $O(1)$  & \textemdash{}  & $\oh(m^{\frac{2\omega-1}{\omega}})$\tabularnewline
                (C) This paper [\Cref{cor:Recursive}]  & $\phi^{1/2}n^{o(1)}$  & $O(1)$  & \textemdash{}  & $\oh(m^{1.5})$ \tabularnewline
        \end{tabular}
        
        \caption{Previous Results for Approximating Balanced Cuts. Recall from \Cref{def:ExpansionApprox}
                that $f(\phi)$ is the loss in conductance, and $\beta$ is the loss
                in balance.}
        \label{table:BalCutPrevious} 
\end{table}

Our key tool for obtaining this running time is an $\oh(m^{1.5})$-time algorithm for the balanced sparse cut problem. The is the first subquadratic-time algorithm for both balanced sparse cut and sparsest cut. All our results are summarized in  \Cref{table:BalCutPrevious}. We only explain these results roughly here, and defer detailed discussions to \Cref{sec:overview}. Roughly, an algorithm for the sparsest cut problem is given a parameter $\phi<1$, and must either output that the input graph has conductance at least $\phi$, or output a sparse cut, with conductance $f(\phi)\geq \phi$, where we want $f(\phi)$ to be as close to $\phi$ as possible. (Ideally, we want $f(\phi)=\phi$, but $f(\phi)=O(\phi^{\Omega(1)}n^{o(1)})$ is typically acceptable.) 
For the balanced sparsest cut, we have parameter $\beta$ indicating the balancedness of the output. Those algorithms with $\beta=O(1)$ (i.e. all, but the first) can be used to solve the balanced sparse cut problem. 
In \Cref{table:BalCutPrevious}, the most important parameter to compare our and previous algorithms are the time bounds in the last column. (All algorithms guarantee acceptable values of $f(\phi)$, and keep in mind that the first algorithm does not work for the balanced sparse cut problem.)
We present three algorithms for balanced sparse cut:
\begin{itemize}[noitemsep]
\item The first algorithm (A) guarantees the same $f(\phi)$ as the previous $O(nm)$ time algorithm, but takes $O(n^\omega)$. It thus improves the previous $O(mn)$ time bound for dense graphs and match the $O(n^\omega)$ time bound previously hold only for the sparsest cut problem. 
\item The second algorithm (B) guarantees slightly worse $f(\phi)$ than our first algorithm, but with lower time complexity ($\oh(m^{\frac{2\omega-1}{\omega}})=O(m^{1.578})$). Its time complexity is subquadratic for sparse graphs. 
\item The third algorithm (C) guarantees an even worse $f(\phi)$, but with a even better time bound ($\oh(m^{1.5})$). What is most interesting about this bound is that it lands at a natural stopping point in the sense that polynomially  polynomially improving it (even with a slightly worse $f(\phi)$) would lead to a breakthrough for the dynamic connectivity problem as discussed above. 
\end{itemize}

While there remains significant gaps in the performances of our methods and
their randomized counter parts (which we will discuss in Section~\ref{sec:Related}),
we believe our investigation represents a natural stopping point for a first step
on more efficient deterministic graph cut algorithms. As mentioned earlier, improving our $\oh(m^{1.5})$ bound further would lead to a major breakthrough in dynamic graph algorithms.

Subsequently, a result involving a subset
of the authors of this paper~\cite{ChuzhoyGLNPS19:arxiv}
gave a determinstic algorithm for computing balanced cuts
with $n^{o(1)}$ approximation in time $\oh(m)$.
While this result supersedes our third balanced-cut
algorithm (C) for all values of $\phi$,
its high approximation error of $n^{o(1)}$ means
our first two algorithms (A) and (B)
still give better approximations when $\phi > n^{-o(1)}$.
Furthermore, the $\oh(m^{1.5})$ runtime overhead in approximating
minimum vertex expansion (Theorem~\ref{thm:VertexExpansionMain}) means
the result in~\cite{ChuzhoyGLNPS19:arxiv} does not immediately imply
faster approximate vertex expansion routines.
It can also be checked that the newer
deterministic approximate vertex expansion bounds
in~\cite{ChuzhoyGLNPS19:arxiv}
\footnote{Sections~7.6 and~7.7 of
	{https://arxiv.org/pdf/1910.08025v1.pdf}}
also don't improve the overall running times as stated in
Theorem~\ref{thm:VertexCutMain}.
It remains open to obtain almost-linear time algorithms
for deterministically approximating vertex expansion
(which would imply $k$-vertex connectivity in $\oh(m+n^{1.5})$ time
for constant values of $k$),
or almost-linear time algorithms for deterministically computing
$k$-vertex connectivity.

\paragraph{Techniques.} An interesting aspect of our techniques is
the inter-dependencies between the results.

First, to obtain \Cref{thm:VertexCutMain}, we need a subquadratic-time algorithm for the balanced sparse cut problem, in particular the $\oh(m^{1.5})$-time algorithm.
This is because, based on the deterministic local algorithms of \cite{NanongkaiSY19,ChechikHILP17}, we can construct a deterministic graph partitioning scheme which runs in subquadratic time as long as we have access to a deterministic procedures for approximating the ``vertex expansion'' of a graph.
Then, we construct such procedure by relating the vertex expansion problem to the balanced sparse cut problem via a variant of the cut-matching game \cite{KhandekarKOV07}. 

Secondly, to compute balanced sparse cuts in subquadratic time on sparse graphs, we invoke a deterministic version of the $J$-trees by Madry \cite{Madry10} to reduce to solving the same problem on dense graphs. Here existing algorithms for dense graph are not efficient enough for the a subquadratic time balanced sparse cut algorithm, so we develop a new $O(n^\omega)$-time algorithm for dense graphs. For this, we show how to find balanced sparse cuts from several PageRank vectors {\em without sweeping}; instead we can look at the volumes of some cuts and do binary search. 
Combining this $O(n^\omega)$-time algorithm together with $J$-trees leads to a time bound of $\oh(m^{\frac{2\omega-1}{\omega}})$. By recursive invocations of such balanced sparse-cut routines in conjunction with expander-based graph sparsifications, we finally obtained the final time bound of $\oh(m^{1.5})$.

\section{Overview}\label{sec:overview}

In this section, we briefly outline our techniques.
As from hereon the results will be stated in their fullest formality,
we will introduce notations as we proceed.
A summary of these notations is in Appendix~\ref{sec:notations}.

Our graphs will be represented using $G = (V, E)$,
and we will use $n = |V|$ and $m = |E|$ to denote the number of edges and vertices respectively.
We assume all graphs are connected because otherwise,
we either have a trivial cut, or can run our algorithms
on each of the connected components separately.

\subsection{$k$-Vertex-Connectivity}

We say that $G=(V,E)$ is $k$-vertex-connected (or simply $k$-connected) if there is no set $S \subset V$ of size $|S|<k$ where $G[V-S]$ has more than one connected component.
We need the following notion:
\begin{definition}
        \label{def:separation-triple}
        A separation triple $(L, S, R)$
        is a partition of vertices such that $L,R \neq \empty$ and there is no edge between $L$ and $R$.
The size of $(L, S, R)$ is $|S|$. \end{definition}

Checking whether a graph is $k$-vertex-connected is then equivalent
to finding a separation triple of size less than $k$.
We say that a pair of vertices $x,y$ is $k$-connected if there is no separation triple of size less than $k$ where $x\in L$ and $y\in R$.

Our deterministic algorithm for checking if $G$ is $k$-connected 
is based on a divide-and-conquer approach.
This is done by first exhibiting a sequence of structural results
in Section~\ref{sec:structure-vc}, %
and then providing a divide-and-conquer algorithm
in Section~\ref{sec:vc-algo}.  
The key structural theorem from Section~\ref{sec:structure-vc} is the following:
\begin{restatable}[]{theorem}{StructureVCMain}
        \label{thm:vertex-cut-characterization}
        For any separation triple $(L,S,R)$, consider forming the graph
        $H_L$ by 
        \begin{enumerate} [noitemsep,nolistsep]
                \item removing all vertices of $R$,
                \item replacing $R$ with a clique $K_{right}$ of size $k$, and
                \item adding a biclique between $S$ and $K_{right}$,
        \end{enumerate}
        as well as $H_R$ symmetrically.
        Then $G$ is $k$-connected if and only if
        \begin{enumerate}[noitemsep,nolistsep]
                \item $|S|\ge k$.
                \item Both the graphs $H_{L}$ and $H_R$ are $k$-connected.
                \item Each pair $x, y \in S$, $x$ and $y$ are $k$-connected.
        \end{enumerate}
\end{restatable}

Note that checking if $x,y$  are $k$-connected can be done using augmenting-path based max-flow algorithms in $O(mk)$ time \cite{EvenT75}.
Therefore, the time for checking  the third condition is small when $k$ and $|S|$ are small. This theorem naturally motivates a divide-and-conquer algorithm where we recurse
on both $H_L$ and $H_R$. In order to reduce the number of recursion levels, it is useful to
find a small $S$ that splits the vertices as even as possible. This motivates the following notion:

\begin{definition}
        \label{def:vertex-expansion}
        The \emph{vertex expansion} of a separation triple $(L,S,R)$ is
        \[
        h\left(L,S,R\right)
        =
        \frac{\left|S\right|}
        {\min\left\{\left|L\right|,\left|R\right|\right\} + \left|S\right|},
        \]
        and the vertex expansion of $G$ is
        $h(G) = \min_{(L,S,R)}h(L,S,R)$. A \emph{$c$-approximation to the minimum vertex expansion} is a separation triple whose vertex expansion at most $c\cdot h(G)$.
\end{definition}

At high level, our recursive algorithm works roughly as follows. We assume we can compute quickly
an $n^{o(1)}$-approximation $(\hat{L},\hat{S},\hat{R})$ to the minimum vertex expansion. If $G$ has low vertex expansion, then we apply the structural theorem on such $(\hat{L},\hat{S},\hat{R})$ and recurse on  
both $H_{\hat{L}}$ and $H_{\hat{R}}.$ Otherwise, $G$ has high vertex expansion. In this case, observe that any separation triple $(L,S,R)$ of size less than $k$ must be very unbalanced, i.e., either $|L|$ or $|R|$ is very small. Now, this is exactly the situation where we can use
local vertex connectivity routine
from~\cite{NanongkaiSY19,ChechikHILP17} for quickly detecting such separation triple $(L,S,R)$.
 By careful implementation of this idea together with some
standard techniques, the performance of our algorithm can be formalized as:

\begin{restatable}[]{theorem}{VCAlgoMain}
\label{thm:VCAlgoMain}
Given a routine that computes $n^{o(1)}$-approximations
to the minimum vertex expansion of an undirected unweighted
graph $G$ with $m$ edges in time $m^{\theta}$ for some $\theta > 1$,
we can compute a $k$-vertex cut,
or determine if none exists, when $k < n^{1/8}$, in time
\[
\oh\left( m + 
\min\left\{
n^{1 + \frac{1}{2}\theta} k^{1 + \frac{k}{2}},
n^{1 + \frac{3}{5}\theta} k^{\frac{8}{5}+\frac{3}{5}\theta}
\right\}
\right).
\]
\end{restatable}

Proving this is the main goal of Section~\ref{sec:vc-algo}.
For our eventual value of $\theta = 1.5 + o(1)$ in \Cref{subsection:Bananas},
this gives a running time of
$\oh(m
+ \min\{ n^{1.75} k^{1 + k/2}, n^{1.9} k^{2.5}\})$ respectively.

\subsection{From Vertex Expansion to Edge Conductance}

To approximate  vertex expansion of a graph,
we relate it to its much more well-studied edge analog,
namely \emph{conductance}.
For any graph $G = (V,E)$, the conductance of a cut $S \subset V$ is 
$
\Phi\left(S\right)
=
\frac{\left|E_{G}\left(S,V \setminus S\right)\right|}
{\min\left\{\vol\left(S\right),\vol\left(V \setminus S\right) \right\}}
\label{eq:def:conductance}
$ where $\vol(S)=\sum_{u\in S}\deg u$.
The conductance $\Phi(G)$ of a graph $G$ is the minimum conductance of a subset of vertices, i.e., $\Phi(G)=\min_{\emptyset\neq S\subset V}\Phi(S)$.
Recall that $\Phi(G)$ is NP-hard to compute \cite{LeightonR99}.
 Most efficient algorithms make the following bi-criteria approximation: 

\begin{restatable}{definition}{ExpansionApprox}
        \label{def:ExpansionApprox}
        A subset $S \subseteq V$ of a graph $G=(V,E)$ is a
        \emph{$(\phihat,c)$-most-balanced $\phi$-conductance cut}
        for some parameters $\phi$, $\phihat$, and $c$ if
        \begin{enumerate}
        \im $\vol(S) \leq m$ and $\Phi(S) \le \phi$, and
        \im any set $\Shat \subseteq V$ satisfying $\vol(\Shat) \leq m$
        and $\Phi(\Shat)\le\phihat$ satisfies $\vol(\Shat) \leq c \cdot \vol(S)$.
        \end{enumerate}
        Furthermore, for a function $f$ where $f(\phi)\ge \phi$ and a value $c \ge 1$,
        we say that an algorithm $\mathcal{A}$ is an
        \emph{$(f(\phi), c)$-approximate balanced-cut} algorithm
        if for any graph $G$ and any parameter $\phi > 0$ given as input,
        it either:
        \begin{enumerate}
        	\item certifies that $\Phi(G) \geq \phi$, or
        	\item outputs a $(\phi, c)$-most-balanced $f(\phi)$-conductance cut.
        \end{enumerate}
\end{restatable}

Our algorithmic definitions allow for general functions that transform conductances
because the Cheeger-based algorithms~\cite{Alon86,OrecchiaSV12} take $\phi$ to $1/2$ powers.
A more detailed description of previous graph partitioning algorithms
in terms of this formulation is in Section~\ref{sec:Related}.
This notion of approximation helps us connects all algorithmic component throughout the paper.
First, in Section~\ref{sec:cutmatching}, we show a reduction that,
given a $f(\phi), c)$-approximate balanced-cut routing,
we can obtain an $n^{o(1)}$-approximation to the minimum vertex expansion as we need from the previous section:

\begin{restatable}[]{theorem}{VertexExpansionMain}
        \label{thm:VertexExpansionMain}
        Given any $(f(\phi), c)$-approximate balanced-cut routine
        $\textsc{ApproxBalCut}$ such that
        $f(\phi) \leq \phi^{\xi} n^{o(1)}$ for some absolute constant $0<\xi\le 1$,
        we can compute an $n^{o(1)}$-approximation to the minimum vertex expansion
        on a graph with $n$ vertices and $m$ edges
        by invoking $\textsc{ApproxBalCut}$ a total of $O(c \log{n})$ times,
        each time on a graph with $n$ vertices and maximum degree $O(c \log{n})$,
        plus a further deterministic overhead of $\oh(c m^{1.5})$.
\end{restatable}

The key to \Cref{thm:VertexExpansionMain} is the cut-matching game framework by Khandekar et
al.~\cite{KhandekarKOV07,KhandekarRV09}.
We observe that the cut-matching game variant by Khandekar, Khot, Orecchia, and Vishnoi~\cite{KhandekarKOV07}
can be seen as a general reduction that allow us to can reduce the problems of approximating various notions of graph expansion (e.g. sparsity, vertex expansion, conductance) to the problem of many computing low conductance
cuts in a $O(c \log n)$-regular graph.

This enables us to concentrate on approximating low conductance
balanced cuts, which is our second main contribution.

\subsection{Approximating Low-Conductance Balanced Cuts}

We develop three $(f(\phi), c)$-approximate balanced-cut routines. The first algorithm, based on computating the PageRank matrix, breaks the $n^3$ barrier on dense graphs and have running time $O(n^\omega)$. We then combine it with the $j$-tree technique of Madry~\cite{Madry10} previously developed for randomized algorithms and obtain an algorithm with running time 
$\oh(m^{\frac{2\omega-1}{\omega}})$ but has slightly worse approximation. This breaks $n^2$ barrier for sparse graphs. Finally, we speed up the running time further to  $\oh(m^{1.5} )$ by recursively invoking the second bound in conjunction with expander-based graph sparsification.

 We start with the formal statement of our first algorithm which is proved in Section~\ref{sec:pagerank}:
\begin{restatable}[]{theorem}{PageRankMain}
        \label{thm:pagerank-main}
        There is an $(O(\phi^{1/2} \log{m}), 10)$-approximate balanced-cut
        algorithm that runs in deterministic $O(n^\om)$ time
        on any multigraph $G=(V,E)$ with $n$ vertices, and any parameter $\phi$.
\end{restatable}

This algorithm is a derandomization of the \emph{PageRank-Nibble}
algorithm by Andersen, Chung, and Lang \cite{AndersenCL06}. Roughly speaking,  this algorithm computes the
\emph{PageRank vector }$p_{v}\in\mathbb{R}_{\ge0}^{V}$\emph{ of vertex
$v$ }encoding a distribution of random walk starting at $v$. We first observe that computing the PageRank
vector $p_{v}$ for all $v\in V$ \emph{simultaneously} 
can be easily done in $O(n^{\omega})$ by computing an inverse of some matrix. How to exploits these vector are more subtle and  challenging. 

 A \emph{sweep cut} w.r.t. $p_{v}$ is a cut of the form $V_{\ge t}^{p_{v}}$
where $V_{\ge t}^{p_{v}}=\{u\in V\mid p_{v}(u)\ge t\}$. Checking
if there is a sweep cut $V_{\ge t}^{p_{v}}$ with conductance at most
$\phi$ in can be easily done in  $O(m)$. To do this, we compute $|E(V_{\ge t}^{p_{v}},V-V_{\ge t}^{p_{v}})|$
and $\vol(V_{\ge t}^{p_{v}})$ of \emph{all }$t$ in $O(m)$, by sorting
vertices according their values in $p_{v}$ and ``sweeping'' through
vertices in the sorted order. Unfortunately, spending $O(m)$ time
for each vertex $v$ would give $O(mn)$ time algorithm which is again
too slow.\footnote{To the best of our knowledge, there is no deterministic data structure
even for checking whether $|E(S,V-S)|>0$ in $o(m)$ time, given a
vertex set $S\subset V$. That is, it is not clear how to approximate
$|E(V_{\ge t}^{p_{v}},V-V_{\ge t}^{p_{v}})|$ in $o(m)$ time even
for a fix $t$.}

To overcome this  obstacle, we show a novel way to obtain a sweep
cut \emph{without} approximating the cut size $|E(V_{\ge t'}^{p_{v}},V-V_{\ge t'}^{p_{v}})|$
for any $t'$. We exploit the fact the sweep cut
is w.r.t. a PageRank vector $p_{v}$ and not some arbitrary vector.
This allows us to do a binary search tree for $t$ where the condition
depends solely on the volume $\vol(V_{\ge t}^{p_{v}})$ and not the
cut size $|E(V_{\ge t}^{p_{v}},V-V_{\ge t}^{p_{v}})|$. This is the key to the efficiency of our algorithm.

Next, we in turn use this scheme to speed up the computation of
balanced low-conductance cuts via the $j$-tree constructions
by Madry~\cite{Madry10}.
In Section~\ref{sec:madry}, we show:

\begin{restatable}[]{theorem}{MadryMain}
        \label{thm:MadryMain}
        Given an $(f(\phi), c)$-approximate balanced-cut routine with
        running time $T_{BalCut}(n, m)$,
        along with any integer parameter $k > 0$,
        there is also an $(f(O(\phi \log^3 n)), 10 c)$-approximate balanced-cut
        routine with running time:
        \[
                \oh\left( k \left(m + T_{BalCut}\left(m / k , m \right) \right) \right).
        \]
\end{restatable}

Our strategy for \Cref{thm:MadryMain} is as follows. First, we transform the graph to have constant degree via a standard reduction. At this point, we switch to the \emph{sparsest cut} problem, since it is easier to work with in the steps to follow. Recall that the sparsity $\sigma(G)$
of a graph $G$ is $\sigma(G)=\min_{S}\frac{|E(S,V-S)}{\min\{|S|,|V-S|\}}$, which is similar to conductance but gives equal ``weight'' to each vertex.
Then, we apply Madry's $j$-tree construction on the (bounded degree) graph for a choice of $j$ depending on $k$, obtaining $k$ many $j$-trees such that there exists a $j$-tree with a near-optimal sparsest cut. Moreover, we show that we can assume that this near-optimal sparsest cut has a specific structure: either it cuts only the tree edges of the corresponding $j$-tree, or it cuts only the core. The former case is handled with a simple dynamic programming without recursion, while the latter requires a recursive low-conductance cut algorithm on the core; the algorithm tries both cases and takes the better option. Our algorithm does this for each of the $k$ many $j$-trees and takes the best cut overall, amounting to $k$ recursive low-conductance cut calls. Finally, since the graph has bounded degree, transitioning from sparsest cut back to low-conductance cut incurs only another constant factor loss in the approximation.

An immediate application of \Cref{thm:MadryMain} gives
\begin{corollary}
\label{cor:OneShotMadry}
        There is an $(O(\phi^{1/2} \log^{2.5}{m}), 100)$-approximate balanced-cut
algorithm that runs in deterministic $\oh(m^{\frac{2 \om - 1}{\om}})$ time
on any multigraph $G=(V,E)$ with $m$ edges,  and any parameter $\phi$.
\end{corollary}

\BP[]
Given \Cref{thm:MadryMain} the $(O(\phi^{1/2} \log{m}), 10)$-approximate balanced-cut
        algorithm that runs in deterministic $\oh(n^\om)$ time of \Cref{thm:pagerank-main} and setting $k$ to $m^{\frac{\om-1}{\om}}$ give an
        \[
        \left(O\left(\log n\right)\cdot O\left(\phi^{1/2} \log^{1.5}{m}\right),
        10\cdot 10\right)
        =
        \left(O\left(\phi^{1/2} \log^{2.5}{m}\right),
        100\right)
        \]
        -approximate balanced-cut algorithm with running time 
\[
        \oh\left(m^{\frac{\om-1}{\om}}\left(m+\oh((m/m^{\frac{\om-1}{\om}})^\om)\right)\right)\\
        =\oh\left(m^{\frac{\om-1}{\om}}\left(m+(m/m^{\frac{\om-1}{\om}})^\om\right)\right)\\
        =\oh\left(m^{\frac{2\om-1}{\om}}\right).
\]
\EP

Of course, just as in the construction by Madry~\cite{Madry10},
it is tempting to invoke the size reductions given in
Theorem~\ref{thm:MadryMain} recursively.
Such recursions will lead to a larger overhead on conductance,
which in turn factors only into the running time of the
$k$-vertex-cut algorithm as stated in Theorem~\ref{thm:VCAlgoMain}.
However, we obtain a faster running time by using recursion to
speed up the dense case instead.
By combining balanced cuts with graph sparsification.
We perform a 4-way recursion akin to the one used for
random spanning trees in~\cite{DurfeeKPRS17}
to obtain the following result in Section~\ref{sec:recursion}. 

\begin{restatable}[]{theorem}{RecursionMain}
\label{thm:RecursionMain}
Given any $(f(\phi), c)$-approximate balanced-cut routine
$\textsc{ApproxBalCut}$ in time $m^\theta$ for some $1<\theta\le 2$ such that
$f(\phi) \leq \phi^{\xi} n^{o(1)}$ for some absolute constant $0<\xi\le 1$,
we can obtain an $(n^{o(1)} \cdot f(\phi), c )$-approximate
balanced-cut routine with running time
\[
        \oh\left(n^{2\theta-2}m^{2-\theta}\right) = \oh\left(n^2\right).
\]
\end{restatable}

This result immediately implies a running time of $\oh(n^2)$
for dense graphs, but at the cost of higher approximation factors
compared to the matrix-inverse based one given in
Theorem~\ref{thm:pagerank-main}.
By further combination with previous algorithms,
we obtain our third runtime bound.

\begin{corollary}
\label{cor:Recursive}        

There is an $(\phi^{1/2}n^{o(1)}, 1000)$-approximate balanced-cut
algorithm that runs in deterministic $\oh(m^{1.5})$ time
on any multigraph $G=(V,E)$ with $m$ edges,  and any parameter $\phi$.
\end{corollary}

\BP[Proof]
By \Cref{cor:OneShotMadry} and since $\om < 2.38$, we have an $(O(\phi^{1/2} \log^{2.5}{m}), 100)$-approximate balanced-cut
algorithm that runs in deterministic $m^{1.58}$ time
on any multigraph $G=(V,E)$ with $m$ edges, and any parameter $\phi$. We can apply \Cref{thm:RecursionMain} with $\xi=1/2$ on this algorithm to get an $(n^{o(1)}\cdot O(\phi^{1/2}\log^2 m), c)$-approximate balanced-cut routine with running time $\oh(n^2)$.

Then we give the algorithm to \Cref{thm:MadryMain} with $k=m^{0.5}$ to get an $(n^{o(1)}\cdot O(\phi^{1/2}\log^3 m), 1000)$-approximate balanced-cut algorithm with running time 
\[
\oh\left(m^{0.5}\left(m+\oh\left(\left(m/m^{0.5}\right)^2\right)\right)\right)\\
=\oh\left(m^{0.5}\left(m+\left(m/m^{0.5}\right)^2\right)\right)\\
=\oh\left(m^{1.5}\right).
\]
\EP

\subsection{Putting Everything Together}
\label{subsection:Bananas}

We can now combine the pieces to obtain a sub-quadratic
algorithm for vertex connectivity.

\BP[Proof of \Cref{thm:VertexCutMain}.]
By plugging the algorithm from \Cref{cor:Recursive} to \Cref{thm:VertexExpansionMain} with $c=1000$ and $f(\phi)=\phi^{1/2}n^{o(1)}$. This results in an $n^{o(1)}$-approximation to the minimum vertex expansion
on a graph with $n$ vertices and $m$ edges with running time
\[
	O\left(\log {n}\right)
		\cdot
			\oh\left(\left(n\cdot O(\log n)\right)^{1.5}\right)
	+
	\oh\left(m^{1.5}\right)
	=
	\oh\left(m^{1.5}\right).
\]        
Finally, we can use \Cref{thm:VCAlgoMain} with $\theta=1.5+o(1)$ on the minimum vertex expansion routine above, computing a $k$-vertex cut or determining if none exists, when $k<n^{1/8}$, in time 
\[
\oh\left(m
+
\min\left\{n^{1.75}k^{1+\frac{k}{2}}, n^{1.9}k^{2.5}\right\}\right),
\]
which is the desired result of \Cref{thm:VertexCutMain} when $k\ge 2$.
\EP

We remark that while there are many other ways of combining the
various pieces for graph partitioning,
specifically Lemmas~\ref{lem:JTreeMain} and Theorem~\ref{thm:RecursionMain},
we believe it is unlikely for an algorithm built from just these
pieces to obtain a running time of $o(m^{1.5})$.
A (somewhat cyclic) way of seeing this is to consider the gains
of applying of these tools once when starting from an $m^{\theta}$ time
algorithm (for some $1 < \theta < 2$):
Theorem~\ref{thm:RecursionMain} implies that a dense graph can be solved
in time
\[
n^{2 \theta - 2 } m^{2 - \theta},
\]
which when plugged into the $j$-tree recurrence
given by Lemmas~\ref{lem:JTreeMain} gives a running time of
\[
k \cdot \left(m + \left(\frac{n}{k} \right)^{2\theta - 2 }  m^{2 - \theta} \right).
\]

To simplify this, we first consider the sparse case where $n \approx m$,
where the above runtime simplifies to
\[
k \cdot \left(n + \left(\frac{n}{k} \right)^{2\theta - 2 }  n^{2 - \theta} \right)
=
kn + n^{\theta} k^{3 - 2 \theta}.
\]
Observe that if $\theta > 1.5$, then $k^{3-2\theta} < 1$ and we can set $k = \sqrt{n}$ so that the running time is $kn + n^{\theta} k^{3 - 2 \theta} = n^{1.5}$. For general $m$, this approach speeds up an $m^\theta$-time algorithm to an algorithm with $\oh(m^{1.5})$ running time.
On the other hand, if $\theta \le 1.5$, then $k^{3-2\theta} \ge 1$ and so the running time is at least $n^\theta$, which means that setting $k>1$ does not help. This suggests that any further improvements beyond $o(m^{1.5})$ requires new tools.

\section{Related Works}
\label{sec:Related}

In the broadest sense, our results are related to the derandomization
of graph algorithms, which is a wells studied
topic~\cite{Reingold08,KawarabayashiT15,HenzingerRW17,MurtaghRSV17}.
Deterministic algorithms have a wide ranges practical of advantages,
such as the reproducibility of errors, that make them significantly
more preferable in areas such as numerical analysis and high
performance computing.
While there has been extensive work on pseudorandomness and
derandomization~\cite{BabaiFNW91,nisan1994hardness,ImpagliazzoW97,Umans2003pseudo,CarmosinoIS18},
the pursuit of more efficient deterministic graph algorithms
is much more fine grained.
Even if one can show that $BPP \subseteq P$, it is not clear
that any randomized nearly-linear time algorithm can also be
derandomized to run in nearly-linear time.

With a few exceptions such as computing global minimum
cuts~\cite{KawarabayashiT15,HenzingerRW17}, most problems involving
graph cuts have significantly faster randomized algorithms than
deterministic ones.

\subsection{$k$-Vertex Connectivity}

Vertex connectivity has been studied extensively in graph algorithms.
For the $k \leq 3$ case, determinsitic $O(m)$ time algoirthm were given
by Tarjan~\cite{Tarjan72} and Hopcroft and Tarjan~\cite{HopcroftT73},
and were one of the primary motivations for studying depth first searchers.
For $k\geq 4$, there is a long list of algorithms, starting from the
five-decade-old $O(n^2)$-time algorithm of Kleitman~\cite{Kleitman1969methods}
when $k = O(1)$ and $m = O(n)$ and many others for larger values of
$k$ and $m$~\cite{LinialLW86,KanevskyR87,NagamochiI92,HenzingerRG96}. 

However, even for the case of $k = 4$, no sub-quadratic deterministic algorithms are known.
Even with randomization, a subquadratic time algorithm was only given
recently by Nanongkai et al.~\cite{NanongkaiSY19,NanongkaiSY19_linear}.
Their algorithm runs in subquadratic
(i.e. $O(m+n^{2-\theta})$ for some constant $\theta>0$)
time when $k =o(n^{2/3})$, and near-linear time for constant values of $k$.

\subsection{Graph Partitioning}

Algorithms for approximating minimum conductance cuts have been
extensively studied.
A partial list of results that fit into our notion of
$(f(\phi), \beta)$-approximate balanced-cuts from
\Cref{def:ExpansionApprox} is given in \Cref{table:BalCutPrevious}.

This problem (in general case) can be solved in $\tilde{O}(m)$ time
by randomized algorithms (e.g.~by flow-based algorithms \cite{KhandekarRV09,KelnerLOS14,Sherman13,DBLP:conf/soda/Peng16}
or spectral-based algorithms \cite{SpielmanT04,AndersenCL06,OrecchiaV11,OrecchiaSV12}).
However, previous deterministic algorithms still take $\Omega(mn)$
time by, for example, computing PageRank vectors \cite{AndersenCL06}
from every vertex. 
Such a runtime upper bound is not sufficient for a sub-quadratic time vertex-cut algorithm.

We note that, if the balance guarantee is not
needed, it is known how to solve this problem in $O(n^{\omega})$
since the 80's (see e.g. \cite{Alon86}).
One can compute the second eigenvector $v_2$ of the laplacian matrix of
$G$ deterministically in $O(n^{\omega})$ time \cite{PanC99},
and perform a ``sweep cut'' defined by $v_{2}$ in $O(m)$ time.
As observed in \cite{Alon86,JerrumS88}, the returned cut $S$ will have conductance
at most $O(\sqrt{\phi^{*}})$ if there exists a cut with conductance
$\phi^{*}$. However, there is no guarantee about the balance of $S$.
Our running time for approximate balanced cuts matches the bound.
However, the balance guarantee is crucial for efficiently partitioning
graphs: otherwise one could repeatedly remove $O(1)$ sized vertex
subsets, leading to an $\Omega(n)$ factor overhead.

While undirected graph partitioning takes randomized nearly-linear time via reductions
to approximate maximum flow~\cite{Sherman09,DBLP:conf/soda/Peng16},
in general it is a fundamental question whether there is a deterministic almost
linear time algorithm as well.
Prior to this paper, the best deterministic algorithm takes $O(mn)$ time via
PageRank~\cite{AndersenCL06}.
Our algorithms are the first deterministic approximate balanced
cut algorithm with subquadratic runtime,
as well as the first with subcubic runtime on dense graphs.

\section{Structural Properties of Vertex-Cuts}
\label{sec:structure-vc}

Througout this section, we fix a connected graph $G = (V,E)$
and an arbitrary separation triple $(L,S,R)$ in $G$.

In this section show the top-level structural result that
enables recursive $k$-vertex-connectivity algorithms.

\StructureVCMain*

We first formally define the left subgraph $H_L$ of $G$, and right subgraph $H_R$ of $G$ as follows. 

\begin{definition} [$H_L$ and $H_R$] \label{def:HLHR}
Given a connected graph $G$, a  separation triple $(L,S,R)$, and a positive integer $k$, we define two subgraphs $H_L$ and $H_R$ as follows. We define $H_L = (V_L,E_L)$ where 
\begin{align}\label{eq:HL_graph} 
V_L = V_{1,L} \sqcup V_{2,L} \quad \mbox{and} \quad E_L = E_{1,L} \sqcup E_{2,L} \sqcup E_{3,L}, %
\end{align}
where $\sqcup$	denotes disjoint union of sets, and sets in \Cref{eq:HL_graph} are defined as follows. 

\begin{itemize}[noitemsep]
\item $V_{1,L} = L \sqcup S$.
\item $V_{2,L}$ is the set of $k$ new vertices. 
\item $E_{1,L}$ is the set of edges from the  induced subgraph $G[L \sqcup S]$.
\item $E_{2,L} = \{ (u,v) \colon  u \in S, v \in V_{2,L} \} $. 
\item $E_{3,L} = \{ (u,v) \colon  u \neq v, u \in V_{2,L}, \text{ and } v \in V_{2,L} \} $.
\end{itemize}

Similarly, we define $H_R = (V_R,E_R)$ where 
\begin{align}\label{eq:HR_graph} 
V_R = V_{1,R} \sqcup V_{2,R} \quad \mbox{and} \quad E_R = E_{1,R} \sqcup E_{2,R} \sqcup E_{3,R},%
\end{align}
where sets in \Cref{eq:HR_graph} are defined as follows. 
\begin{itemize}[noitemsep]
\item $V_{1,R} = S \sqcup R$.
\item $V_{2,R}$ is the set of $k$ new vertices. 
\item $E_{1,R}$ is the set of edges from the induced subgraph $G[S \sqcup R]$.
\item $E_{2,R}  = \{  (u,v) \colon  u \in V_{2,R}, v \in S \} $. 
\item $E_{3,R} = \{ (u,v) \colon u \neq v, u \in V_{2,R}, \text{ and } v \in V_{2,R} \} $.  \qedhere
\end{itemize}
\end{definition}

\begin{remark}
From \Cref{def:HLHR}, $H_L[V_{2,L}]$ is a clique of size $k$ such that every vertex in $S$ has an edge to all vertices in $V_{2,L}$. Symmetrically, $H_R[V_{2,R}]$ is a clique of size $k$ such that every vertex $V_{2,R}$ has an edge to all vertices in $S$.
\end{remark}

\subsection{Interactions Between Two Separating Triples}

When we consider two separation triples $(L,S,R)$ and $(L',S',R')$, it is useful to draw a standard \textit{crossing diagram} as shown in Figure \ref{fig:crossing-diagram}. For example, the neighbors of a ``quadrant'', e.g. $L \cap R'$ are contained in parts of the two vertex-cuts.   

\begin{figure}[!h]
	\centering
	\includegraphics[width=0.3\linewidth]{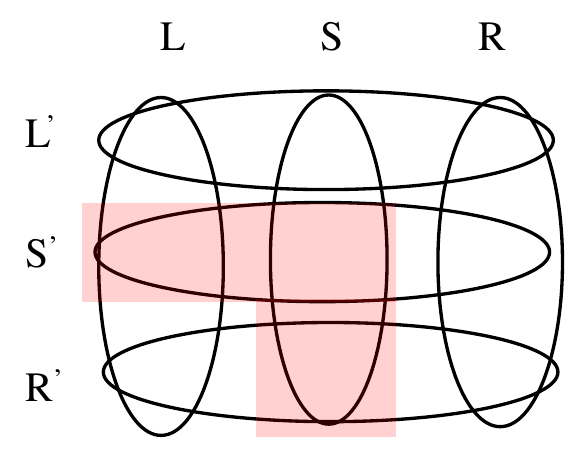}
	\caption{A crossing-diagram for two separation triples $(L,S,R)$ and $(L',S',R')$. The neighbors of $L \cap R'$ is a subset of  $(R' \cap S) \cup (S' \setminus R)$. }
	\label{fig:crossing-diagram}
\end{figure}

\begin{proposition} \label{prop:crossing_neighbors}
	Let $(L',S',R')$ and $(L,S,R)$ be any two  separation triples. We have $N(L \cap R') \subseteq (R' \cap S) \cup (S' \setminus R)$.
\end{proposition}
\begin{proof}
	By \Cref{def:sep-triple}, there is no edge between $L$ and $R$. Also, there is no edge between $L'$ and $R'$. Thus, a neighbor $u$ of $L \cap R'$ cannot be in $L' \cup R$ or in $L \cap R'$. Therefore, if $x \in R'$, then $x \in R' \cap S$, and if $x \not \in R'$, then $x \in S' \setminus R$. 
\end{proof}
\begin{theorem} [\cite{NagamochiI92}]\label{thm:sparsification}
	Given an undirected graph $G = (V,E)$, there is an $O(m)$-time algorithm that partitions $E$ into a sequence of forests $F_k,  k = 1, \ldots, n $ such that  the forest subgraph $H_k = (V, \bigcup_{i=1}^kF_i)$ is $k$-connected if and only if $G$ is $k$-connected. In addition, any vertex set of size $< k$ is a vertex-cut in $G$ if and only if it is a vertex-cut in $H_k$. Furthermore, $H_k$ has aboricity $k$, meaning that $|E(S,S)| \leq k|S|$ for any subset $S \subseteq V$.  
\end{theorem} 

\subsection{Proof of \Cref{thm:vertex-cut-characterization}: Part 1 (Sufficiency)}
We first show that the conditions are sufficient.
To do this, we show the contrapositive: if $G$ has a vertex-cut of size smaller than $k$,
then at least one of the conditions in \Cref{thm:vertex-cut-characterization} is false.

If either condition 3 or 4 is false, then we are done.

Otherwise, we show that $H_L$ or $H_R$ is not $k$-connected. 
 
We now assume that $G$ has a vertex cut of size smaller than $k$. This means $\kappa_G < k$. We denote $(L^*,S^*,R^*)$ as an optimal separation triple. Note that $|S^*| = \kappa_G \leq k-1$. 

\begin{claim}
 We have  $S \cap L^* = \emptyset$  or $S \cap R^* = \emptyset$. 
\end{claim}
\begin{proof}
Suppose otherwise that   $S \cap L^* \not = \emptyset$ and $S \cap R^* \not = \emptyset$. There exists $u \in S \cap L^*$, and $v \in S \cap R^*$. Hence, $u \in S$ and $v \in S$, and so $\kappa_G(u,v) \geq k$. On the other hand, $u \in L^*$ and $v \in R^*$. Therefore, $\kappa_G(u,v) = \kappa_G < k$, which is a contradiction.
\end{proof} 

We now assume WLOG that $S \cap R^* = \emptyset$.  The case $S \cap L^* = \emptyset$ is similar.

\begin{claim} \label{lem:existsxy}
 There exists a vertex $x \in S \cap L^*$, and another vertex $y \in  (L\sqcup R) \cap R^*$.    
\end{claim}
\begin{proof}
 Since $|S| \geq k > |S^*|$ and $S \cap R^* = \emptyset$, there exists a vertex $x \in S \cap L^*$.   Also, since $S \cap R^* = \emptyset$, we have $R^* \cap ( L \sqcup R) \neq \emptyset$. In particular, $L \cap R^* \neq \emptyset$, or $R \cap R^* \neq \emptyset$.
\end{proof}

By \Cref{lem:existsxy}, $y \in L \cap R^*$ or $ y \in R \cap R^*$.  We assume WLOG that $y \in L \cap R^*$. The other case is similar. So far, we have that 
\begin{align}\label{eq:sofar_1}
S \cap R^* = \emptyset, \quad x \in S \cap L^*, \quad  \mbox{ and  } \quad  y \in L \cap R^*.
\end{align}

Figure \ref{fig:crossing-diagram2} shows the corresponding crossing diagram from \Cref{eq:sofar_1} with additional facts from the following claim.

\begin{figure}[!h]
\centering
\includegraphics[width=0.75\linewidth]{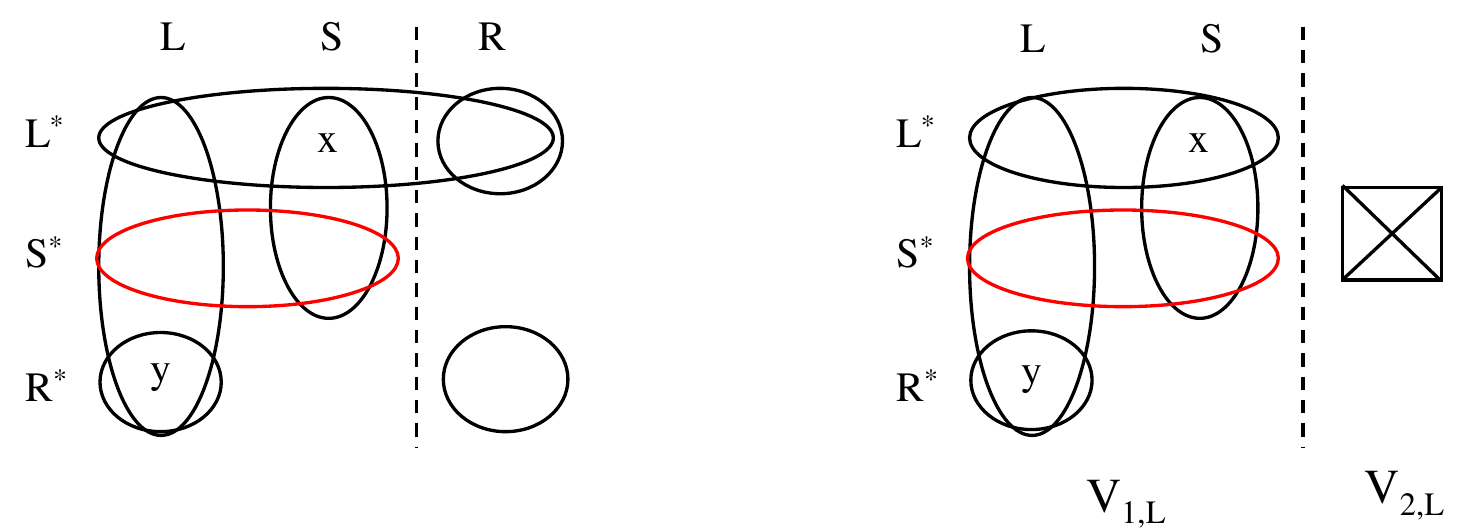}
  \caption{A crossing-diagram for two separation triples $(L,S,R)$ and $(L^*,S^*,R^*)$ before and after transformation from $G$ to $H_L$.}
  \label{fig:crossing-diagram2}
\end{figure}

\begin{claim} \label{lem:collection-of-facts-G} For the two separation triples $(L,S,R)$ and $(L^*, S^*, R^*)$ in $G$,
\begin{enumerate}[noitemsep,nolistsep]
\item \label{item:fact1} $N(L \cap R^*) \subseteq S^* \setminus R$,
\item \label{item:fact2} $S^*$ is an $(x,y)$-vertex-cut in $G$, and
\item \label{item:fact3} $S^* \cap R = \emptyset$. 
\end{enumerate}
\end{claim}
\begin{proof}
We first show that  $N(L \cap R^*) \subseteq S^* \setminus R$.  By \Cref{prop:crossing_neighbors}, $N(L \cap R^*) \subseteq (R^* \cap S) \sqcup (S^* \setminus R)$. By \Cref{eq:sofar_1}, $S \cap R^* = \emptyset$. Therefore, $N(L \cap R^*) \subseteq  S^* \setminus R$. Next, $S^*$ is an $(x,y)$-vertex-cut in $G$. Since $ x\in S \cap L^*$ and $y \in L \cap R^*$, $x \in L^*$ and $y \in R^*$. Therefore. the claim follows.%

We now show that  $S^* \cap R = \emptyset$.  Since $N(L \cap R^*)$ is an $(x,y)$-vertex-cut in $G$, and $S^*$ is the smallest vertex-cut, we have $|N(L \cap R^*)| \geq |S^*|$. Since $N(L \cap R^*)  \subseteq  (S^* \setminus R)  \subseteq S^*$, $ |N(L \cap R^*)| \leq |S^* \setminus R| \leq |S^*|$. Therefore, we have $|S^*| \leq   |N(L \cap R^*)| \leq |S^* \setminus R| \leq |S^*|.$ In particular, $|S^* \setminus R| = |S^*|$.
\end{proof}

It remains to show that $S^*$ is also a vertex-cut in $H_L$ as suggested by Figure \ref{fig:crossing-diagram2}. We now make the argument precise. 
\begin{claim} \label{claim:preserve-HL}
After transformation from $G$ to $H_L$ by \Cref{def:HLHR}, vertices $x$ and $y$ are  in $H_L$. The set $S^*$ and $L \cap R^*$ do not change. In other words, for the left-subgraph $H_L$, we have, 
\begin{enumerate}[nolistsep, noitemsep]
\item  \label{item:HL1} $x \in V_{1,L}$,
\item \label{item:HL2} $y \in V_{1,L}$,
\item \label{item:HL3} $S^* \subseteq V_{1,L}$, and
\item \label{item:HL4} $L \cap R^* \subseteq V_{1,L}$. 
\item \label{item:HL5} $N(L \cap R^*) \subseteq V_{1,L}$. 
\end{enumerate}
\end{claim}
\begin{proof}
By \Cref{def:HLHR}, for any node $v$ in $G$, if $v \in L \sqcup S$, then $v \in V_{1,L}$. We will use this fact throughout the proof. We now show first two items. By \Cref{eq:sofar_1}, $ x \in S \cap L^*$ and $y \in R^* \cap L$. Thus,  $ x \in S $ and $y \in L$. Since $x$ and $y$ are both in the set $L \sqcup S$, $x \in V_{1,L}$ and $y \in V_{1,L}$.  Next, we show that  $S^* \subseteq V_{1,L}$.   By \Cref{lem:collection-of-facts-G} part \ref{item:fact3},  $S^* \cap R = \emptyset$. Thus, $S^* \subseteq L \sqcup S$. Therefore, $S^* \subseteq V_{1,L}$.  Next, we show that  $L \cap R^* \subseteq V_{1,L}$. Since $L \cap R^*$, we have $L \cap R^* \subseteq L$. Therefore, $L \cap R^* \subseteq V_{1,L}$. Finally, we show that  $N_{H_L}(L \cap R^*) \subseteq V_{1,L}$.    By \Cref{lem:collection-of-facts-G} part \ref{item:fact1}, we have $N(L \cap R^*) \subseteq S ^*$ in $G$. Also, $S^* \subseteq V_{1,L}$ in $H_L$. Therefore, $N(L \cap R^*) \subseteq S^*$ in $H_L$.
\end{proof}

\begin{lemma}
The left subgraph $H_L$  is not $k$-connected. 
\end{lemma}
\begin{proof}
Since $|S^*| < k$, it is enough to show that $S^*$ is an $(x,y)$-vertex-cut in $H_L$. To do so, we prove four items (all in $H_L$).
\begin{itemize}[nolistsep,noitemsep]
\item $N(L \cap R^* ) \subseteq S^*$.
\item $y \in L \cap R^* $.
\item  $x \not \in L \cap R^* $.
\item  $x \not \in S^*$. 
\end{itemize}
We prove the first item. By \Cref{lem:collection-of-facts-G} part \ref{item:fact1}, we have $N(L \cap R^*) \subseteq S ^*$ in $G$.   By \Cref{claim:preserve-HL} part \ref{item:HL3}, \ref{item:HL4} and \ref{item:HL5}, the set  $S^*, L \cap R^*, $ and $N(L\cap R^*)$  exist in $H_L$.  By \Cref{def:HLHR}, the new edges in $H_L$  do not join any vertex in $L$. In particular, the new edges do not join any vertex in $L \cap R^*$. This means $N(L \cap R^*)$ does not change after transformation from $G$ to $H_L$. Therefore, $N(L \cap R^*) \subseteq S^*$ in $H_L$. Next, we show that $y \in L \cap R^*$ in $H_L$. By \Cref{eq:sofar_1}, $y \in L \cap R^*$ in $G$. By \Cref{claim:preserve-HL} part \ref{item:HL2}, $y$ exists in $H_L$. Therefore,  $y \in L \cap R^*$ in $H_L$. Next, we show that $x \not \in R^* \cap L$ in $H_L$. By \Cref{eq:sofar_1}, $x \in S \cap L^*$ in $G$.  Hence, $x \not \in L$ and $x \not \in R^*$ in $G$. By \Cref{claim:preserve-HL} part \ref{item:HL1}, $x$ exists in $H_L$.  Therefore, $x \not \in L \cap R^* $ in $H_L$. Finally, we show that $x \not \in S^*$ in $H_L$. By \Cref{lem:collection-of-facts-G} part \ref{item:fact2},  $S^*$ is an $(x,y)$-vertex-cut in $G$. Hence, $x \not \in S^*$ in $G$, which means  $x \not \in S^*$ in $H_L$.

\end{proof}
\begin{remark}
It is possible that $H_R$ is not $k$-connected when the two assumptions from above are different. 
\end{remark}

\subsection{Proof of \Cref{thm:vertex-cut-characterization}: Part 2 (Necessity)}
It remains to show the other direction. That is, we show that if not all the conditions in \Cref{thm:vertex-cut-characterization} are true, then $G$ has a vertex-cut of size smaller than $k$. Before the proof, we start with simple observation. 

\begin{lemma} \label{lem:Sisnotallclique}
For any separation triple $(L',S',R')$ in $H_L (\text{or } H_R)$ such that $|S'| < k$, 
$S' \not \subseteq V_{2,L} (\text{or }  V_{2,R})$. That is, the clique $V_{2,L} (\text{or } V_{2,R})$ in \Cref{def:HLHR} of size $k$ does not contain the vertex-cut $S'$. 
\end{lemma} 
\begin{proof}
We prove the result for $H_L$. The proof for the case $H_R$ is similar. 
Suppose $S' \subseteq V_{2,L}$. Let $H'_L = H_L - S'_H$. The new graph $H'_L$ is essentially the same as $H_L$ except that the modified clique $V'_{2,L} = V_{2,L} \setminus S'$ has size $k - |S'| \geq 1$ (since $|S'| < k$).  In essence, the graph $H'_L$ has the same structure as $H_L$, but with a smaller clique. That is, from the separation triple $(L,S,R)$ in $G$, we obtain the graph $H'_L$ by contracting $R$ into a clique of size at least 1. Since $G$ is  connected and by \Cref{def:HLHR},  $H'_L$ is  connected. Therefore, $S'$ does not disconnect $H_L$, contradicting to the fact that $S'$ is a vertex-cut.   
\end{proof}

\begin{observation} \label{obs:clqiue-one-side}
The clique $V_{2,L} (\text{or } V_{2,R})$ cannot span both $L'$ and $R'$. That is, if $V_{2,L}\cap L' \not = \emptyset$, then $V_{2,L} \cap R' = \emptyset$. Likewise, if $V_{2,L} \cap R' \not = \emptyset$, then $V_{2,L} \cap L' = \emptyset$.  
\end{observation}
\begin{proof}
Suppose  $V_{2,L} \cap R' \not = \emptyset $, and $V_{2,L}\cap L' \not = \emptyset$. There is an edge between $L'$ and $R'$ since $V_{2,L}$ is a clique. Therefore, we have a contradiction since $(L',S',R')$ is a separation triple, but there is an edge between $L'$ and $R'$. %
\end{proof}

\begin{lemma} \label{lem:move-clique}
If there is a separation triple  $(L',S',R')$ in $H_L (\text{or } H_R)$ such that $|S'| < k$ and $V_{2,L} \cap L' \not = \emptyset (\text{or } V_{2,R} \cap L' \not = \emptyset)$, then there is a separation triple $(L'',S'',R'')$ in $H_L (\text{or } H_R)$ where  $L'' = L' \cup V_{2,L}$, $S'' = S' \setminus V_{2,L}$, and $R'' = R'$. In particular, $|S''| \leq |S'|$ and $V_{2,L} (\text{or } V_{2,R}) \subseteq L''$. 
\end{lemma}
\begin{proof}
We prove the result for $H_L$. The proof for thcase $H_R$ is similar. 
If $V_{2,L} \subseteq L'$, then we are done.  Now, suppose otherwise. By \Cref{obs:clqiue-one-side}, we have  $V_{2,L} \subseteq L' \sqcup S'$ %

We claim that $N(V_{2,L}) \subseteq L' \sqcup S'$. First of all,  there is a vertex $z \in V_{2,L} \cap L'$ since $V_{2,L} \subseteq L' \sqcup S', |S'| < k$ but $|V_{2,L}| = k$. Also, by \Cref{def:HLHR},  $N(V_{2,L}) = S$.   Suppose that there is a vertex $z' \in N(V_{2,L})$ such that $z' \in R'$. Since (1) $V_{2,L}$ is a clique that every node has edges to every vertex in $S = N(V_{2,L}) $, (2) $z \in V_{2,L} \cap L'$ and (3) $z' \in  N(V_{2,L}) \cap R'$, there is an edge between $L'$ and $R'$. However, this contradicts to the fact that $(L',S',R')$ is a separation triple where $L'$ and $R'$ cannot have an edge between each other. Therefore, the claim follows.

We construct a new separation triple in $H_L$ as follows. Let $L'' = L' \cup V_{2,L}$, $S'' = S' \setminus V_{2,L}$, and $R'' = R'$. Clearly, $V_{2,L} \subseteq L''$, and $|S'| < k$.  

We claim that the vertex set $(L'',S'',R'')$ forms a separation triple in $H_L$. First of all, it is clear that $L''$, $S''$, and $R''$ form a partition of all vertices in $H_L$ (i.e., they are pairwise disjoint, and $L'' \sqcup S'' \sqcup R'' = V_{H_L}$). It is enough to verify that $L'', S'',$ and $R''$ are not empty, and that $S''$ is a vertex-cut in $H_L$. We first show that each set $L'', S'', R''$ is non-empty. Clearly, $L''$ and $R''$ are not empty since we add new elements to the set $L''$, and  $R'' = R'$. We show that $S''$ is also non-empty.  By \Cref{lem:Sisnotallclique}, there is a vertex in $S'$ that is not in the clique $V_{2,L}$. Also, we only move $V_{2,L}$ from $S'$ to $L'$. Therefore, $S''$ is not empty.  We now show that $S''$ is a vertex-cut in $H_L$. Since $N(V_{2,L}) \subseteq L' \sqcup S'$, and $L' \sqcup S' = L'' \sqcup S ''$, we have $N(V_{2,L}) \subseteq L'' \sqcup S''$. Also, $R''$ is not an emptyset. Hence,  $H_L - S''$ has no path from any vertex in $V_{2,L}$ to any vertex in $R''$. Therefore, $S''$ is a vertex-cut in $H_L$.    %
\end{proof}
 \begin{remark}
 If $V_{2,L} \cap R' \not = \emptyset$ for the separation triple $(L',S',R')$, then we can swap $L'$ and $R'$ so that we can still apply \Cref{lem:move-clique}. 
 \end{remark}
 
\begin{lemma} \label{lem:cut-in-hl-is-a-cut-in-g}
If $H_L$ has a vertex-cut $S'$ corresponding to the separation triple $(L',S',R')$ such that $|S'| < k$  and $V_{2,L} \subseteq L'$, then $S'$ is also a vertex-cut in $G$. 
\end{lemma}
\begin{proof}
Since $|S'| < k$, we only need to show that  $S'$ is a vertex-cut in $G$. Since $V_{2,L} \subseteq L'$, we have $N_{H_L}(V_{2,L}) \subseteq L' \sqcup S'$. This means $H_L - S'$ has no paths from any vertex in $V_{2,L}$ to any vertex in $R'$. Hence, by \Cref{def:HLHR}, $G -S'$ has no  paths from any vertex in $R$ (the set $R$ was contracted into a $V_{2,L}$) to any vertex in $R'$. Therefore, $S'$ is a vertex-cut in $G$. 
\end{proof}

To finish the proof of \Cref{thm:vertex-cut-characterization}, we show that if not all the conditions in \Cref{thm:vertex-cut-characterization} are true, then $G$ has a vertex-cut of size smaller than $k$. 

If $|S| < k$ or $\kappa_G(x,y) < k$ for some $x,y \in S$, then we are done. Suppose now that  $|S| \geq k$ and $\kappa_G(x,y) \geq k$. This implies $H_L$ (or $H_R$) is not $k$-connected. We show that an  optimal vertex-cut in $H_L$ or $H_R$ whose size is smaller than $k$ can be used to construct a vertex-cut in $G$ of size smaller than $k$. 

We assume WLOG that $H_L$ contains a vertex-cut of size smaller than $k$. The other case that  $H_R$ contains a vertex-cut of size smaller than $k$ is similar.

 Let $(L^*,S^*,R^*)$ be an optimal separation triple in $H_L$. Note that the vertex-cut $S^*$ has size $< k$. 
We claim that  $V_{2,L} \subseteq L^*$ or $V_{2,L} \subseteq R^*$.  Suppose that $V_{2,L} \not \subseteq L^*$ and  $V_{2,L} \not \subseteq R^*$. By \Cref{obs:clqiue-one-side}, we have $V_{2,L} \cap S^* \not = \emptyset$.  By \Cref{lem:move-clique}, we obtain a new separation triple $(L^* \cup V_{2,L},  S^* \setminus V_{2,L}, R^*)$. Clearly, $ |S^* \setminus V_{2,L}| < |S^*|$ since $V_{2,L} \cap S^* \not = \emptyset$.  However, this is impossible since $S^*$ is the smallest vertex-cut, a contradiction. 

We now show that $G$ has a vertex-cut of size at most $k$. Since $S^*$ is an optimal vertex-cut,  $H_L$ has a vertex-cut $S^*$ corresponding to the separation triple $(L^*,S^*,R^*)$ such that $|S^*| < k$, and $V_{2,L} \subseteq L^*$  (if $V_{2,L} \subseteq R^*$, we can swap $L^*$ and $R^*$). By \Cref{lem:cut-in-hl-is-a-cut-in-g}, $S^*$ is also a vertex-cut in $G$.

\subsection{Vertex-Expansion and $k$-Connectivity} 

\begin{definition}[Vertex expansion of a separation triple $h(L,S,R)$] \label{def:vexpansion_triple}
Given a separation triple $(L,S,R)$, the \textit{vertex expansion} of $(L,S,R), h(L,S,R),$ is $\frac{|S|}{\min(|L|,|R|) + |S|}$.  
\end{definition}

\begin{definition} [Vertex expansion of a graph $h(G)$]\label{def:vexpansion_graph}
the \textit{vertex expansion} of $G, h(G), $ is  \\ $\min_{(L,S,R) \in G} h(L,S,R)$, i.e., the minimum vertex expansion over all separation triples in $G$. 
\end{definition}

\begin{proposition}\label{pro:smallerside}
For any separation triple $(L,S,R)$, $\minlr \leq n/2$.
\end{proposition}
\begin{proof}
Suppose $\minlr > n/2$. We have $|L|+|R| = \minlr + \max(|L|,|R|) > n/2 + n/2 = n$, which is a contradiction. 
\end{proof}

\begin{proposition} \label{pro:sparse-vcut}
If $h(L,S,R) \leq \eta$, then 
\begin{itemize}[noitemsep, nolistsep]
\item $\min(|L|,|R|) \geq  (1/\eta -1)\kappa_G$.
\item $ |S| \leq n\eta/(2-2\eta)$. 
\end{itemize}
\end{proposition}
\begin{proof}
By definition of vertex-expansion \Cref{def:vexpansion_triple}, we have the following equations:
\begin{align*} 
|S|/(\minlr + |S|) &\leq \eta, \\
|S| &\leq \eta(\minlr + |S|), \\
|S|(1-\eta)  \leq& \eta \minlr. 
\end{align*}

Hence, we get $|S| \leq \minlr \eta/(1-\eta)$, which is at most $ n\eta/(2-2\eta)$ by \Cref{pro:smallerside}. We also get $\minlr \geq (1/\eta - 1)|S|$, which is at least $ (1/\eta-1)\kappa_G$ since $|S| \geq \kappa_G$.
\end{proof}

\begin{proposition} \label{pro:high-vexp-vol}
If $h(G) \geq \eta$, and there is no separation triple $(L,S,R)$ such that $\minlr \leq 2k/\eta$ and $|S| < k$, then $G$ is $k$-connected. 
\end{proposition}
\begin{proof}
Suppose $G$ has a separation triple $(L',S',R')$ such that $|S'| < k$. By the given condition, $\min(|L'|,|R'|) > 2k/\eta$. Therefore, we have  $$ \eta \leq h(L',S',R') = \frac{|S'|}{\min(|L'|,|R'|) + |S'| } < k/\min(|L'|,|R'|) < \eta/2.$$ The first inequality follows from $h(G) \geq \eta$, and \Cref{def:vexpansion_graph}, the second equality follows from \Cref{def:vexpansion_triple}. The third inequality follows from $|S'| < k$. The last inequality follows from $\min(|L'|,|R'|) > 2k/\eta$. Therefore, $\eta < \eta/2$, and we have a contradiction.  
\end{proof}

\begin{corollary} \label{cor:vexp-clean}
For $a \in (0,1)$, if $h(L,S,R) \leq 1/(2n^{1-a-o(1)})$, then 
\begin{itemize}[noitemsep, nolistsep]
\item $\min(|L|,|R|) \geq  n^{1-a-o(1)}$.
\item $ |S| \leq  n^{a+o(1)}/2$. 
\end{itemize}
\end{corollary}
\begin{proof}
The results follows from By \Cref{pro:sparse-vcut} where we use $\eta = 1/(2n^{1-a-o(1)})$. By \Cref{pro:sparse-vcut}, we get  $\min(|L|,|R|) \geq  (1/\eta -1)\kappa_G$, which is $\geq (1/\eta -1) = 2n^{1-a-o(1)}-1  \geq n^{1-a-o(1)}$. By \Cref{pro:sparse-vcut}, we have $|S| \leq n\eta/(2-2\eta) = (n/2) (1/(2n^{1-a-o(1)}-1)) \leq (n/2)(1/n^{1-a-o(1)}) \leq n^{a+o(1)}/2$. 
\end{proof}

\section{Deterministic Vertex Connectivity Algorithm}
\label{sec:vc-algo}

In this section we give our main vertex connectivity algorithm.
Our main result is

\VCAlgoMain*

\subsection{Overview}

Our algorithm is based on two structural lemmas
about $k$-connectivity.
Recall from \Cref{def:vertex-expansion} that the
vertex expansion of a separation triple $(L,S,R)$
is $h(L,S,R) = \frac{|S|}{\min(|L|,|R|)+|S|}$ and the
vertex expansion of a graph $G$ is $h(G) = \min_{(L,S,R)\in G} h(L,S,R)$.

The first observation is the following. Suppose that the vertex expansion
of $G$ is $h(G) \ge \gamma$ for some parameter $\gamma$.
Then, any separation triple $(L,S,R)$
of size less than $k$ must be such that either
$|L|<k \gamma^{-1}$ or $|R|<k \gamma^{-1}$ because
otherwise its vertex expansion is $h(L,S,R)<k/(k \gamma^{-1})\le \gamma$.
Therefore, $G$ is not $k$-connected if and only if there is a
set $L\subset V$ where $|L|\le k \gamma^{-1}$ and $|N(L)|<k$ where $N(L)$ is
the neighbors of $L$. Note that
\[
\vol\left(L\right)
\leq
2\left|E\left(L, L\right) \right|
+
E\left(L,N\left(L\right)\right)
\leq
2 k\left|L\right| + k\left|L\right|
=
O\left(\gamma^{-1}\right),
\]
because $G$ has arboricity at most $k$.

This is exactly where the \emph{local vertex connectivity }(LocalVC)
algorithm introduced in \cite{NanongkaiSY19} can help us. This algorithm
works as follows: given a vertex $x$ in a graph $G$ and parameters
$\nu$ and $k$, either
\begin{enumerate}
	\item certifies that there is no set $S\ni x$
where $\vol(S)\le\nu$ and $|N(S)|<k$, or
	\item returns a set $S\ni x$
where $|N(S)|<k$. See \Cref{def:localvc} for a formal definition.
\end{enumerate}
There are currently two deterministic algorithms for this problem:
an $\oh(\nu^{1.5}k)$-time algorithm by a subset of the 
authors~\cite{NanongkaiSY19} and an
$\oh(\nu k^{k})$-time algorithm via a slight adaptation of the
algorithm by Chechik et al.~\cite{ChechikHILP17}.

For simplicity, we will assume that $k$ is a constant,
and use the $\oh(\nu k^{k})$-time bound here,
which by our assumption we view as $\oh(\nu)$ time.

From the above observation about the set $L$, it is enough to run
the $\LocalVC$ algorithm from every vertex $x$ with a parameter
$\nu=O(n \gamma^{-1})$ to decide if such $L$ exists.
This takes $\oh(n \gamma^{-1})$
total time to decide $k$-connectivity of $G$ with the assumption
$h(G)\ge \gamma^{-1}$.

To remove the assumption, we start by calling our deterministic
vertex expansion algorithm.
As described in \Cref{subsection:Bananas}, on sparse graphs
this routine finds in $\oh(n^{1.5})$ time a separating
triple $(L,S,R)$ such that
\[
h(L,S,R)\le h(G)\cdot n^{o(1)}.
\]
If $h(G) \ge \oh(\gamma^{-1})$, then the above algorithm based on
local vertex connectivity can be immediately invoked.
So it suffices to consider the remaining case where
$h(L,S,R) \leq \gamma)$.

Notice that, in this case, we have:
\begin{align}
\left|L\right|, \left|R\right| & \geq \gamma^{-1}, \text{and}\\
\left|S\right| & \leq n \cdot \gamma
\end{align}
That is, the cut is quite balanced, and we can thus use
divide-and-conquer.

For such a separation $(L,S,R)$ with $h(L,S,R) \leq \gamma$,
we first check if $|S|\ge k$ and whether every pair of 
$x,y\in S$ are $k$-vertex-connected.
By simple augmenting-path based max-flow routines
(such as the Ford-Fulkerson algorithm),
this takes time
\[
|S|^{2}\cdot O\left(mk\right)
\leq
|S|^{2}\cdot O\left(nk^2\right)
\leq \oh\left(n^3 \gamma^2 k^2 \right),
\]
where the first inequality follows from being able to
trim the graph down to the first $k$ spanning trees.

In \Cref{sec:splitvc}, we will also show a faster algorithm
that is useful when $k=\omega(1)$.

If any of these $poly(|S|)$ checks returns a small cut, then we're done.
Otherwise, by the structural property of separating triples
given in \Cref{thm:vertex-cut-characterization}
and proven in \Cref{sec:structure-vc},
it suffices to check recursively if $H_{L}$ and $H_{R}$
(which are $L$ and $R$ with extra vertices attached,
see statement of Theorem~\ref{thm:vertex-cut-characterization})
are both $k$-connected.
In \Cref{sec:vc_time}, we perform a detailed analysis of the
running time of this recursion.
For small values of $k$, and ignoring overheads
coming from the extra edges in $H_L$ and $H_R$, the running time
recurrence that we obtain in terms of $\gamma$ is essentially
\[
T\left( n \right)
=
\max\left\{ n \gamma^{-1},
\max_{n_1 + n_2 = n,
	\left|n_1\right|, \left|n_2\right| \geq \gamma^{-1}}
T\left( n_1 \right) + T\left( n_2 \right)
+ \oh\left( n^3 \gamma^2\right)
+ \oh\left( n^{1.5} \right) \right\}.
\]
The first term is maximized at the topmost level,
so can be considered separately against the total
cost of a recursion that always takes the second case.
The depth of such a recursion is more or less bounded by
the reduction in $n$ at each step, which is $\gamma^{-1}$.
So the total layers of recursion is $\oh(n \gamma)$, and as
the total size of each level of recursion is $n$,
the total work can be bounded by
\[
T\left( n \right)
\leq
\oh\left(
n^4 \gamma^3 + n^{2.5} \gamma 
+ n \gamma^{-1}
\right).
\]
This is minimized at $\gamma = n^{-0.75}$,
for a total of $\oh(n^{1.75})$.

In the above back-of-the-envelope calculation,
we treated $k$ as a constant for simplicity.
For larger values of $k$,
specifically $k=\omega(\log n/\log\log n)$,
we instead use the $\ot(\nu^{1.5} k)$-time LocalVC algorithm by \cite{NanongkaiSY19} as the running time of $O(\nu k^k)$ is too slow.
In this situation, with our framework, a $O(m^{5/3 -\epsilon})$-time
low-vertex-expansion algorithm is needed to break quadratic time for vertex-connectivity.

\subsection{Algorithm}

We now formalize this vertex connectivity algorithm
that we outlined above.
First, we formalize the local vertex connectivity
routine that searches for a small cut starting from a single vertex.
\begin{definition} [LocalVC] \label{def:localvc}
	$\LocalVC(G,x,\nu,k)$ is any algorithm that takes as input a pointer to any vertex $x \in V$ in an adjacency list representing a connected graph $G = (V,E)$, positive integers $\nu, k$ such that 
	\begin{align} \label{eq:local-conditions}
	\nu k \leq  c_1m  \quad \nu + k \leq  c_2n \quad  \mbox{ and } \quad \min_{v \in V} \deg(v) \geq k  
	\end{align}
	for some positive constant $c_1, c_2$ and outputs either a vertex-cut $S$ corresponding to a separation triple $(L,S,R)$ such that 
	$$ x \in L, \quad \vol(L) \leq O(\nu k), \mbox{ and } \quad |S| \leq k,$$or the symbol $\perp$ certifying that there is no separation triple $(L,S,R)$ such that 
	$$ x \in L, \quad \vol(L) \leq \nu, \mbox{ and } \quad |S| \leq k.$$
\end{definition}

\begin{theorem}[\cite{NanongkaiSY19,ChechikHILP17}] \label{thm:localvc}
	There is a deterministic $\LocalVC$ algorithm that runs in $\ot(\min(\nu^{3/2}k, \nu k^k) )$ time. 
\end{theorem} 

\begin{definition} [SplitVC] \label{def:splitvc}
	$\SplitVC(G,S,k)$ is any algorithm that takes as input a connected graph $G = (V,E)$, a vertex-cut $S$, and positive integer $k$  such that $|S| \geq k$, and decides if there exists a pair $x \in S$, and $y \in S$ such that $\kappa_G(x,y) < k$. If so, it returns an $(x,y)$-vertex-cut of size less than $k$. Otherwise, it returns $\perp$. 
\end{definition}

It is easy to see that we can implement SplitVC by using at most $|S|^2$ calls to max-flow.
We show in Appendix~\ref{sec:splitvc}
that it is possible to implement deterministic SplitVC with running time $O((|S|+k^2)mk)$.

Here, we denote $n$ as the original input size, and treat $n$ as a global variable. We denote $|V|$ as the size of the current input. Let $\Lambda = \max( (8/c_1) k^2)^{1/a}, (9/c_2)k^2)^{1/a}, n^{\epsilon} ) = \Theta( \max(k^{2/a}, n^{\epsilon})) $ where $c_1$ and $c_2$ are the constants in \Cref{eq:local-conditions}, and $\epsilon > 0$ is sufficiently small constant. 

\begin{algorithm}[H]
\caption{MainVC$(G,k,a,n)$}
Input: Graph $G = (V,E)$, integer $k > 0$, real $a \in (0,1/2)$, integer $n > 0$. \\
Assumptions: $G$ has aboricity $k$.\\ %
Output: A vertex-cut of size $< k$ or the symbol $\perp$ certifying
that $\kappa_G \geq k$. 
\begin{algorithmic}[1] \label{alg:main-vc}
\If {$|V| \leq \baseCaseMainVC $} 
\State compute $\kappa_G$ using any deterministic algorithm.  
\State \Return answer based on $\kappa_G$. 
\EndIf 
\State Let $\eta = 1/(2|V|^{1-a})$.
\If {$h(G) \geq \eta$} 
\If{$\min_{v \in V} \deg(v) < k $ }
\State \Return $ N( u_{\min})$ where $u_{\min}$ is the vertex with minimum degree. 
\EndIf
\State Let $\nu \gets 6k^2/\eta$. \label{line:set-nu}
\For{each $x \in V$} 
 \If{LocalVC$(G, x , \nu, k-1)$  returns a vertex-cut} \Comment{\Cref{def:localvc}}
     \State \Return the corresponding vertex-cut in $G$. \label{line:local-vc-find-vcut} 
\EndIf
\EndFor 
\State \Return $\perp$. \label{line:g-is-k-connected-high-vexp}
\EndIf
\State Let $(L,S,R)$ be a separation triple such that $h(L,S,R) \leq \eta |V|^{o(1)}$. \Comment{\Cref{thm:vertex expansion}} 
\If {$|S| < k$, or $\SplitVC(G,S,k)$ returns a vertex-cut} \Comment{\Cref{def:splitvc}}
\State \Return the corresponding vertex-cut in $G$. \label{line:sparse-cut-find-vcut}
\EndIf
\State Let $H_L$ and $H_R$ be the left and right subgraph from $G$, respectively.  \Comment{\Cref{def:HLHR}}
\State Let $\tilde H_L$ and $\tilde H_R$ be the sparsified graph from
$H_L$ and $H_R$ respectively.  \Comment{\Cref{thm:sparsification}}
\If{MainVC$(\tilde H_L, k, a, n)$ or MainVC$(\tilde
  H_R, k, a, n)$ returns a vertex-cut} 
\State \Return the corresponding vertex-cut in $G$. \label{line:recursion-find-vcut} \Comment{ \Cref{lem:move-clique}.}
\EndIf
\State \Return $\perp$. \label{line:g-is-kconnected-low-vexp}
\end{algorithmic}
\end{algorithm}

\subsection{Correctness}
\BL \label{lem:main-vc-correct}
\Cref{alg:main-vc} returns either a vertex-cut of size $< k$ or $\perp$
certifying that $G$ is $k$-connected. 
\EL

 We use induction on number of vertices. We prove that given a connected graph $G$ with $n$ vertices,  \Cref{alg:main-vc} correctly returns a vertex-cut of size $<k$ or $\perp$. For the
base case, if $G$ has $ \leq \baseCaseMainVC $ vertices, we run any
deterministic vertex-connectivity algorithm to decide if $\kappa_G <
k$.   For the inductive hypothesis, we assume that \Cref{alg:main-vc}
outputs correctly for any connected graph with at most $r$
vertices where 
\begin{align} \label{eq:r-is-big} 
r \geq \baseCaseMainVC = \max( (8/c_1) k^2)^{1/a}, (9/c_2)k^2)^{1/a}, n^{\epsilon} ). %
\end{align}  We show as an inductive step that \Cref{alg:main-vc} outputs correctly for the graph with at most $r+1$ vertices.

We fix an arbitrary connected graph $G = (V,E)$ with $r+1$
vertices.    By \Cref{thm:sparsification}, we can also assume that $G$ has aboricity $k$. Therefore, the precondition for \Cref{alg:main-vc} is satisfied.  %

We first verify that preconditions in \Cref{eq:local-conditions} for LocalVC are satisfied. If $G$ has minimum degree less than $k$, we can  output the neighbors of the vertex of minimum degree. Otherwies, $G$ has minimum degree at least $k$. It remains to verify $\nu + k \leq  c_2|V|$, and $\nu k \leq  c_1|E|$. 

\begin{claim}\label{claim:localvc-good}
Let $\nu$ and $k$ be the numbers as defined in \Cref{alg:main-vc} (line~\ref{line:set-nu}). Denote $n' = r+1$, and $m'$ as number of edges in $G$. We have $\nu + k \leq c_2n'$ and $\nu k \leq c_1m'$. %
\end{claim}
\begin{proof}
Since $G$ has aboricity $k$, $m' \leq  n'k$. Therefore, it is enough to show that  $\nu k < c_1 n'k $ and $\nu + k < c_2n'$. By \Cref{eq:r-is-big}, we have $n' \geq \baseCaseMainVC  \geq \max( (8/c_1) k^2)^{1/a}, (9/c_2)k^2)^{1/a} )$.  By simple algebra using $\nu = 8k^2n'^{1-a}$ ,  $n' \geq  ((8/c_1) k^2)^{1/a} $ implies $\nu k \leq c_1n'k$, and  $n' \geq ((9/c_2)k^2)^{1/a}$ implies $\nu + k \leq c_2n'$. 
\end{proof}

\begin{lemma} \label{lem:main-vc-correct-high-vexp}
If $h(G) \geq \eta$, then \Cref{alg:main-vc} correctly outputs a
vertex-cut of size at most $k$ (line~\ref{line:local-vc-find-vcut}) or the symbol $\perp$(line~\ref{line:g-is-k-connected-high-vexp}).
\end{lemma}
\begin{proof}
Suppose $G$ has  a separation triple $(L,S,R)$ such that 
\begin{align} \label{eq:sep-smallvol}
\minlr \leq 2k/\eta, \quad  \text{ and }   |S| < k.
\end{align}
We claim that \Cref{alg:main-vc} returns a vertex-cut of size at most $k$ (line~\ref{line:local-vc-find-vcut}).  We show that $\vol(L) \leq 6k^2/\eta$. 
 Without loss of generality,  we assume $|L| \leq |R|$. By \Cref{eq:sep-smallvol}, $|L| \leq 2k/\eta$. Since $G$ has aboricity $k$, and $|L| \leq 2k/\eta$, we have $\vol(L) \leq 2|E(L,L)| + |E(L,S)| \leq 2k|L| + k|L| = 6k^2/\eta$. Also, \Cref{alg:main-vc} (line~\ref{line:local-vc-find-vcut}) runs LocalVC on every seed vertex. So far, we have that there is $x \in L$, and $\vol(L) \leq 6k^2/\eta \leq \nu$, and $|S| \leq k - 1$. Also, by \Cref{claim:localvc-good}, the preconditions for \Cref{thm:localvc} are satisfied. Therefore, by \Cref{thm:localvc}, LocalVC at node $x$ outputs a vertex-cut of size at most $k-1$.

Suppose now that $G$ has no separation triple $(L,S,R)$ satisfying
\Cref{eq:sep-smallvol}. Recall that $h(G) \geq \eta$. Therefore,
by \Cref{pro:high-vexp-vol}, $G$ is $k$-connected. In this case, by
\Cref{thm:localvc}, LocalVC (line~\ref{line:local-vc-find-vcut})  always outputs $\perp$, and
\Cref{alg:main-vc} correctly returns the symbol $\perp$ (line~\ref{line:g-is-k-connected-high-vexp}). 
\end{proof}
 
Therefore, if $h(G) \geq \eta$, then
\Cref{lem:main-vc-correct-high-vexp} says that \Cref{alg:main-vc}
outputs correctly. Now, suppose that $h(G) \leq \eta (r+1)^{o(1)}$. We show that
\Cref{alg:main-vc} outputs correctly. 

\begin{lemma} \label{lem:main-vc-correct-low-vexp}
If $h(G) \leq \eta (r+1)^{o(1)} $, then  \Cref{alg:main-vc} correctly outputs a
vertex-cut of size at most $k$ (line~\ref{line:sparse-cut-find-vcut} or  line~\ref{line:recursion-find-vcut})  or the symbol $\perp$ (line~\ref{line:g-is-kconnected-low-vexp}). 
\end{lemma}
\begin{proof}
Since $h(G) \leq \eta (r+1)^{o(1)}$, there is a 
separation triple $(L,S,R)$ such that $h(L,S,R) \leq \eta  (r+1)^{o(1)}$. If $|S| < k$, or we can find a pair of vertices $x,y$ in $S$ such that $\kappa(x,y) < k$, then \Cref{alg:main-vc} (line~\ref{line:sparse-cut-find-vcut}) outputs the corresponding vertex-cut of size at most $k$, and we are done. Now, we assume $|S| \geq k$ and $\kappa(x,y) \geq k$ for all $x,y \in S$. 

Since we set $\eta = 1/(2(r+1)^{1-a})$,  $h(L,S,R) \leq  1/(2(r+1)^{1-a-o(1)})$. By \Cref{cor:vexp-clean}, we have 
\begin{align}  \label{eq:sparse-vcut}
\min(|L|,|R|) \geq (r+1)^{1-a-o(1)}. %
\end{align}   
Let $H_L$ and $H_R$ be the left and right subgraphs as in
\Cref{def:HLHR}.  We claim that the number of vertices of $H_L$ and $H_R$ are strictly smaller than $n$.  We focus on $H_L$ because the case $H_R$ is similar.  Suppose otherwise that number of vertices from $G$ to $H_L$ does not decrease. This means  $k \geq |R|$ by \Cref{def:HLHR}. By \Cref{eq:sparse-vcut}, $|R| \geq \minlr \geq  (r+1)^{1-a-o(1)}$. Therefore, $k \geq (r+1)^{1-a-o(1)}$, so $r < k^{1/(1-a-o(1))}$,  contradicting to \Cref{eq:r-is-big}.

By \Cref{thm:sparsification}, we obtain $\tilde H_L$ and $\tilde H_R$ where number of vertices does not change from that of $H_L$ and $H_R$, which means the number of vertices are less than $n$. Also, both $\tilde H_L$ and $\tilde H_R$ have aboricity $k$. Furthremore,  any vertex-cut of $\tilde H_L$ ($\tilde H_R$) with cardinality $< k$ is a vertex-cut in $H_L$ ($\tilde H_R$).  We now prove the inductive step. 

Suppose $G$ is not $k$-connected. We show that \Cref{alg:main-vc}
(line~\ref{line:recursion-find-vcut}) returns a vertex-cut in $G$. By
\Cref{thm:vertex-cut-characterization}, $H_L$ or $H_R$ is not
$k$-connected. By \Cref{thm:sparsification}, $\tilde H_L$ or $\tilde
H_R$ is not $k$-connected. Since $\tilde H_L$ and $\tilde H_R$ have
less than $n$ vertices, and they have aboricity $k$,
\Cref{alg:main-vc} returns a vertex-cut for $\tilde H_L$ (or $\tilde
H_R$) by inductive hypothesis. Also, any vertex-cut in $\tilde H_L$
(or $\tilde H_R$) is a vertex-cut in $H_L$ (or $H_R$). By
\Cref{lem:move-clique,lem:cut-in-hl-is-a-cut-in-g}, we can construct the corresponding vertex-cut in $G$ . %
Therefore,  \Cref{alg:main-vc} (line~\ref{line:recursion-find-vcut}) finds a vertex-cut in $G$. 

Suppose now that $G$ is $k$-connected. We show that \Cref{alg:main-vc}
(line~\ref{line:g-is-kconnected-low-vexp}) returns the symbol
$\perp$. By \Cref{thm:vertex-cut-characterization}, $H_L$ or $H_R$ is 
$k$-connected. By \Cref{thm:sparsification}, $\tilde H_L$ or $\tilde H_R$  is $k$-connected.  Since $\tilde H_L$ and $\tilde H_R$ have fewer than $n$ vertices, and they have aboricity $k$, \Cref{alg:main-vc} returns  the symbol $\perp$ by inductive hypothesis.
Therefore,  \Cref{alg:main-vc} (line~\ref{line:g-is-kconnected-low-vexp}) correctly return the symbol
$\perp$. %
\end{proof}

Therefore, by
\Cref{lem:main-vc-correct-high-vexp,lem:main-vc-correct-low-vexp}, we
complete the proof of the
inductive step that  \Cref{alg:main-vc} is correct for $G$ of at most $r+1$ vertices. 
Therefore, \Cref{lem:main-vc-correct} is proved. 

\subsection{Running Time} 
\label{sec:vc_time}

\begin{definition} \label{def:generic-runtime} We define the following running times for subroutines in \Cref{alg:main-vc}. 
\begin{itemize}[nolistsep,noitemsep]
\item An algorithm in the base-case runs in $T_{\textbase}(m,n,k)$ time. 
\item Approximate vertex-expansion $h(G)$ runs in $T_{\texth}(m,n)$ time. 
\item $\LocalVC(G, x, \nu, k-1)$ runs in $T_{\textlocal}(\nu, k)$ time. 
\item $\SplitVC(G,S,k)$ runs in $T_{\textsplit}(m, n, k, |S|)$ time.
\end{itemize}
\end{definition}

\begin{definition} \label{def:use-time}
We define $t_{\textcost}, t_{\textlocal}, $ and $t_{\textbase}$ as follows. 
\begin{itemize}[nolistsep,noitemsep]
\item $t_{\textcost} = T_{\texth}(nk,n) + T_{\textsplit}(nk, n, k, n^{a+o(1)}) + O(nk)$. 
\item $t_{\textlocal} = T(k^2n^{1-a},k)$.
 \item $t_{\textbase} = T_{\textbase}(\baseCaseMainVC k , \baseCaseMainVC, k )$.
\end{itemize}
\end{definition}

\begin{lemma} \label{lem:main-vc-runtime}
 \Cref{alg:main-vc} runs in time $\ot( t_{\textcost} n^{a+o(1)} + n (t_{\textlocal} + t_{\textbase}) )) $. 
\end{lemma}

We derive the running time of the \Cref{alg:main-vc}  by providing an upper bound in terms of recurrence relation as in \Cref{sec:recurrence}, and solving the recurrence relation in \Cref{sec:solve-recurrence}. We prove \Cref{lem:main-vc-runtime} in \Cref{sec:proof-main-vc-runtime}. 

Throughout this section, we denote $a' = a + o(1)$. 

\subsubsection{Recurrence Relation}\label{sec:recurrence}

\begin{lemma}
Suppose \Cref{alg:main-vc} never encounters the case $h(G) \geq \eta$, then the running time satisfies the following  recurrence relation. 
\begin{align} \label{eq:recurrence_main_vc}
T_{k,a}(n) \leq  T_{k,a}(\ell+s+k) + T_{k,a}(n-\ell+k) + T_{\texth}(nk,n) + T_{\textsplit}(nk, n, k, s) + O(nk), 
\end{align}
where $\ell,s, k$ satisfy 
\begin{align} \label{eq:condition_recurrence}
k < s \leq n^{a'}/2, \quad   \ell \geq n^{1-{a'}} , \quad \ell \leq n/2, \quad \mbox{ and } \quad k < n^{1-{a'}},
\end{align}
and the base case is $T_{k,a}(n) = T_{\textbase}(nk,n,k)$ for $n \leq \baseCaseMainVC$.   
\end{lemma}
\begin{proof}
By assumption the input graph $G = (V,E)$ has aboricity $k$, meaning that $m = nk$. If $n \leq \baseCaseMainVC$, then we run any deterministic vertex-connectivity algorithm in  $T_{\textbase}(nk,n,k)$ time. Now suppose $n > \baseCaseMainVC$. \Cref{alg:main-vc} first computes the vertex expansion $h(G)$ of the graph $G$ in   $T_{\texth}(nk,n)$ time. Next, we obtain a separation triple $(L,S,R)$ such that $h(L,S,R) \leq 1/(2n^{1-a'})$. By \Cref{cor:vexp-clean}, we have 
\begin{align} \label{eq:vexp-size-clean}
\min(|L|,|R|) \geq  n^{1-a'} \quad \mbox{ and } \quad |S| \leq  n^{a'}/2. 
\end{align}
Next, we runs $\SplitVC$ algorithm to check if there is a pair of vertices  $x \in S, y \in S$ such that $\kappa_G(x,y) < k$.  This takes $T_{\textsplit}(nk,n,k,s)$ time where $s = |S|$.  Without loss of generality, we assume that there is no such pair and so the algorithm continues. In this situation, we have $k < |S| = s$. Furthermore, by \Cref{eq:vexp-size-clean}, $s = |S| \leq n^{a'}/2$. This justifies the first inequalities in \Cref{eq:condition_recurrence}. Next, we construct left and right subgraphs and sparsify them in $O(m) = O(nk)$ time. We assume WLOG that $|L| \leq |R|$ (otherwise, we can swap $L$ and $R$ in the separation triple). Let $\ell = |L|$. By \Cref{eq:vexp-size-clean}, $\ell = |L| \geq n^{1-a'}$, so we get the second inequality in \Cref{eq:condition_recurrence}. The third inequality in \Cref{eq:condition_recurrence} follows from \Cref{cor:vexp-clean} that $\ell = \minlr \leq n/2$. The final inequality in \Cref{eq:condition_recurrence} follows from $k < n^{a'}/2$, and $a < 0.5$.    By \Cref{def:HLHR}, we have that $H_L$ has $\ell +s + k$ vertices, and $H_R$ has $n-\ell+k$ vertices. Also, the number of new edges is $O(|S|k+k^2) = O(nk)$ for $H_L$ and $H_R$. We apply \Cref{thm:sparsification} for both $H_L$ and $H_R$, which takes additional $O(nk)$ time. Therefore, the running time for solving two subproblems is additional $T_{k,a}(\ell+s+k) + T_{k,a}(n-\ell+k)$.  
\end{proof}

\subsubsection{Solving Recurrence Relation} \label{sec:solve-recurrence}

This section is devoted to solve recurrence relation $T_{k,a}(n)$ in \Cref{eq:recurrence_main_vc}.  The main result is the following lemma.

\begin{lemma}\label{lem:formula_tn}
 An explicit function in \Cref{eq:recurrence_main_vc} is $T_{k,a}(n) = \ot(t_{\textcost} n^{a+o(1)}+  n  t_{\textbase}).$ 
\end{lemma}

Throughout this section, we denote $t_{\textcost} = \ot(n^{x})$ for some $x \geq 1$. 

\begin{definition}[Recursion tree]
A \textit{recursion tree} for a recurrence relation is a tree that is generated by tracing the function calls recursively. Each node $v$ in the tree contains (1) size$(v)$, which is the input to the function, and (2) cost$(v)$, which is the cost at current node excluding the cost for recursions. %
\end{definition}

Let $\mathcal{T}$ be a
recursion tree for the recurrence relation in \Cref{eq:recurrence_main_vc} where each node of size $n_i$ in the tree has cost $O(n_i^x)$. Without loss of generality, we assume the recursion tree $\mathcal{T}$  always has left subproblem of size $\ell+s+k$, and right subproblem of size $n-\ell+k$ (otherwise, we can swap left and right subtrees without affecting total cost).

We give the intuition for solving the recurrence relation \Cref{eq:recurrence_main_vc} using recursion tree. We consider the right child as the subproblem with ``true'' size, and left child as the subproblem with extraneous nodes in the graph. Hence, the total cost on along the right spine from the root counts the cost without extra nodes. For the extra nodes in the graph when recurse on the left, we can essentially charge the cost on the ``true'' nodes. The number of extra nodes is sufficiently small, and  left-branching can happen at most $O(\log n)$ time.   Therefore, the total extra cost can be bounded. We now make the intuition precise.  %

\begin{definition} 
We call the \textit{right-spine} $R$ of the tree $\mathcal{T}$ to be
the set of nodes from the root using right branch all the way to the node before leaf in $\mathcal{T}$.  We denote $v_i \in R$ as the node with path length $i$ from the root to $v_i$. By convention, $v_0$ is the root.%
\end{definition}

\begin{definition} \label{def:recurrence_cn}
Let $C_{k,a}(n)$ be a function satisfying the following recurrence relation 
\begin{align} \label{eq:recurrence_main_vc_upperbound}
C_{k,a}(n) \leq  \sum_{v_i \in R} \text{cost}(v_i) + \sum_{i\colon v_i \in R}C_{k,a}(n_i),  \quad n_i = \ell_i + s_i +k,
\end{align}
where the parameters satisfy 
\begin{align} \label{eq:condition_recurrence_upperbound}
\quad \sum_{v_i \in R}\ell_i \leq n, \mbox{ and } \quad \ell_i+s_i+k \leq 3n/4 \quad \mbox{ for all } v_i \in R, %
\end{align}
and, 
\begin{align} \label{eq:condition_recurrence_upperbound2}
   k < s_i \leq n_i^{a'}/2 \quad \mbox{ and } \quad \ell_i \geq n_i^{1-a'} .
\end{align}
The base case is $C_{k,a}(n) = T_{\textbase}(nk,n,k)$ for $n \leq \baseCaseMainVC$.   
\end{definition}
\begin{remark}
In essence, the recurion tree for $C_{k,a}(n)$ can be obtained by contracting right spine of  the recursion tree from $T_{k,a}(n)$ into a single node. The recusion continues for each left subtree of each node in the right spine $R$. 
\end{remark}
\begin{lemma} \label{lem:teqc}
$T_{k,a}(n) = O(C_{k,a}(n))$.
\end{lemma}
\begin{proof}
 If $n \leq \baseCaseMainVC$, then both functions coincide by definition. We now focus on  $n > \baseCaseMainVC$. We show that the function $C(n)$ can be obtained by rearranging the summation of the cost of all nodes in the recusion tree $\mathcal{T}$. The term $ \sum_{v_i \in R} \text{cost}(v_i)$ corresponds to the summation over all cost of nodes in the right spine $R$. Now, for each node $v_i \in R$, let $\mathcal{T}_i$ be the corresponding left-subtree of $v_i$ in the recursion tree $\mathcal{T}$.  The size of $v_i$ is $\ell_i + s_i + k$ where $\ell_i$ and $s_i$ correspond to the terms $\ell$ and $s$ in \Cref{eq:condition_recurrence} for the left child of any node in the recursion tree $\mathcal{T}$. Therefore,  the total cost is  $\sum_{v_i \in R} \text{cost}(v_i)$ plus the cost of each remaining subtree $\mathcal{T}_i$, which we can compute recursively.  Note that by \Cref{eq:condition_recurrence}, we have $  k < s_i \leq n_i^{a'}/2 $   and  $ \ell_i \geq n_i^{1-a'} .$   %

It remains to show that $\sum_{v_i\in R} \ell_i \leq n$ and that $\ell_i+s_i+k \leq 3n/4$ for all $v_i \in R$.  We first show that  $\sum_{v_i\in R} \ell_i \leq n$. The size of the leaf node in the right spine $R$ is $n - \sum_{v_i \in R}\ell_i$, which is $\geq 0$. Therefore, $n \geq \sum_{v_i \in R}\ell_i$.  Next, we show that $\ell_i+s_i+k \leq 3n/4$ for all $v_i \in R$. Let $n_i$ be the size at node $v_i$. Since recursion does not increase the size of node, we have $n_i \leq n$. By \Cref{eq:condition_recurrence}, we have $\ell_i \leq n_i/2, s_i \leq n_i^{a'}/2, k < n_i^{a'}/2$ where $\ell_i, s_i$ corresponds to the parameters at node $v_i \in R$.   Therefore, $\ell_i + s_i + k  \leq n_i/2 + n_i^{a'}/2 + n_i^{a'}/2 \leq n/2 + n^{'a}/2 + n^{a'}/2  \leq 3n/4$ for any $n \geq 8$ and any $a \in (0, 1/2)$.    %
\end{proof}

\begin{claim} \label{claim:cost-right-spine}
The cost of all nodes in the right spine $R$, $\sum_{v_i \in R} \text{cost}(v_i)$, is $O(n^{x+a'}+  t_{\textbase}))$.  
\end{claim}
\begin{proof}
Each node in $R$ has size at most $n$, which means that the cost is at
most $O(n^x)$ per node. We show that number of nodes in the right-spine $R$ is $O(n^{a'})$, and this implies $O(n^{x+a'})$ term in the total cost. By design, $\mathcal{T}$ always has right subproblem of size $n - \ell +k \leq n - n^{1-a'} +k \leq n - n^{1-a'} + n^{a'}/2 \leq n - n^{1-a'}/2$. The second inequality follows from \Cref{eq:condition_recurrence} where $k < s \leq n^{a'}/2$, so $k < n^{a'}/2$. Therefore, the function $L(n) \leq L(\lceil  n- n^{1-a'}/2 \rceil) + 1, L(1) = 1$ is an upperbound of the
number of nodes in the right-spine $R$. It is easy to see that $L(n) = O(n^{a'})$. Finally, the term $  t_{\textbase} $ follows from the base case of the recurrence where $n \leq \baseCaseMainVC$ Therefore, the claim follows. 
\end{proof}

We now solve the function $C_{k,a}(n)$ for \Cref{eq:recurrence_main_vc_upperbound}. Let $\mathcal{T}'$ be the recursion tree for $C_{k,a}(n)$. 
Let $\text{child}(v)$ be the set of children of node $v$ in $\mathcal{T}'$. Let $\text{level}(i)$ be the set of nodes with distance $i$ from root in $\mathcal{T}'$.  For each node $v \in \mathcal{T}'$ except the root, we denote the size of $v$ as size$(v) = n_v = \ell_v + s_v + k$ according to \Cref{eq:recurrence_main_vc_upperbound}. 

We make useful observation about the recursion tree $\mathcal{T}'$.
\begin{observation} \label{obs:useful-recursion-cn} For any non-root  internal node $v$ in the recursion tree $\mathcal{T}'$, 
\begin{align}
k &< s_v,\label{eq:k-less-than-s} \\
 (1+2s_v/\ell_v)   &\leq   (1+2/n^{\epsilon(1-2a')}), \label{eq:factor-base-case} \\
\sum_{u \in \text{child}(v)}\ell_{u}& \leq  n_v, \label{eq:sum-ell-less-than-n} \\
s_v & < \ell_v. \label{eq:s-less-than-ell} 
\end{align}

\end{observation}

\begin{proof}
The results follow from \Cref{def:recurrence_cn} and \Cref{eq:condition_recurrence_upperbound}.  We now show that  $(1+2/n^{\epsilon(1-2a')}) \geq (1+2s_v/\ell_v)$ for any internal node $v$. It is enough to show that $\ell_v/s_v \geq n^{\epsilon(1-2a')}$.   Since $v$ is an internal node, $n_v \geq \baseCaseMainVC \geq n^{\epsilon}$. By \Cref{eq:condition_recurrence_upperbound2}, $\ell_v \geq n_v^{1-a'}$ and $s_v \leq n_v^{a'}$. Hence, $\ell_v/s_v \geq n_v^{1-2a'} \geq n^{\epsilon(1-2a')}$.  Finally, $s_v < \ell_v$ since $a \in (0, 1/2)$, and $s_v \leq n^{a'}_v/2$, and $\ell_v \geq n^{1-a'}$  by \Cref{eq:condition_recurrence_upperbound2}.
\end{proof}

\begin{claim} \label{claim:internal-node}
For each level in the recursion tree $\mathcal{T}'$,  the total size of internal nodes is most $2n$. 
\end{claim}
\begin{proof}

First, we show that the recursion tree $\mathcal{T}'$ has depth at most $c\ln n$ for some constant $c$.  This follows from \Cref{eq:condition_recurrence_upperbound} where each subproblem size is at most 3/4 factor of the current size.

Let $\text{child}^*(v)$ be the set of non-leaf children of node $v$ in $\mathcal{T}'$. Let $\text{level}^*(i)$ be the set of non-leaf nodes with distance $i$ from root in $\mathcal{T}'$.  

We claim that the total size of internal nodes at level $i$, $\sum_{u \in \text{level}^*(i)} n_u$, is at most $n(1+2/n^{\epsilon (1-2a')})^i$.  We prove the claim by induction on number of level. Base case is at level $i = 1$. We have 

\begin{align*}
\sum_{u \in \text{level}^*(1)} n_u &= \sum_{v \in \text{child}^*(\text{root})} n_v \\
& = \sum_{v \in \text{child}^*(\text{root})}\ell_v+s_v+k \\
&\stackrel{(\ref{eq:k-less-than-s})} \leq   \sum_{v \in \text{child}^*(\text{root})}\ell_v(1+2s_v/\ell_v) \\
&\stackrel{(\ref{eq:factor-base-case})}= (1+2/n^{\epsilon(1-2a')})\sum_{v \in \text{child}^*(\text{root})}\ell_v\\
&\stackrel{(\ref{eq:sum-ell-less-than-n})} \leq (1+2/n^{\epsilon(1-2a')})n.
\end{align*}

 For inductive hypothesis, we assume that 
\begin{align} \label{eq:inductive-hypothesis-internal-size}
\sum_{u \in \text{level}^*(i)}n_u \leq n(1+2/n^{\epsilon (1-2a')})^i, \text{ for } i \geq 1 
\end{align} 
 We now prove as inductive step that  $\sum_{u \in \text{level}^*(i+1)} n_u \leq n(1+2/n^{\epsilon (1-2a')})^{i+1}$ (as convention, we define the sum over an empty set as zero) . This follows from \Cref{obs:useful-recursion-cn}, and the followings. 
\begin{align*}
\sum_{u \in \text{level}^*(i+1)}n_u&= \sum_{u \in \text{level}^*(i)} \sum_{v \in \text{child}^*(u)}(\ell_v+s_v+k)\\
&\stackrel{(\ref{eq:k-less-than-s})} \leq  \sum_{u \in \text{level}^*(i)} \sum_{v \in \text{child}^*(u)}(\ell_v+2s_v)\\
&=  \sum_{u \in \text{level}^*(i)} \sum_{v \in \text{child}^*(u)}\ell_v(1+2s_v/\ell_v)\\
&\stackrel{(\ref{eq:factor-base-case})} \leq  (1+2/n^{\epsilon (1-2a')}) \sum_{u \in \text{level}^*(i)} \sum_{v \in \text{child}^*(u)}\ell_v\\
&\stackrel{(\ref{eq:sum-ell-less-than-n})} \leq  (1+2/n^{\epsilon(1-2a')}) \sum_{u \in \text{level}^*(i)}n_u\\
&\stackrel{(\ref{eq:inductive-hypothesis-internal-size})} \leq  (1+2/n^{\epsilon(1-2a')}) n(1+2/n^{\epsilon(1-2a')})^i \\
& =  n(1+2/n^{\epsilon(1-2a')})^{i+1}
\end{align*} 
Therefore,  $\sum_{u \in \text{level}^*(i)} n_u \leq n(1+2/n^{\epsilon (1-2a')})^{i}$, which is $\leq 2n$ for sufficiently large $n$ and $i \leq c\log n$. 

\end{proof}

\begin{lemma} \label{lem:totalsizen-each-level}
For each level $i$ in the recursion tree $\mathcal{T}'$, the total size $\sum_{u \in \text{level}(i)} n_u$ is $O(n)$. 
\end{lemma}
\begin{proof}
Let $\text{child}^*(v)$ be the set of non-leaf children of node $v$ in $\mathcal{T}'$. Also, let $\text{child}^{\dagger}(v)$ be the set of leaf children of node $v$ in $\mathcal{T}'$. Note that  $\text{child}(v) = \text{child}^*(v) \cup \text{child}^{\dagger}(v)$.  Let $\text{level}^*(i)$ be the set of internal nodes at level $i$ in $\mathcal{T}'$. Also,  Let $\text{level}^{\dagger}(i)$ be the set of leaf nodes at level $i$ in $\mathcal{T}'$.  Note that  $\text{level}(i) = \text{level}^*(i) \cup \text{level}^{\dagger}(i)$. 

We first show that 
\begin{align} \label{eq:sum-dagger-6n}
\sum_{u \in \text{level}^{\dagger}(i+1)}n_u \leq 6n
\end{align}  
by the followings.
\begin{align*}
\sum_{u \in \text{level}^{\dagger}(i+1)}n_u  &= \sum_{u \in \text{level}^*(i)} \sum_{v \in \text{child}^*(u)}n_u \\
&= \sum_{u \in \text{level}^*(i)} \sum_{v \in \text{child}^*(u)}(\ell_v+s_v+k)\\
&\stackrel{(\ref{eq:k-less-than-s})} \leq \sum_{u \in \text{level}^*(i)} \sum_{v \in \text{child}^*(u)}(\ell_v+2s_v)\\
& \stackrel{(\ref{eq:s-less-than-ell})} \leq  \sum_{u \in \text{level}^*(i)} \sum_{v \in \text{child}^*(u)} 3\ell_v\\
&\leq \sum_{u \in \text{level}^*(i)} \sum_{v \in \text{child}(u)} 3\ell_v\\
& \stackrel{(\ref{eq:sum-ell-less-than-n})}  \leq 3 \sum_{u \in \text{level}^*(i)} n_u\\
&\leq 3 (2n) = 6n & \text{by \Cref{claim:internal-node}.}
\end{align*}

By \Cref{claim:internal-node}, and  \Cref{eq:sum-dagger-6n}, we have the followings. 
\begin{align*}
\sum_{u \in \text{level}(i+1)}n_u&= \sum_{u \in \text{level}^{\dagger}(i+1)}n_u + \sum_{u \in \text{level}^*(i+1)} n_u\\
 &\leq 6n + 2n = O(n).
\end{align*}
Therefore, the result follows.
\end{proof}

\begin{corollary} \label{cor:total-node}
The number of nodes in the  recursion tree $\mathcal{T}'$ is  $\ot(n)$. 
\end{corollary}
\begin{proof}
By \Cref{lem:totalsizen-each-level}, each level has total size $O(n)$, and there are $O(\log n)$ levels. Thus, total size is $O(n\log n)$. Each node has at least a unit size. Therefore, number of nodes is $O(n \log n)$. 
\end{proof}

\begin{corollary} \label{cor:total-leaves-tn}
The number of leaves in the recursion tree $\mathcal{T}$ is  $\ot(n)$. 
\end{corollary}

We are now ready to prove \Cref{lem:formula_tn}.

\begin{proof}[Proof of \Cref{lem:formula_tn}]
By \Cref{lem:teqc}, $T_{k,a}(n) = O(C_{k,a}(n))$, so it is enough to bound the cost for $C_{k,a}(n)$. Denote the recursion tree $\mathcal{T}'$ for the function relation $C_{k,a}(n)$. We have
\begin{align*}
 C_{k,a}(n) & = \sum_{i=0}^{O(\log n)} \sum_{u \in \text{level(i)}} \text{cost}(v) && \text{$O(\log n)$ depth by \Cref{eq:condition_recurrence_upperbound}.} \\
 &\leq \sum_{i=0}^{O(\log n) } ( \sum_{u \in \text{level}(i)} O(n_u^{x+a'})+ |\text{level(i)}|  t_{\textbase})   && \text{by \Cref{claim:cost-right-spine}.} \\
&\leq  \sum_{i=0}^{O(\log n)} (\sum_{u \in \text{level}(i)} n_u)^{x+a'} + |\mathcal{T}'|   t_{\textbase}    && \text{$|\mathcal{T}'|$ is the number of nodes in the tree.}\\
&= \ot(n^{x+a'} + n t_{\textbase}  ) && \text{by \Cref{lem:totalsizen-each-level,cor:total-node}.} \\
&= \ot(t_{\textcost} n^{a+o(1)} + t_{\textbase}n   ) && \text{$t_{\textcost} = n^{x},$ and $ a' = a + o(1).$}
\end{align*}
Therefore, the result follows.
\end{proof}

\subsubsection{Proof of \Cref{lem:main-vc-runtime}}  \label{sec:proof-main-vc-runtime}

Let $T'_{k,a}(n)$ be the running time of \Cref{alg:main-vc}.  By \Cref{lem:formula_tn}, $T_{k,a}(n) =  \ot(t_{\textcost} n^{a+o(1)}+   t_{\textbase}n  ) $. Hence, it is enough to show that  $$ T'_{k,a}(n) \leq \ot(T_{k,a}(n) + t_{\textlocal}  n). $$

Let $\mathcal{T}$ be a recursion tree for the \Cref{alg:main-vc}. Suppose there is a leaf-node $u$ that is not the base-case. This means $h(G) \geq \eta$, and we run $\LocalVC$ on every node in $G$, and return the answer.  We can overestimate this cost by extending  node $u$ to have two children $v, w$ of size  $\ell_v + s_v + k$, and $n_u - \ell_v+ k$. The total size is then $n_u +s_v+ 2k$. By design, LocalVC will be run at nodes $v$ and $w$ instead of node $u$. Hence, the running time due to LocalVC strictly increases. The recursion continues which also increases the total cost.    By repeat this process, we will end up with the recursion where all leaf nodes are base-case. Therefore, we have the upperbound in terms of $T_{k,a}(n)$ with extra-cost due LocalVC.  The extra-cost due to LocalVC is at most the total size of the leaf nodes in $\mathcal{T}$. By \Cref{cor:total-leaves-tn}, the total size of leaf nodes is $\ot(n)$. Hence, the number of LocalVC calls is at most $\ot(n)$. 

\subsection{Proof of \Cref{thm:VCAlgoMain}} 

We analyze \Cref{alg:main-vc}. First, the correctness follows from
\Cref{lem:main-vc-correct}. It remains to derive the final running
time.  The first term $O(m)$ follows from \Cref{thm:sparsification} as we first sparsify
the original graph in $O(m)$ time. Since $G$ (and later subgraphs in
the recursions) has aboricity $k$,  $m \leq
nk$, and from now we treat $m$ as $nk$.  We now derive the latter term.

We now instantiate the running time for the corresponding algorithms in \Cref{def:generic-runtime}.
\begin{corollary} \label{cor:current-runtime} We have the following running times for subroutines in \Cref{alg:main-vc}. 
\begin{itemize}[nolistsep,noitemsep]
\item There is an algorithm for the base-case that runs in time $O(m(n+k^3))$. 
\item There is an algorithm for approximate vertex-expansion $h(G)$ that runs in time $\ot(m^{1.6})$. 
\item There is a $\LocalVC$ algorithm that runs in $\ot(\min(\nu^{1.5}k, \nu k^k ))$ time. 
\item There is a $\SplitVC$ algorithm that runs in $O(mk(|S|+k^2))$ time.
\end{itemize}
\end{corollary}
\begin{proof}
 We use current best-known running time for the base case \cite{Gabow06} (with small $k$) which runs in time $O(mn+mk^3)$. By \Cref{thm:vertex expansion}, we can compute
approximate vertex-expansion in time $\ot(m^{1.6})$. By \Cref{thm:localvc}, LocalVC runs in time $\ot(\min(\nu^{1.5}k, \nu k^k ))$.  By \Cref{thm:split-vc-runtime}, SplitVC runs in time $O(mk(|S|+k^2))$.  
\end{proof}

Consequently, by \Cref{cor:current-runtime}, we  instantiate the running time for each term in \Cref{def:use-time} where we putting together subroutines in \Cref{def:generic-runtime} into \Cref{alg:main-vc}. %

\begin{corollary} \label{cor:use-time}
We have the running time for $t_{\textcost}, t_{\textlocal}, $ and $t_{\textbase}$ as follows. Assuming that $k > n^{\epsilon}$ where we can select sufficiently small $\epsilon > 0$. 
\begin{itemize}[nolistsep,noitemsep]
\item $t_{\textcost} =  \ot( n^{\theta} k^{\theta} + n^{1+a+o(1)}k^2+nk^4)$. 
\item $t_{\textlocal} = \ot( \min(n^{1.5-1.5a} k^4, n^{1-a} k^{2+k} ))$.
\item $t_{\textbase} = O(k^{1+4/a}+ k^{4+2/a})$. 
\end{itemize}
\end{corollary}
\begin{proof}
The results follow from simple algebraic calculation from \Cref{def:use-time}, and \Cref{cor:current-runtime}. For LocalVC, we use $\nu = O(k^2n^{1-a})$. 
\end{proof}

\begin{corollary} \label{cor:use-time-ez}
When $\theta > 1.5$,
the running time $\tilde T_{a,k}(n)$ for \Cref{alg:main-vc} can
be bounded by 
$$ 
 \tilde T_{a,k}(n) = \ot(n^{\theta+a+o(1)} + n^{2-a} k^{2 + k}). 
$$
\end{corollary}
\begin{proof}
By \Cref{lem:main-vc-runtime}, we have $ \tilde T_{a,k}(n) = \ot( t_{\textcost} n^{a+o(1)} + n (t_{\textlocal} + t_{\textbase}) )) $. By \Cref{cor:use-time} where $k = O(1)$, we have 
$t_{\textcost} = \ot(n^{\theta})$ (since $ a < 0.5$) , $t_{\textlocal} = \ot (n^{1-a}) $, and  $ t_{\textbase} = O(\baseCaseMainVC ^2) = O(n^{2\epsilon})$. 
\end{proof}

Setting $a = 0.25$ then gives the first term of the running
time as stated in \Cref{thm:VCAlgoMain}.
For the larger $k$ case, we instead use
$t_{\textlocal} = \ot( n^{1.5-1.5a} k^4 )$.

\begin{corollary} \label{cor:main-vc-runtime}
Let $a' = a + o(1)$.
The running time $\tilde T_{a,k}(n)$ for \Cref{alg:main-vc}
can be bounded by
\[
 \tilde T_{a,k}\left(n\right)
 =  \ot\left( n^{a'} n^{\theta}k^{\theta}
 	+ n^{2.5-1.5a}k^4 + nk^{4/a+1}\right)
\]
\end{corollary}
\begin{proof}
By \Cref{lem:main-vc-runtime}, we have
$ \tilde T_{a,k}(n)
=
\ot( t_{\textcost} n^{a+o(1)} + n (t_{\textlocal} + t_{\textbase}) )$.
By \Cref{cor:use-time}, we obtain the following. 
\begin{multline*}
\widetilde{T}_{a,k}\left(n\right)
\leq
\oh\left( n^{a'} \left( n^{\theta}k^{\theta}
  + n^{1+a'}k^2 + nk^4\right)
  	+ n^{2.5 - 1.5a}k^4
    + n\left(k^{4/a+1} + k^{5+2/a}\right)
\right)\\
\leq \oh\left( n^{a'} n^{\theta}k^{\theta}
+ n^{2.5-1.5a}k^4 + nk^{4/a+1}\right).
\end{multline*}
The last equality follows by the assumption of $k < n^{1/8}$
(precondition for \Cref{thm:VCAlgoMain}), and $a < 0.5$.
\end{proof}

Here the optimal choice of $a$ is 
\[
a = \frac{(2.5-\theta)+(4 - \theta) \log_n k}{2.5}.
\]
By the assumption on $k$, we have that $a < 0.5$,
and thus $a$ satisfies the precondition for \Cref{alg:main-vc}.
Substituting this vlaue into \Cref{cor:main-vc-runtime}
gives the final term of running time for \Cref{alg:main-vc}, which is
$\oh(n^{1+ 0.6\theta}k^{1.6+0.6\theta})$. Therefore, we conclude the proof of \Cref{thm:VCAlgoMain}.

\section{Low-Vertex-Expansion Cuts from Balanced Low-conductance Cuts}

\label{sec:cutmatching}

In this section, we show how an approximate balanced cut algorithm can
be used to give an algorithm that approximates the minimum vertex expansion.

\VertexExpansionMain*

\begin{theorem}
\label{thm:vertex expansion}There is a deterministic algorithm that,
given a graph $G$ and a parameter $\eta$, runs in $\tilde{O}(m^{1.6})$
and either 
\begin{itemize}
\item certifies that $h(G)\ge\eta$, or
\item returns $(L,S,R)$ where $h(L,S,R)\le\eta n^{o(1)}$.
\end{itemize}
\end{theorem}

This algorithm is based on the \emph{cut-matching game }introduced
by Khandekar, Rao, and Vazirani \cite{KhandekarRV09}.%

{} In the original paper \cite{KhandekarRV09}, they show a specific
randomized approach, based on random projection, for implementing
the framework very fast. Later, it is shown in \cite{KhandekarKOV07}
that such random projection is not inherent: given any algorithm for
finding a sparsest cut, then the framework can be implemented. In
particular, this framework can be implemented deterministically. However,
in \cite{KhandekarKOV07}, they use \emph{exact} algorithm for computing
sparsest cut which is NP-hard. Below, we will show that the idea \cite{KhandekarKOV07}
still works even if when we use only approximate algorithms for finding
sparse cuts.

Below, we first prove two main steps in the cut matching game in our
context. In \Cref{subsec:embed}, we show that how to lower bound
the vertex expansion by \emph{embedding} an expander. In \Cref{subsec:matching},
we show how to find low vertex expansion cut using a single commodity
flow. Then, \Cref{subsec:cutmatching}, we describe the cut-matching
game variant by \cite{KhandekarKOV07} with a relaxation
that we can use approximation algorithms. Then, we show how everything fits
together.

The main technique in this section is the \emph{cut matching game}
by Khandekar, Rao, and Vazirani \cite{KhandekarRV09}. This is a very
flexible framework that can be used to certify various notion of expansion
of graphs. The framework was used for approximating the sparsity $\sigma(G)$
of a graph $G$ where $\sigma(G)=\min_{S}\frac{|E(S,V-S)}{\min\{|S|,|V-S|\}}$
in \cite{KhandekarRV09} and for approximating the conductance $\Phi(G)$
of $G$ in \cite{SaranurakW19}. For our purpose, it is not hard to
adjust the framework so that it works for vertex expansion $h(G)$
of the graph $G$.

Unfortunately, one of the two main components of the framework in
\cite{KhandekarRV09} is randomized. More precisely, their framework
requires computing a cut $C$ with some specific property, and they
give a very fast randomized algorithm for computing such $C$ based
on a random-projection technique. It is not clear how to derandomize
the algorithm for computing such cut. Fortunately, in a technical
report by Khandekar, Khot, Orecchia, Vishnoi \cite{KhandekarKOV07},
they show how to adjust the analysis of \cite{KhandekarRV09} 
so that it works when $C$ is the most-balanced low-conductance cut.
However, this is an NP-hard problem.

We observe that we do not need an exact algorithm for computing $C$.
For this, we require an additional property of our approximate
balanced cut routine, that it returns a set $S$ with $|S| \leq n /2$
such that $\Phi(G[V -  S]) \geq \alpha$, or that the graph
with $S$ removed is an expander.
It turns out that there is a black box method \cite{NanongkaiS17,Wulff-Nilsen17}
to obtain this guarantee from an approximate balanced
cut algorithm.
Formally, we obtain from this black-box,
$T_{cut}=\oh(m^{1.5})$ and $\alpha=1/n^{o(1)}$,
which in turn gives the overall running time.

\subsection{Lower Bounding Vertex Expansion via Expander Embedding}

\label{subsec:embed}

We first define some basic notions about \emph{flow}. Although we
are working with undirected vertex-capacitated graphs, we will define
the problems also on directed graphs and also edge capacitated graphs.
Let $G$ be a directed graph $G=(V,E)$ and $s,t\in V$. An $s$-$t$
flow $f\in\mathbb{R}_{\ge0}^{E}$ is such that, for any $v\in V-\{s,t\}$,
the amount of flow into $v$ equals the amount of flow out of $v$,
i.e., $\sum_{(u,v)\in E}f(u,v)=\sum_{(v,u)\in E}f(v,u)$. Let $f(v)=\sum_{(u,v)\in E}f(u,v)$
be the \emph{amount of flow at $v$}. The \emph{value} $|f|$ of $f$
is $\sum_{(s,v)\in E}f(s,v)-\sum_{(v,s)\in E}f(v,s)$. %

{} Let $c\in(\mathbb{R}_{>0}\cup\{\infty\})^{V}$ be vertex capacities.
$f$ is \emph{vertex-capacity-feasible} if $f(v)\le c(v)$ for all
$v\in V$. If $G$ is undirected, one way to define an $s$-$t$ flow
is by treating $G$ as a directed graph where there are two directed
edge $(u,v)$ and $(v,u)$ for each undirected edge $\{u,v\}$. We
will assume that a flow only goes through an edge in one direction,
i.e., for each edge $\{u,v\}\in E$, either $f(u,v)=0$ or $f(v,u)=0$. 

For $1\le i\le k$, let $f_{i}$ be an $s_{i}$-$t_{i}$ flow with
value $d_{i}$. We call $\cF=\{f_{1},\dots,f_{k}\}$ a \emph{multi-commodity
flow}. We call the $k$ tuples $(s_{1},t_{1},d_{1}),\dots,(s_{k},t_{k},d_{k})$
the \emph{demands }of $\cF$. We say that $\cF$ is \emph{edge-capacity-feasible
}if $\sum_{i}f_{i}(e)\le c(e)$ for each $e\in E$, and is \emph{vertex-capacity-feasible}
if $\sum_{i}f_{i}(v)\le c(v)$ for all $v\in V$. We usually just
write \emph{feasible}. Let $W=(V,E,w)$ be a weighted graph. If a
multicommodity flow $f$ has the demands $\{(u,v,w(e))\mid e=(u,v)\in W\}$,
then we say that $f$ \emph{respects} $W$.

\begin{definition}
[Embedding with Vertex Congestion]Let $G=(V,E)$ and $W=(V,E_{W})$
be two graphs with the same set of vertices. We say that $W$ can
be embedded into $G$ with vertex congestion $c$ iff there exists
a feasible multicommodity flow $f$ in $G$ respecting $W$ when each
node in $G$ has capacity $c$.
\end{definition}

\begin{lemma}
\label{lem:embed}
Suppose $W$ has sparsity $\sigma(W) \geq \phi$, and 
$W$ can be embedded into $G$ with vertex congestion $c$, then $h(G)\ge\phi/2c$.
\end{lemma}

\begin{proof}
Consider any separation triple $(L,S,R)$ where $|S|\le|L|,|R|$,
otherwise $h(L,S,R)\ge1/2\ge\phi/2c$. Assume w.l.o.g. that $|L|\le|R|$.
In particular, $|L|\le|\overline{L}|$. Then we have $E_{W}(L,\overline{L})\ge\phi|L|$.
As $W$ can be embedding into $G$ with vertex congestion $c$, $|N_{G}(L)|=|S|\ge\phi|L|/c$.
So we conclude $h(L,S,R)=\frac{|S|}{\min\{|L|,|R|\}+|S|}\ge\frac{|S|}{(c/\phi)\cdot|S|+|S|}\ge\phi/2c$.
\end{proof}

\subsection{Finding Low Vertex Expansion Cut via Single Commodity Flow}

\label{subsec:matching}
We prove the following in this section: 
\begin{lemma}
\label{lem:matching}Let $G=(V,E)$ be a $n$-vertex $m$-edge graph.
Let $A,B\subset V$ be two disjoint vertex sets. Let $c$ be a congestion
parameter. There is a deterministic algorithm that runs in $\tilde{O}(m\sqrt{n})$
time, and either 
\begin{itemize}
\item return an embedding $M$ into $G$ with vertex congestion $c$, where
$M$ is a matching of size $\min\{|A|,|B|\}$ between vertices of
$A$ and $B$ , or
\item return a separation triple $(L,S,R)$ where $h(L,S,R)<1/c$.
\end{itemize}
\end{lemma}

To prove the above, we define the following flow problem. 
Let $G'$ be a vertex-capacitated graph defined from $G$ as follows.
We create a source vertex $s$ and a sink vertex $t$. For each vertex
$v\in A$, we add a dummy vertex $v'$ and an edge $(s,v')$ and $(v',v)$.
Let $A'$ be the set of such dummy vertices incident to $s$. For
each vertex $v\in B$, we add a dummy vertex $v'$ and an edge $(t,v')$
and $(v',v)$. Let $B'$ be the set of such dummy vertices incident
to $t$. Each vertex in $V$ has capacity $v$. Each vertex in $A'\cup B'$
has capacity $1$. The vertices $s$ and $t$ have capacity $\infty$.

The algorithm just computes a max flow in $G'$ using an deterministic $\tilde{O}(m\sqrt{n})$-time
algorithm by the subset of the authors \cite{NanongkaiSY19} (or a
more well-known one with running time $\tilde{O}(m^{1.5})$ by Goldberg
and Rao \cite{GoldbergR98}). If the size of max flow value is $\min\{|A|,|B|\}$,
then by definition, we can obtain an embedding $M$ into $G$ as stated
in \Cref{lem:matching}. Now, we need to show that if the max flow
value is less, then we obtain $(L,S,R)$ where $h(L,S,R)<1/c$.
\begin{lemma}
Let $(L',S',R')$ be the minimum vertex cut in $G'$ where $s\in L'$
and $t\in R'$. If the cut value of $(L',S',R')$ is less than $\min\{|A|,|B|\}$,
then $(L,S,R)=(L'\cap V,S'\cap V,R'\cap V)$ is a vertex-cut in $G$
where $h(L,S,R)<1/c$.
\end{lemma}

\begin{proof}
Observe that the cut value of $(L',S',R')$ is 
\[
|S'\cap A'|+|S'\cap B'|+c|S'-A'\cup B'|
\]
which is less than $|A|$ by assumption. Note that $(S'-A'\cup B')=(S'\cap V)=S$. 

Observe that $|S'\cap A'|\ge|A\cap R|$. This is because for each
$v\in A\cap R$, we must include the corresponding dummy node into
$v'$ into $S'$, otherwise $s$ can reach $R\subset R'$ and hence
and reach $t$, and this would mean that $(L',S',R')$ is not a $s$-$t$
vertex cut. So we have
\[
|A|>|S'\cap A'|+|S'\cap B'|+c|S'-A'\cup B'|\ge|A\cap R|+c|S|.
\]
Hence,
\[
|A\cap(L\cup S)|=|A|-|A\cap R|>c|S|.
\]
So we have $|L|+|S|>c|S|$. Symmetrically, we also conclude that $|R|+|S|>c|S|$.
In particular, $L\ne\emptyset$ and $R\ne\emptyset$. As $(L,S,R)$
is obtained from $(L',S',R')$ by just removing $A'\cup B'\cup\{s,t\}$
from the graph, there must be no edge between $L$ and $R$ as there are no edges between $L'$ and $R'$. That
is, $(L,S,R)$ is indeed a separation triple. Therefore, we conclude that
$h(L,S,R)=|S|/(\min\{|L|,|R|\}+|S|)<1/c$. 
\end{proof}

\subsection{Cut with expander complement}

In this section, we show the following:

\begin{theorem}
	\label{thm:det cond}
	Given an $(f(\phi), \beta)$-approximate balanced cut
	algorithm $\A$ such that $f(\phi) \leq \phi^{\xi} \polylog(n)$
	for some absolute constant $0 < \xi \le 1$
	and $\beta \leq O(\log^{4}n)$ where $\A$ works on graph instances with maximum degree $O(\log n)$,
	we can obtain an algorithm that
	takes an undirected $n$-vertex $m$-edge graph $G$ with
	maximum degree $O(\log n)$, 
	calls $\A$ for $n^{o(1)}$ many times with parameter $\phi \ge 1/n^{o(1)}$ and uses 
	$\oh(m)$ overhead, then 
	either
	\begin{itemize}
		\item certifies that $\Phi(G)\ge\gamma$, or
		\item outputs a cut $S$ where $\Phi_{G}(S)\le1/\log^2 n$ and $\vol(S)=\Omega(m/\log n)$, 
		\item outputs a cut $S$ where $\Phi_{G}(S)\le1/\log^2 n$, $\vol(S)=O(m/\log n)$,
		and $\Phi(G[V-S])\ge$$\gamma$, 
	\end{itemize}
	where $\gamma=1/n^{o(1)}$.
\end{theorem}

We follow the techniques from~\cite{NanongkaiS17,Wulff-Nilsen17},
which show how to use an (approximate) most-balanced low-conductance
algorithm to obtain a cut with expander complement.
Here, we state their results. 

\BL[\cite{NanongkaiS17,Wulff-Nilsen17}]\label{thm:cut expander complement reduction}
Let $f$ be a function such that $f(\phi)\ge\phi$ for all $\phi\in[0,1]$.
Let $c_{size}=c_{size}(n)$ be some number depending on $n$. Suppose
that there is an algorithm $\A$ that, given a $n$-vertex $m$-edge graph
$G$ with maximum degree $\Delta$ and a parameter $\phi$, either 
\begin{itemize}
	\item certifies that $\Phi_{G}\ge\phi$, or
	\item returns a $(\phi,c_{size}(n))$-most-balanced $f(\phi,n)$-conductance cut. 
\end{itemize}
Then, for any $k\ge 1$, there is an algorithm that, given a $n$-vertex
graph $G$ with maximum degree $\Delta$ and a parameter $\phi$, calls $\A$ for $O(m^{1/k}/c_{size})$ many times plus $\oh(m)$ overhead and then either
\begin{itemize}
	\item certifies that $\Phi_{G}\ge\phi$, or
	\item returns a cut $S$ where $\Phi_{G}(S)\le f^{k}(\phi)$ where $\vol(S)\ge\Om(m)$,
	or
	\item returns a cut $S$ where $\Phi_{G}(S)\le f^{k}(\phi)$ where $\Phi(G[V-S]) \ge \phi$,
\end{itemize}
where $f^{1}(\phi) = f(\phi,n)$ and $f^{k}(\phi) = f^{k-1}(f(\phi))$.
\EL

\begin{proof}[Proof of \Cref{thm:det cond}]
	The given approximate balanced cut routine either
	(1) certifies that $\Phi_{G}\ge\phi$, or
	(2) returns a $(\phi,c_{size}(n))$-most-balanced $f(\phi,n)$-conductance cut,
	where $c_{size}(n) = O(\log^4 n)$ and
	$f(\phi,n) \leq \phi^{\xi} \polylog(n)$. 
	
	Observe that $f^{k}(\phi)=\phi^{\xi^{k}} (\polylog(n))^k$.
	In order to have $f^{k}(\phi)\le1/\log^{2}n$, we can set
	\[
	\phi=\left( \frac{1}{ (\polylog(n))^k } \right)^{\xi^{-k}}.
	\]
	If we set $k = c\cdot \log\log n$ for small enough constant $c$, then
	we have $\xi^{-k} = \log^\eps(n)$ for a very small constant $\eps >0$.
	Therefore, $\phi = 1/n^{o(1)}$.  
	So \Cref{thm:cut expander complement reduction} gives us
	the routine that we need for \Cref{thm:det cond} where the number of calls to $\A$ is $O(m^{1/k}/c_{size}) = n^{o(1)}$.
\end{proof}

\subsection{Cut-matching game via approximate low-conductance cut}

\label{subsec:cutmatching}

Given an $n$-vertex graph $G$, the cut-matching game from \cite{KhandekarKOV07} is a framework
which proceeds in rounds. There will be $O(\log n)$ rounds. Let $W_{0}=\emptyset$.
Let $c$ be a congestion parameter. Let $\A$ be the algorithm from
\Cref{thm:det cond}. At round $i$, we maintain the invariant that
$W_{i-1}$ is already embedded into $G$ with congestion $c\times(i-1)$
and $\vol(W_{i-1})=\Omega(n(i-1))$. We run $\A$ on $W_{i-1}$ which can result in three cases.

First, if $\A$ certify that $\Phi(W_{i-1})\ge\gamma$. Hence, $\sigma(W_{i-1})\ge\gamma$
and we have that $h(G)\ge\gamma/(2c(i-1))=\Omega(\gamma/(c\log n))$
by \Cref{lem:embed}, and so we terminate.

Second, if $\A$ returns a cut $S$ where $\Phi(S)\le1/\log^2 n$ and
$\vol(S)=\Omega(\vol(W_{i-1}))$, then we set $A=S$ and $B=V-S$
and invoke \Cref{lem:matching}. If we obtain $(L,S,R)$ where $h(L,S,R)<1/c$,
we terminate. Otherwise, we obtain a matching $M_{i}$ of size $|M_{i}|\ge\min\{|S|,|V-S|\}\ge\Omega(\vol(W_{i-1})/(i-1))=\Omega(n)$.
We set $W_{i}\gets W_{i-1}\cup M_{i}$ and hence $\vol(W_{i})=\Omega(ni)$
and $W_{i}$ is embeddable into $G$ with congestion $c\times i$. 

Third, if $\A$ returns a cut $S$ where $\Phi(S)\le1/\log^2 n$ and
$\Phi(G[V-S])\ge\gamma$, then we set $A=S$ and $B=V-S$ and invoke
\Cref{lem:matching}. If we obtain $(L,S,R)$ where $h(L,S,R)<1/c$,
we terminate. Otherwise, we obtain a matching $M_{i}$ embeddable
to $G$. We know $|M_{i}|=|S|$ because $|S|\le\vol(S)\le\vol(V-S)/\Theta(\log n)\le|V-S|$
by \Cref{thm:det cond}. We set $W_{i}\gets W_{i-1}\cup M_{i}$. We
claim that $\sigma(W_{i})\ge\Omega(\gamma)$. As $W_{i}$ is embeddable
into $G$ with congestion $c\times i$, so $h(G)\ge\gamma/(2ci))=\Omega(\gamma/(c\log n))$
by \Cref{lem:embed}, and so we terminate.

In \cite{KhandekarKOV07}, the following is proven:
\begin{lemma}
[Section 4 of \cite{KhandekarKOV07}]The second case can occur at most $O(\log n)$ times. 
\end{lemma}

Therefore, then indeed there at $O(\log n)$ rounds. This is because
the algorithm will terminate whenever the first or third case occur. 
\begin{lemma}
In the third case, if we obtain a matching $M_{i}$, then $\sigma(W_{i})\ge\Omega(\gamma)$.
\end{lemma}

\begin{proof}
Let $S$ be the cut that $\A$ returns in the third case. Consider
any cut $C\subset V$ where $|C|\le|V-C|$. If $|C\cap S|\ge2|C-S|$
and so $|C\cap S|\ge2|C|/3$, then there is 
\begin{align*}
|E_{M_{i}}(C\cap S,V-S)| & \ge|E_{M_{i}}(C\cap S,V-C\cap S)|-|E_{M_{i}}(C\cap S,C-S)|\\
 & \ge|C\cap S|-|C-S|\\
 & \ge|C\cap S|/2\\
 & \ge|C|/3.
\end{align*}
Next, if $|C\cap S|\le2|C-S|$ and so $|C-S|\ge|C|/3$then 
\begin{align*}
|E_{W_{i-1}}(C-S,V-S)| & \ge\gamma|C-S|\\
 & \ge\gamma|C|/3.
\end{align*}
That is, $E_{W_{i}}(C,V-S)\ge\gamma|C|/3$. Hence, $\sigma(W_{i})\ge\gamma/3$.
\end{proof}
Observe that the above algorithm just invoke the algorithm from \Cref{lem:matching}
and \Cref{thm:det cond} $O(\log n)$ times.
Hence, the total running time
is $O(\log n)\times(\tilde{O}(m\sqrt{n})$
plus the cost of $O(\log{n})$ calls to the approximate balanced cut routine.
To conclude, whenever the first case occurs or a matching is obtained
in the third case, we have $h(G) = \Omega(\gamma/(c\log n))$.
Otherwise, we must obtain $(L,S,R)$ where $h(L,S,R) < 1/c$ during some round. 
By setting $\eta=\Theta(\gamma/(c\log n))$, we conclude the proof
of \Cref{thm:VertexExpansionMain}.

\section{Near-Expander Low-Conductance Cuts on Dense Graphs}
\label{sec:pagerank}

The main result in this section is an approximate balanced cut
algorithm on dense graphs.

\PageRankMain*

Although, the algorithm is stated for general parameters,
to start with, it is more convenient to think of $\phi,\phi^{*}=1/\poly\log n$
but $\phi\gg\phi^{*}$, and $c=\poly\log n$.
Our goal is to find a $(\phi^{2}/\poly\log m,\tilde{O}(1))$-balanced
$\phi$-conductance in $O(n^{\omega})$ time.
We emphasize that the requirement $\vol(S)\ge\vol(S^{*})/c$ is crucial
for us.
Without this requirement, there is in fact a very simple $O(n^{\omega})$
time algorithm that we discuss in Section~\ref{sec:Related}

At the high level, our algorithm is a derandomization of the \emph{PageRank-Nibble}
algorithm by Andersen, Chung, and Lang \cite{AndersenCL06} for finding
$(\phi^{2}/\poly\log m,O(1))$-balanced $\phi$-conductance cut in
$\tilde{O}(m/\phi^{2})$ time. However, there are several obstacles
we need to overcome. To show how, we need some basics about PageRank
and the algorithm of \cite{AndersenCL06}. For any vertex $v$, the
\emph{PageRank vector }$p_{v}\in\mathbb{R}_{\ge0}^{V}$\emph{ of vertex
	$v$ }is a vector encoding a distribution over vertices when we perform
some variant of random walk starting at $v$. Roughly speaking, PageRank-Nibble
works as follows::
\begin{enumerate}
	\item Sample a vertex $v$ according to some distribution.
Compute an approximation of $p_{v}$.
	\item Find a $\phi$-conductance cut
$C$ by which is a \emph{sweep cut} w.r.t. to $p_{v}$, if exists.
(We will define sweep cuts later).
	\item Then, set $G\gets G[V-C]$ if $C$
is found and repeat.
	\item The next PageRank vector $p_{v}$ is w.r.t. the
new graph.
\end{enumerate}
The running time obtained by Andersen et al.~\cite{AndersenCL06}
critically relies on the vertex $v$ being sampled randomly.
If $v$ is chosen in some arbitrary order, the
running time can be $\Omega(mn)$: we may need to check $\Omega(n)$
vertice,  spend $\Omega(m)$ time to compute (an approximation of)
$p_{v}$ for each $v$, and still could not find any
$\phi$-conductance cut.

Our first key idea is to instead compute the
PageRank vector $p_{v}$ w.r.t. to the input graph $G$ for
all $v\in V$ \emph{simultaneously} by inverting the PageRank
matrix defined from $G$.
That is, instead of sampling, we consider all starting points $v$.

The more subtle and  challenging obstacle is to compute sweep
cuts. A sweep cut w.r.t. $p_{v}$ is a cut of the form $V_{\ge t}^{p_{v}}$
where $V_{\ge t}^{p_{v}}=\{u\in V\mid p_{v}(u)\ge t\}$. In \cite{AndersenCL06},
the cost for computing sweep cut can be charged to the cost for computing
an approximation of $p_{v}$. As our approach for computing $p_{v}$
changes, we need to account the cost ourselves. We can trivially check
if there is a sweep cut $V_{\ge t}^{p_{v}}$ with conductance at most
$\phi$ in $O(m)$. To do this, we compute $|E(V_{\ge t}^{p_{v}},V-V_{\ge t}^{p_{v}})|$
and $\vol(V_{\ge t}^{p_{v}})$ of \emph{all }$t$ in $O(m)$, by sorting
vertices according their values in $p_{v}$ and ``sweeping'' through
vertices in the sorted order. Unfortunately, spending $O(m)$ time
for each vertex $v$ would give $O(mn)$ time algorithm which is again
too slow.\footnote{With randomization, given $p_{v}$ for all $v$, we can compute a
	sweep cut for all $v$ in $\tilde{O}(n^{2})$ total time. For example,
	we can precompute the degrees of vertices in $O(m)$ time. So we can
	compute $\vol(V_{\ge t}^{p_{v}})$ of all $t$ in $O(n)$ for each
	$v$. Then, we can build a randomized data structure based on linear
	sketches in $\tilde{O}(m)$ time. So we can approximate $|E(V_{\ge t}^{p_{v}},V-V_{\ge t}^{p_{v}})|$
	of all $t$ in $O(n)$ for each $v$.} To the best of our knowledge, there is no deterministic data structure
even for checking whether $|E(S,V-S)|>0$ in $o(m)$ time, given a
vertex set $S\subset V$. That is, it is not clear how to approximate
$|E(V_{\ge t}^{p_{v}},V-V_{\ge t}^{p_{v}})|$ in $o(m)$ time even
for a fix $t$.

To overcome this second obstacle, we show that how to obtain a sweep
cut without approximating the cut size $|E(V_{\ge t'}^{p_{v}},V-V_{\ge t'}^{p_{v}})|$
for any $t'$. More precisely, by preprocessing the PageRank vectors
$p_{v}$ of all $v$ in $\tilde{O}(n^{2})$, we show how obtain a
sweep cut with conductance $\phi$, if exists, in time $\poly\log m$
instead of $O(m)$. To do this, we exploit the fact the sweep cut
is w.r.t. a PageRank vector $p_{v}$ and not some arbitrary vector.
This allows us to do a binary search tree for $t$ where the condition
depends solely on the volume $\vol(V_{\ge t}^{p_{v}})$ and not the
cut size $|E(V_{\ge t}^{p_{v}},V-V_{\ge t}^{p_{v}})|$. The details
of technique is in \Cref{sec:sweepcut}. Once all sweep cuts are computed.
We follow the analysis of \cite{AndersenCL06,KawarabayashiT15} and
obtain $(\phi^{2}/\poly\log m,O(1))$-balanced $\phi$-conductance
cut.

\subsection{Near-Expander Cuts}

For the analysis of this algorithm, it is most convenient
to work near-expanders and near-expander cuts.

\BD[Nearly Expander]
Given $G=(V,E)$ and a set of nodes $A\s V$, we say $A$ is a nearly $\phi$-expander in $G$ if
\[ \forall S\s A,\ \vol(S)\le\vol(A)/2:\ |E(S,V\setminus S)|\ge\phi\vol(S) .\]
\ED

\BD[Near-Expander Edge Cut]\label{def:near-expander}
Given $G=(V,E)$ and parameters $\phi_1,\phi_2$, a set $S\s V$ with $\vol(S)\le m$ is a \emph{$\phi_2$-near-expander $\phi_1$-conductance cut} if it satisfies:
\BE
\im $\Phi(S) \le \phi_1$,
\im Either $\vol(S)\ge m/4$, or $V\setminus S$ is a nearly $\phi_2$-expander in $G$.
\EE
\ED

We can then prove the following equivalent form of
\Cref{thm:pagerank-main}.

\BT\label{thm:main-w}
Given a multigraph $G=(V,E)$ and parameter $\phi$, there is a deterministic algorithm that runs in $O(n^\om)$ time and either:
\BE
\im Outputs a $\Om(\phi/\log(m))$-near-expander $O(\sr{\phi\log m})$-conductance cut, or
\im Certifies that $\Phi(G) \ge \phi$.
\EE
\ET

We use this notion since it is more general and could prove useful in future applications.
We can use it to obtain a most balanced low conductance cut
(as defined in Definition~\ref{def:ExpansionApprox}) by the lemma below.

\BL\label{lem:ne-mb}
Given $G$ and parameters $\phi_1,\phi_2<1$ and $c\ge4$, if $S$ is a $(\f{c\phi_2+\phi_1}{c-1})$-near-expander $\phi_1$-conductance cut, then $S$ is a $(\phi_2,c)$-most-balanced $\phi_1$-conductance cut. 
\EL
\BP

If $\vol(S)\ge m/4$, then $S$ must be a $(\phi_2,4)$-most-balanced $\phi_1$-conductance cut, since any cut $S'$ with $\vol(S')\le m$ must satisfy $\vol(S')\le m\le4\vol(S)$. Since $c\ge4$, $S$ is also a $(\phi_1,c)$-most-balanced $\phi_2$-conductance cut.

From now on, suppose that $\vol(S)<m/4$. Suppose for contradiction that $S$ is not a $(\phi_2,c)$-most-balanced $\phi_1$-conductance cut.  Then, there exists a set $T$ ($\vol(T)\le m$) of conductance at most $\phi_2$ with 
\begin{gather}
\vol(T)>c\vol(S).\label{eq:bal1}
\end{gather}
Our goal is to show that the set $T\setminus S \s V\sm S$ satisfies 
\begin{gather}
|E(T\sm S,V\sm(T\sm S))|<\f{c\phi_2+\phi_1}{c-1}\vol(T\sm S)\label{eq:bal2},\end{gather}
 contradicting the assumption that $V\sm S$ is a nearly $\big(\f{c\phi_2+\phi_1}{c-1}\big)$-expander.

We have
\[ \vol(T\sm S)\ge\vol(T)-\vol(S)\stackrel{(\ref{eq:bal1})}>\vol(T)-\f1c\vol(T)=\f{c-1}{c}\vol(T)\]
and
\[ |E(T\sm S , V\sm(T\sm S))| \le |E(T,V\sm T)|+|E(S,V\sm S)|\le \phi_2\vol(T)+\phi_1\vol(S)\stackrel{\mathclap{(\ref{eq:bal1})}}<\phi_2\vol(T)+\phi_1\cd\f{1}{c}\vol(T). \]
Therefore,
\[ \f{|E(T\sm S,V\sm(T\sm S))|}{\vol(T\sm S)}< \f { (\phi_2+\phi_1/c)\vol(T)} {(c-1)/c\cd\vol(T)} = \f{c\phi_2+\phi_1}{c-1}  ,\]
establishing (\ref{eq:bal2}).
\EP

The rest of this section is for proving \Cref{thm:main-w} and is organized as follows.
In \Cref{sec:prelim_pagerank}, we give definitions and basic properties about PageRank.
In \Cref{sec:sweepcut}, we show the key technical result of this section. 
This is the subroutine for finding a low conductance cut, given a PageRank vector. 
Previous algorithms based on PageRank need to compute the size of a cut (i.e. the number of edge crossing the cut). 
In fact, most algorithms ``sweep through'' a PageRank vector and compute the size of $n$ many cuts. 
Hence, this procedure is often called ``sweep cut''.
We show an algorithm that, given a PageRank vector and access to degrees of vertices, does not need to compute a cut size of any cut at all, and can still guarantee to return a low conductance cut (given appropriate parameters).

Subsequent sections show how to exploit the algorithm in  \Cref{sec:sweepcut}. 
We present in a backward manner.
\Cref{sec:cut_from_excess} shows how to compute find a low conductance cut by calling the algorithm in \Cref{sec:sweepcut} once, given a set of vertices with large \emph{excess} (where excess is defined from a PageRank vector).
\Cref{sec:large_excess} shows that if there exists a low conductance cut, then there must exist a set of vertices with large excess.
\Cref{sec:near_expander_cut} shows the final algorithm which computes the PageRank vectors from all vertices simultaneously, and obtain
many low conductance cuts by the guarantees of previous sections, and then combines them to obtain the near-expander cut as desired in  \Cref{thm:main-w}.

\subsection{Preliminaries about PageRank}
\label{sec:prelim_pagerank}

We will follow \cite{KawarabayashiT15}'s treatment of PageRank [ACL'06], which is more algorithmic and can be better adapted to our binary search algorithm.
For our algorithm and analysis, we will only need the following definition of PageRank vector, following Section~2.1 of \cite{AndersenCL06}:

\BD[PageRank vector {\cite{AndersenCL06}}]
Given real number $\al\in(0,1]$ and vector $\mathbf v\in \R^V$ satisfying $\1^T\bv=1$, the \emph{PageRank vector starting at $\mathbf v$}, denoted $PR(\mathbf v)$, is the unique solution to
\begin{gather}
 PR(\bv) = \al\bv+(1-\al)PR(\bv)W . \label{eq:pr}
\end{gather}
Here, $W=\f12(I+D\inv A)$, where $A$ is the adjacency matrix of $G$ and $D$ is the diagonal matrix with entry $\deg(v)$ at row and column $v$, for each $v\in V$. ($W$ is known as the lazy random walk transition matrix.)

We refer to $p$ as a \emph{PageRank vector} if $p=PR(\bv)$ for some (possibly unspecified) vector $\bv$ with $\1^T\bv=1$.
\ED

We will use the following important properties of PageRank vectors; see Section~2 of \cite{AndersenCL06} for more details.
\begin{fact}\label{fact:pr}
For any vector $\bv$, the PageRank vector $PR(\bv)$ is unique. Moreover, it satisfies $PR(\bv)\ge0$ and $\1^T PR(\bv)=1$.
\end{fact}

\BD
Given a PageRank vector $p$ and a real number $t$, we define the vertex sets $\V t:=\{v\in V:p(v)/d(v)\ge t\}$ and $V^p_{\le t}:=\{v\in V:p(v)/d(v)\le t\}$.
\ED

We show some terminology and lemmas from \cite{KawarabayashiT15}.

\BD[Median Expansion]
Given any value $t$ such that $\pt\V t\ne\emptyset$, there exists some $t_{med}<t$ (the ``median'') such that half the edges $\pt V^p_{\ge t}$ go to vertices in $V^p_{\ge t_{med}}$ and half go to vertices in $V^p_{\le t_{med}}$. We call $t_{med}$ the \emph{median expansion} at $t$.
\ED
\BCL\label{lem:med1}
If $t_{med}$ is the median expansion at $t$, then:
\begin{gather}
\vol(V^p_{\ge t_{med}}) \ge \vol(V^p_{\ge t}) + |\pt V^p_{\ge t}|/2 . \label{eq:med1}
\end{gather}
\ECL
\BP
By definition, at least half the edges of $\pt\V t$ go to vertices in $\V{t_{med}}$. These edges make up at least $|\pt \V t|$ extra volume in $\V{t_{med}}$ when compared to $\V t$.
\EP
The following lemma is proved in the proof of Lemma~33 in \cite{KawarabayashiT15}:
\BL[Lemma~33 of {\cite{KawarabayashiT15}}]\label{lem:med2}
If $t_{med}$ is the median expansion at $t$, then:
\begin{gather}
t-t_{med} \le \f{6\al}{|\pt V^p_{\ge t}|} %
. \label{eq:med2}
\end{gather}
\EL

\subsection{Sweep Cut without Cut-size Query}
\label{sec:sweepcut}
Previous algorithms based on PageRank need to compute the number of edges crossing some cut, for at least one cut.
This might take $O(m)$ time in the worst-case. In this section, we show that this we do not need to query for the cut-size at all.

\BL[Sweep-cut with few queries]\label{lem:sweep-cut}
Suppose we have PageRank vector $p$, a real number $t_0$  that satisfies $\vol(V^p_{\ge t_0})\le 1.5m$, and a real number $\tau\in(t_0,1]$. Suppose that we have access to an data structure that, given $t$, can return 
$\VV{t}=\sum_{v\in\V{t}}\deg(v)$ in $O(\log n)$ time.
Then, for any $\tau \in (t_0,1]$, we can compute a cut of conductance at most
$\sr { \f {54\al}{(\tau-t_0) \vol(V^p_{\ge\tau})} }$ in $O(\log (\frac{\log m}{\phi} )) = O(\log m)$ queries. 
Hence, the running time is $O(\log^2 m)$.
\EL
We note that our constraint $\VV{t_0}\le1.5m$ is looser than the constraint $\VV{t_0}\le m$ in Section~7.2 of \cite{KawarabayashiT15}. This explains the difference in our bound $\sr { \f {54\al}{(\tau-t_0) \vol(V^p_{\ge\tau})} }$ compared to the bound $\sr { \f {18\al}{(\tau-t_0) \vol(V^p_{\ge\tau})} }$ in \cite{KawarabayashiT15}

\begin{algorithm}[H]
\caption{SweepCutBinarySearch$(G,p,t_0,\tau)$}
Assumptions: $t_0$ satisfies $\VV{t_0}\le 1.5m$ and $\tau$ satisfies $\tau\in(t_0,1]$; access to an data structure $\D$ that, given $t$, can return $\VV{t}=\sum_{v\in\V{t}}\deg(v)$ in $O(\log n)$ time.
\\Output: A cut with conductance at most $\phi$ (as defined in line~\ref{line:def-phi})
\\Runtime: $O(\log m)$ queries to $\D$ in $O(\log^2 m)$ time.
\begin{algorithmic}[1]
\State $\phi\gets\sr { \f {54\al}{(\tau-t_0) \vol(V^p_{\ge\tau})} }$, $t_{init} \gets \tau$, $L_{init} \gets \lc\log_{(1+\phi/2)}(2m)\rc$ \label{line:def-phi}
\State $t_+\gets t_{init}$, $L\gets L_{init}$ \Comment{Maintain tuple $(t_+,L)$ in binary search}
\While {$L>1$} \label{line:while}
  \State $t_{mid} \gets t_+-\sum_{i=0}^{\lf L/2\rf-1}\f{18\al}{\phi\,\vol(\V{t_+})(1+\phi/2)^i}$ \label{line:tmid}
  \If {$\vol(\V{t_+})(1+\phi/2)^{\lf L/2\rf}\ge\VV{t_{mid}}$} \label{line:if}%
    \State $L\gets\lf L/2\rf$ \Comment{$(t_+,L) \gets (t_+,\lf L/2\rf)$}
  \Else %
    \State $t_+\gets t_{mid}$, $L\gets\lc L/2\rc$ \Comment{$(t_+,L)\gets(t_{mid},\lc L/2\rc)$}
  \EndIf
\EndWhile
\State \Return $\V{t_+}$ \Comment{\textbf{Guarantee:} $\Phi(\V{t_+})\le O(\phi)$}
\end{algorithmic} \label{alg:binary-search}
\end{algorithm}

\begin{remark}
We will be applying this with $(\tau-t_0 )= \Th(\e)$ and $\vol(V^p_{\ge\tau}) = \Om(1/(\e\log m))$ for some $\e$, so that the $\e$'s cancel in the denominator of $\phi=\sr { \f {54\al}{(\tau-t_0) \vol(V^p_{\ge\tau})} }$, and we get a cut of conductance $O(\sr {\al\log m})$.
\end{remark}

Clearly, the number of queries in the algorithm $O(\log L_{init}) = O(\log (\frac{\log m}{\phi} )) = O(\log m)$ queries. 
The rest of this section is dedicated to proving the promised guarantee:

\BL\label{lem:binary-search-find}
Algorithm~\ref{alg:binary-search} outputs a cut of conductance at most
$\sr { \f {O(\al)}{(\tau-t_0) \vol(V^p_{\ge\tau})} }$.
\EL

We now proceed to the proof of \Cref{lem:binary-search-find}. We will maintain the following invariant throughout the binary search:
\begin{invariant}\label{inv}
The tuple $(t_+,L)$ always satisfies the following. Define
\begin{gather}
t_-:=t_+-\sum_{i=0}^{L-1}\f{18\al}{\phi\vol(V_{\ge t_+}^p)(1+\phi/2)^i} ; \label{eq:inv1}
\end{gather}
then, we must have $t_-\ge t_0$ and
\begin{gather}
\vol(V_{\ge t_+}^p)(1+\phi/2)^L\ge\vol(V_{\ge t_-}^p) . \label{eq:inv2}
\end{gather}
\end{invariant}
There are three items that need to be proven:
\BCL
\Cref{inv} is satisfied at the beginning, for tuple $(\tau_{init},L_{init})$.
\ECL
\BP
We have $\VV{t_+}(1+\phi/2)^{L_{init}} \ge (1+\phi/2)^{L_{init}}\ge2m\ge\VV{t_-}$, satisfying (\ref{eq:inv2}). As for $t_-\ge t_0$, we have
\begin{align*}
t_- &\stackrel{(\ref{eq:inv1})} = \tau - \sum_{i=0}^{L-1}\f{18\al}{\phi\vol(V_{\ge t_+}^p)(1+\phi/2)^i}
\\&\ge \tau - \sum_{i=0}^\infty \f{18\al}{\phi\vol(V_{\ge t_+}^p)(1+\phi/2)^i}
\\&= \tau - \f{18\al}{\phi\vol(V_{\ge t_+}^p)} \cd \f1{1-(1+\phi/2)\inv}
\\&= \tau - \f{18\al}{\phi\vol(V_{\ge t_+}^p)} \cd \f{1+\phi/2}{\phi/2}
\\&\ge\tau - \f{18\al}{\phi\vol(V_{\ge t_+}^p)} \cd \f3\phi =\tau- \f{54\al}{\phi^2\VV{t_+}},
\end{align*}
where the last inequality used that $\phi\le1$, which we can safely assume (otherwise, the guarantee $\Phi(\V{t_+})\le\phi$ of Algorithm~\ref{alg:binary-search} is vacuous). Plugging in $\phi=\sr { \f {54\al}{(\tau-t_0) \vol(V^p_{\ge\tau})} }$, we obtain
$t_- \ge \tau - (\tau - t_0) = t_0$, as desired.
\EP
\BCL\label{clm:guarantee}
Suppose that \Cref{inv} is satisfied at the end, for some tuple $(t_+,1)$. Then, $\Phi(\V{t_+})\le \sr { \f {O(\al)}{(\tau-t_0) \vol(V^p_{\ge\tau})} }$.
\ECL
\BP
Define $\phi:=\sr { \f {54\al}{(\tau-t_0) \vol(V^p_{\ge\tau})} }$, and suppose for contradiction that
\[ \Phi(\V{t_+}) = \f {|\pt\V{t_+}|} { \min\{\VV{t_+},\vol(V\setminus \V{t_+})\} } > 3\phi .\]
Assuming \Cref{inv}, we have $t_+\ge t_-\ge t_0$. Also, $\VV{t_0}\le1.5m$ by assumption, so 
\[ \VV{t_+}\le\VV{t_0}\le1.5m \implies \min\{\VV{t_+},\vol(V\setminus \V{t_+})\}\ge\f13\VV{t_+} .\]
In particular, $|\pt\V{t_+}|\ge\phi\,\VV{t_+}\ge(\phi/3)\,\VV{t_+}$. Let $t_{med}$ be the median expansion at $t_+$. By \Cref{lem:med2}, $t_+-t_{med}\le6\al/|\pt\V{t_+}|\le18\al/(\phi\,\VV{t_+})$. In particular,
\[ t_{med}\ge t_+-18\al/(\phi\,\VV {t_+})\stackrel{(\ref{eq:inv1})}= t_- .\]
Therefore,
\[\VV{t_+}(1+\phi/2)^1 \stackrel{(\ref{eq:inv2})}\ge \VV{t_{-}}\ge \VV{t_{med}} \stackrel{\mathclap{\text{Lem~\ref{lem:med1}}}}\ge \VV{t_+}+|\pt \V{t_+}|/2\]
\[\implies |\pt \V{t_+}|\le\phi\,\VV{t_+} \le 3\phi\min\{\VV{t_+},\vol(V\setminus \V{t_+})\} ,\]
contradicting the assumption that $|\pt\V{t_+}|>3\phi\,\min\{\VV{t_+},\vol(V\setminus \V{t_+})\}$. 
\EP
\BCL
Suppose \Cref{inv} is satisfied before an iteration of the \textup{\textbf{while}} loop. Then, it is still satisfied after that iteration.
\ECL
\BP
Suppose $(t,L)$ is the tuple at the beginning of the iteration, and define $t_-$ as in (\ref{eq:inv1}). First, suppose the \textbf{If} statement (line~\ref{line:if}) is true. In order to prove that \Cref{inv} is maintained, we need to show that for $t'_-:=t_+-\sum_{i=0}^{\lf L/2\rf-1}\f{18\al}{\phi\vol(V_{\ge t_+}^p)(1+\phi/2)^i} $, we have $t'_-\ge t_0$ and $\VV{t_+}(1+\phi/2)^{\lf L/2\rf}\ge\VV{t'_-}$. The former inequality is easy: clearly $t'_-\ge t_-$, and we know $t_-\ge t_0$ since we assumed \Cref{inv} is satisfied for $(t,L)$. For the latter inequality, observe that $t_{mid}=t'_-$ by definition (line~\ref{line:tmid}); therefore, since the \textbf{If} is true,
\[ \VV{t_+}(1+\phi/2)^{\lf L/2\rf} \ge \VV{t_{mid}}=\VV{t'_-} ,\]
as desired.

Now suppose that the \textbf{If} is false, which means that
\begin{gather}
\vol(\V{t_+})(1+\phi/2)^{\lf L/2\rf}\le\VV{t_{mid}}. \label{eq:if-false}
\end{gather}
This time, we define
\begin{align*}
t'_- &:= t_{mid}-\sum_{i=0}^{\lc L/2\rc-1}\f{18\al}{\phi\,\VV{t_{mid}}(1+\phi/2)^i}
\\&\stackrel{(\ref{eq:if-false})}\ge t_{mid} - \sum_{i=0}^{\lc L/2\rc-1}\f{18\al}{\phi\cd \big(\VV{t_+}(1+\phi/2)^{\lf L/2\rf}\big)\cd(1+\phi/2)^i}
\\&=t_{mid} - \sum_{i=\lf L/2\rf}^{L-1} \f{18\al}{\phi\,\VV{t_+}(1+\phi/2)^i}
\\&\stackrel{\mathclap{\text{line~\ref{line:tmid}}}}=\ \ t_+-\sum_{i=0}^{L-1}\f{18\al}{\phi\,\VV{t_+}(1+\phi/2)^i}
\\&\stackrel{\mathclap{(\ref{eq:inv1})}}= t_-.
\end{align*}
Again, this means that $t'_-\ge t_-\ge t_0$. Also, $\VV{t'_-}\le\VV{t_-}$, so
\begin{align*}
\VV{t_{mid}}\cd(1+\phi/2)^{\lc L/2\rc} &\stackrel{(\ref{eq:if-false})}\ge \lp \VV{t_+}(1+\phi/2)^{\lf L/2\rf} \rp \cd (1+\phi/2)^{\lc L/2\rc} 
\\&=\VV{t_+}(1+\phi/2)^L
\\&\stackrel{(\ref{eq:inv2})}\ge \VV{t_-}
\\&\ge\VV{t'_-},
\end{align*}
so \Cref{inv} is satisfied for the tuple $(t_{mid},\lc L/2\rc)$.
\EP

\subsection{Computing Low Conductance Cuts from Large Excess}
\label{sec:cut_from_excess}

\BD[Excess]
Given a PageRank vector $p$, the \emph{excess} at a vertex $v$ is defined as
\[ \exc(p,v):=p(v)-\f{\deg(v)}{2m} .\]
The excess of a vertex set $S\s V$ is
\[ \exc(p,S):=\sum_{v\in S}\exc(p,v)=p(S)-\f{\vol(S)}{2m} .\]
When the PageRank vector $p$ is implicit, we shorten the notations to $\exc(v)$ and $\exc(S)$.
\ED

The main result of this section is the following:
\BL\label{lem:find-cut}
Fix a PageRank vector $p$, and define $S:=\{v\in V:\exc(p,v)/\deg(v)\ge\f1{100m}\}$. (Note that $S=\V{1/(2m)+1/(100m)}$.) Suppose $S$ satisfies the following two properties:
\BE
\im $\exc(S) \ge 1/10$ %
\im $\vol(S) \le 1.5m$ %
\EE
Then, there is an algorithm that makes a single call to Algorithm~\ref{alg:binary-search} and outputs a cut of conductance $O(\sr{\al\log m})$.
\EL

\BP
Partition the vertices of $S$ into $O(\log m)$ ``excess buckets'' as follows: for each positive integer $i \le  \log_2(50m) $, form a  bucket $B_{2^{-i}}:=\{v\in S:\exc(v)/\deg(v) \in ( 2^{-i} , 2^{-i+1} ] \}$.

Define $S^-:=\{v\in S:\exc(v)/\deg(v)\le\f1{25m}\}$. Observe that 
\[ \exc(S^-)= \sum_{v\in S^-} \exc(v) \le \sum_{v\in S^-}\deg(v) \cd \f1{25m} \le 2m \cd \f1{25m}=\f1{12.5} ,\]
so $\exc(S\setminus S^-) \ge 1/10-1/12.5=1/50$. Also, all vertices in $S\setminus S^-$ belong to exactly one bucket, which means
\[ \sum_{i=1}^{\lf\log_2(50m)\rf}\exc(B_{2^{-i}}) \ge \exc(S\setminus S^-)=\f1{50}. \]
In particular, there exists a bucket $B_\e$ ($\e=2^{-i}$ for some $i\le \log_2(50m)$) with $\exc(B_\e)\ge\f1{50\lf\log_2(50m)\rf}\ge\f1{50\log_2(50m)}$. Since each vertex $v\in B_\e$ satisfies $\exc(v)/\deg(v)\le\e\iff\deg(v)\ge\exc(v)/\e$, we have
\[ \vol(B_\e)=\sum_{v\in B_\e}\deg(v)\ge\sum_{v\in B_\e}\f1\e\,\exc(v)=\f1\e\cd\exc(B_\e)\ge\f1{50\e\log_2(50m)}= \Om\lp\f1{\e\log m}\rp .\]
Also, snce $B_\e \s \V\e$, we have $\VV\e=\Om(1/(\e\log m))$ as well.

Set $t_0\gets\f1{2m}+\e/2$ and $\tau \gets\f1{2m}+ \e$. \footnote{Recall that $V_{\ge t}$ is defined for \emph{densities}, whereas we have \emph{excesses}, hence the additional $1/(2m)$ term.} Since $\e\ge1/(50m)$, we have $t_0 \ge \f1{2m}+\f1{100m}$, and since $\vol(\V{1/(2m)+1/(100m)})\le1.5m$, the guarantee $\VV{t_0}\le1.5m$ of \Cref{lem:sweep-cut} is satisfied.  Finally, invoking Algorithm~\ref{alg:binary-search} with this $t_0$ and $\tau$ gives a cut of conductance
\[ \sr { \f {18\al}{(\tau-t_0) \vol(V^p_{\ge\tau})} } \le \sr { \f {18\al}{(\e/2) \cd\Om(1/(\e\log m))} } = O\lp\sr{\al\log m}\rp .\]
\EP

\subsection{Guarantee of Large Excess from Existence of Low Conductance Cuts}
\label{sec:large_excess}

Here, we justify the assumptions (1) and (2) of \Cref{lem:find-cut}, assuming that a low-conductance cut $S$ exists in the graph. We use the theorem below of \cite{AndersenCL06}.

\BT[Theorem~4 of {\cite{AndersenCL06}}]\label{thm:pagerank-start}
For any set $C$ with $\vol(C)\le m$ and any constant $\al\in(0,1]$, there is a subset $C_\al\s C$ with volume $\vol(C_\al)\ge\vol(C)/2$ such that for any vertex $v\in C_\al$, the PageRank vector $p=p(\chi_v)$ satisfies $p(C)\ge1-2\Phi(C)/\al$.
\ET

For the application in \cite{AndersenCL06}, the guarantee $\vol(C_\al)\ge\vol(C)/2$ is important, since it means that a random vertex from $C$, weighted by degree, is in $C_\al$ with probability $1/2$. In contrast, since we are in the deterministic setting, we only need $C_\al$ to be nonempty, which \Cref{thm:pagerank-start} guarantees.

\BCL\label{clm:conditions}
Let $C$ be a cut of conductance $\phi$ in $G$ ($\vol(C)\le m$), let $\al\ge400\phi$, and fix an arbitrary vertex $s\in C_\al$. 
For PageRank vector $p=p(x_s)$, Conditions~(1)~and~(2) of \Cref{lem:find-cut} is satisfied.
\ECL
\BP
As in \Cref{lem:find-cut}, define $S:=\{v\in V:\exc(p,v)/\deg(v)\ge\f1{100m}\}$. 

We first prove Condition~(1). By \Cref{thm:pagerank-start}, $p(C)\ge1-2\Phi(C)/\al\ge1-1/200$, so
\[ \exc(C) = \sum_{v\in C}\lp p(v)-\f{\deg(v)}{2m}\rp = p(C) - \f{\vol(C)}{2m} \ge \lp1-\f1{200}\rp-\f12\ge\f13 .\]
Let $C^- := \{ v\in C:\exc(v)/\deg(v)<\f1{100m} \} = C\setminus S$. Observe that
\[ \exc(C^-)=\sum_{v\in C^-} \exc(v) \le \sum_{v\in C^-}\deg(v)\cd\f1{100m}\le\f{2m}{100m}=\f1{50} .\]
Since $C\setminus C^-=C\cap S\s S$, we have
\[ \exc(S) \ge \exc(C\setminus C^-) = \exc(C)-\exc(C^-)\ge\f13-\f1{50}\ge\f1{10}, .\]
proving Condition~(1).

For Condition~(2), let $C^+:=\{v\notin C:\exc(v)/\deg(v)\ge\f1{100m}\} = S\setminus C$. By \Cref{thm:pagerank-start}, $p(C^+)\le p(V\setminus C)\le1/200$, so
\[ \vol(C^+) =  \sum_{v\in C^+}\deg(v) \le \sum_{v\in C^+} \exc(v) \cd 100m \le \sum_{v\in C^+}p(v)\cd100m \le \f1{200}\cd100m\le\f12m .\]
Therefore $\vol(S) \le \vol(C \cup C^+) = \vol(C)+\vol(C^+)\le m+m/2$, proving Condition~(2).
\EP

\subsection{Merging Low Conductance Cuts into a Near-Expander Cut}
\label{sec:near_expander_cut}

\BL\label{lem:compute-pr}
Given a real number $\al\in(0,1]$ and $n$ initial vectors $\bv_1,\lds,\bv_n$, we can compute the $n$ PageRank vectors $PR(\bv_1),\lds,PR(\bv_n)$ in time $O(n^\om)$.
\EL
\BP
We have
\[ PR(\bv_i) \stackrel{(\ref{eq:pr})}= \al \bv_i (I-(1-\al)W)\inv ,\]
where $I-(1-\al)W$ is guaranteed to be invertible by the uniqueness statement in Fact~\ref{fact:pr}. Let $\mathbf V$ be the matrix with $\bv_i$ as the $i$'th row; then, $PR(\bv_i)$ is simply the $i$'th row of $\al \mathbf V (I-(1-\al)W)\inv$. Since matrix inversion and matrix multiplication can be computed in $O(n^\om)$ time, the lemma follows.
\EP

The algorithm computes PageRank starting at the $n$ vectors $\{ \chi_v : v\in V\}$. Then, for each vertex $v\in V$ and each of the $O(\log m)$ buckets $B_\e$ (for that $v$),

\begin{algorithm}[H]
\caption{MostBalancedEdgeCut$(G,\phi)$}
Assumption: $\Phi(G) < \phi$
\\Output: a $\big(\f{\phi}{\log_2(50m)}\big)$-near-expander $(\sqrt{\phi\log m})$-conductance cut
\\Runtime: $O(n^\om)$
\begin{algorithmic}[1]
\State Compute the $n$ PageRank vectors $PR(\chi_v)$ for each $v\in V$ \Comment{$O(n^\om)$ time by \Cref{lem:compute-pr}}
\State Compute the degree $\deg(v)$ for each $v \in V$ \Comment{$O(n^2)$ time}
\State $\m C \gets \emptyset$ \Comment{$C\s2^V$ will be a collection of low-conductance cuts}
\For {each $p_v\gets PR(\chi_v)$}
  \State Create a data structure that, given $t$, can return $\vol(V^{p_v}_{\ge t})=\sum_{u\in V^{p_v}_{\ge t}}\deg(u)$ in $O(\log n)$ time. \Comment{$O(n\log n)$ time using balanced binary search trees}
  \State Form the $O(\log m)$ excess buckets as in \Cref{lem:find-cut} with $\al = 400 \phi$ %
  \State $B_\e\gets$ some bucket with $\exc(B_e)\ge1/(50\log_2(50m))$
    \State Call Algorithm~\ref{alg:binary-search} on $(G,p_v,1/(2m)+\e/2,1/(2m)+\e)$, which returns a set $C$ ($\vol(C) \le 1.5m$) of conductance $O(\sr{\al\log m})=O(\sr{\phi}\log m)$  \label{line:find-C}
    \State Add $C$ to $\m C$
\EndFor
\State $S\gets\emptyset$\Comment{$S\s V$ will be the near-expander conductance cut}
\For {each cut $C\in\m C$ in arbitrary order}
  \If {$\vol(C\setminus S)\ge\vol(C)/2$}
    \State $S\gets S\cup C$ \label{line:add-C}
  \EndIf
  \If {$\vol(S)\ge m/4$}
    \State \textbf{break} \label{line:break} \Comment{Exit the \textbf{for} loop}
  \EndIf
\EndFor
\State\Return $S$ or $V\setminus S$, whichever has smaller volume
\end{algorithmic} \label{alg:near-expander-cut}
\end{algorithm}

\BT
Algorithm~\ref{alg:near-expander-cut} returns a $\big(\f{\phi}{\log_2(50m)}\big)$-near-expander $O(\sr{\phi\log m})$-conductance cut.%
\ET
\BP
Let $\phi' \le O(\sr{\phi\log m})$ be an upper bound to the conductance of any cut $C$ from line~\ref{line:find-C}. 
We first claim that at all times of the algorithm, $|\pt S|\le 2\phi'\vol(S)$. Intuitively, this is because every time we add a set $C$ to $S$ (line~\ref{line:break}), the new edges in $\pt S$ can be ``charged'' to the new volume $\vol(C\setminus S)$, which is always at least $\vol(C)/2$.
More formally, $|\pt S|\le2\phi'\vol(S)$ is clearly satisfied initially with $S=\emptyset$, and whenever a new set $C$ is added (line~\ref{line:add-C}),
\[ |\pt (S\cup C)| \le |\pt S|+|\pt C| \le 2\phi'\vol(S) + \phi'\vol(C) \le 2\phi'\vol(S)+2\phi'\vol(C\setminus S)=2\phi'\vol(S\cup C) .\]

Suppose that line~\ref{line:break} is reached in the algorithm. Let $S'$ be the set $S$ before the last cut $C$ was added to it (line~\ref{line:add-C}). Since $\vol(S')\le m/4$, we must have $\vol(S)=\vol(S'\cup C)\le\vol(S')+\vol(C)\le 1.5m+m/4=2m-m/4$. This means that
\[ \min\{\vol(S),\vol(V\setminus S)\}\ge\f17\vol(S) .\] Therefore, the cut $S$ that is output satisfies
\[ \Phi(S)=\f{|\pt S|}{\min\{\vol(S),\vol(V\setminus S)\}} \le \f{|\pt S|}{\vol(S)/7}\le14\phi' .%
\]
This, along with the fact that $\min\{\vol(S),\vol(V\setminus S)\}\ge m/4$, shows that the algorithm outputs a near-expander $14\phi'$-conductance cut, as desired.

Now suppose that line~\ref{line:break} is never reached. We know that $\vol(S)<m/4$, so the algorithm returns $S$ (and not $V\setminus S$). Suppose for contradiction that $S$ is not $\big(\f{\phi}{\log_2(50m)}\big)$-near-expander $(14\phi')$-conductance cut. By \Cref{def:near-expander}, this can only happen if $V\setminus S$ is not a nearly $\f{\phi}{\log_2(50m)}$-expander, which means there exists $T\s V\setminus S$ with $\vol(T)\le\vol(V\setminus S)/2$ and $|E(T,V\setminus T)|<\f{\phi}{\log_2(50m)}\vol(T)$. Our goal is to show that there exists some $C\in\m C$ that should have been added to $S$ in line~\ref{line:add-C}, a contradiction.

As $\alpha = 400 \phi$, by \Cref{thm:pagerank-start} applied to $\al$ and $T$, there is a vertex $t\in T$ such that if we start PageRank at $t$, we have $p(T)\ge1-2\Phi(T)/\al$. This means that
\begin{gather}
\exc(V\setminus T)\le p(V\setminus T)\le\f{2\Phi(T)}\al \le \f{2\cd\phi/\log_2(50m)}{400\phi}=\f1{200\log_2(50m)} .\label{eq:volC1}
\end{gather}
By \Cref{clm:conditions} applied to $T$, the conditions of \Cref{lem:find-cut} are satisfied. Following the proof of \Cref{lem:find-cut}, there exists a bucket $B_e$ with $\exc(B_\e)\ge\f1{50\log_2(50m)}$ and $\vol(B_\e)\ge\f1\e\,\exc(B_\e)\ge\f1{50\e\log_2(50m)}$. Since \Cref{alg:binary-search} is called with $\tau\gets\f1{2m}+\e$, the cut $C$ returned satisfies
\begin{gather}
\vol(C)\ge\VV{1/(2m)+\e}\ge\vol(B_\e)\ge\f1{50\e\log_2(50m)} \label{eq:volC2}
\end{gather} (where $p$ is the relevant PageRank vector). Also, since \Cref{alg:binary-search} is called with $t_0\gets\f1{2m}+\e/2$, all vertices $v$ in the returned cut $C$ satisfy $\exc(v)\ge\e/2$. Therefore,
\[ \vol(V\setminus T) \le \f2\e\cd\exc(V\setminus T) \stackrel{(\ref{eq:volC1})}\le \f2\e\cd\f1{200\log_2(50m)}=\f1{100\e\log_2(50m)} \stackrel{(\ref{eq:volC2})}\le \f{\vol(C)}2. \] Therefore,
\[ \vol(C\setminus S) \ge \vol(C\cap T)\ge\vol(C) - \vol(V\setminus T)\ge\vol(C)/2,\]
which means that $C$ should have been added to $S$ in line~\ref{line:add-C}, contradiction.
\EP
Thus, Algorithm~\ref{alg:near-expander-cut} achieves the guarantees of \Cref{thm:main-w}.
Apply \Cref{lem:ne-mb} with $\phi_1=\Th(\sr{\phi\log m})$, $\phi_2=\Th(\phi/\log(m))$, and $c:=\Th(\phi^{-0.5}\log^{1.5} m)$
then gives the guarantees of the dense approximate balanced
cut algorithm from \Cref{thm:pagerank-main}.

\section{Most-Balanced Low-Conductance Cuts on Sparse Graphs}
\label{sec:madry}

Our next goal is to speed up the previous
algorithm on sparse graphs.

\MadryMain*

As before in \Cref{sec:pagerank}, for explanation purposes,
it is more convenient to assume that there exists
$S^{*}$ where $\Phi(S^{*})\le1/\poly\log n$ and $\vol(S^{*})=\Omega(m)$,
and our goal is find a $(1/\log n)$-conductance cut $S$ where $\vol(S)=\Omega(m)$.

The main technique we use here is the $j$-trees
by Madry \cite{Madry10}.
A $j$-tree is a graph where consisting of two parts:
\begin{enumerate}
        \item a core $K$ which contains at most $j$ vertices, and
        \item a forest $F$ such that for each tree
$T\in F$, $|V(T)\cup V(K)|=1$.
\end{enumerate}
Intuitively, $j$-trees are graphs with $j$ vertices that have
a forest ``attached'' to it.
Although $j$-trees are very restricted form of graphs,
Madry~\cite{Madry10} shows that a collection of $j$-trees
can approximate an arbitrary graph in the following sense: 
\begin{fact}
        [\cite{Madry10} Paraphrased]
\label{fact:Madry10InANutshell}
        There is an deterministic algorithm that,
        given an $m$-edge graph $G$ and a parameter $t$, runs in $\tilde{O}(mt)$
        time and outputs $t$ many $\tilde{O}(m/t)$-trees $G_{1},\dots,G_{t}$
        such that, for any $C\subset V$
        \begin{itemize}
                \item $|E_{G}(C,V-C)|\le|E_{G_{i}}(C,V-C)|$ for all $i$, and
                \item $|E_{G_{i}}(C,V-C)|\le\alpha|E_{G}(C,V-C)|$ for some $i$ where $\alpha=O(\log^{3}n)$.
        \end{itemize}
\end{fact}

To discuss the main idea, it is more convenient to use a notion of
\emph{sparsity} of cuts.%
Namely, the sparsity of a cut $S$ is $\sigma(S)=\frac{E(S,V-S)}{\min\{|S|,|V-S|\}}$.
A \emph{$(\phi^{*},c)$-most-balanced $\phi$-sparse cut $S$}
is such that $\sigma(S)\le\phi$ and $|S|\ge|S^{*}|/c$ where $S^{*}$
is the set with maximum $|S^{*}|$ out of all sets $S'$ where $\sigma(S')\le\phi^{*}$.
In this language, our goal is to find a cut with
$\phi^{*}=O(\phi^{2}/\poly\log m)$ and $c=\tilde{O}(1)$.

Setting $C$ in \Cref{fact:Madry10InANutshell} above to $S^{*}$ 
gives that there is some $i$ such that
$|E_{G_{i}}(S^{*},V-S^{*})|\le\alpha|E_{G}(S^{*},V-S^{*})|$.
So $S^{*}$ is has sparsity $\sigma_{G_{i}}(S^{*})\le\alpha\phi^{*}$.
Let $S_{i}$ be an $(\alpha\phi^{*},c)$-most-balanced $\phi$-sparse
cut in $G_{i}$.
That is, $\sigma_{G_{i}}(S_{i})\le\phi$ and $|S_{i}|\ge|S^{*}|/c$.
As $|E_{G}(S_{i},V-S_{i})|\le|E_{G_{i}}(S_{i},V-S_{i})|$, we have
\[
\sigma_{G}\left(S_{i}\right)\le\phi.
\]
That is, $S_{i}$ is in fact a $(\phi^{*},c)$-most-balanced
$\phi$-sparse in $G$.
This argument shows that it suffices to compute an 
$(\alpha\phi^{*},c)$-most-balanced $\phi$-sparse cut in each
$j$-tree $G_{i}$ where $j=\tilde{O}(m/t)$.

To compute such cut in a $j$-tree $G_{i}$, we will use different
approaches on the core $K$ and the forest $F$:
\begin{itemize}
        \item The forest $F$ of $G_{i}$ will be computed using
        greedy in linear time.
        \item We then apply the dense graph algorithm from
        \Cref{sec:pagerank} to the core.
\end{itemize}
The main obstacle is that the dense graph algorithm
deals with low-conductance instead of low-sparsity cuts.
These two notions may be different in $G_{i}$ because $G_{i}$
might not have constant degree.
Also, it is non-trivial how one can combine the two solution from
two parts together.
Thus, the bulk of our technical contribution in this section it to show
that computing approximate balanced cuts on the forests
and core separately is sufficient.

\BT\label{thm:balanced_in_sparse_graphs}
There is a deterministic $(\ot(\Delta^{3/2} \phi^{1/2}), c)$-approximate balanced
cut algorithm for graphs with maximum degree $\Delta$ that runs in
$O(m^{1.579})$ time.
\ET

\BD[Multi-forest, multi-tree]
A multigraph $G=(V,E)$ is a \emph{multi-forest} if the support graph $\{(u,v):$ there exists edge $e\in E$ with endpoints $u,v\}$ is a forest. Similarly, a multigraph is a \emph{multi-tree} if its support graph is a tree. We represent multiforests and multitrees in $\tO(n)$ space by storing the number of multi-edges in $E$ for each edge in the support graph.
\ED

\BD[$j$-tree]
Given a multigraph $G=(V,E)$, a \emph{$j$-tree} of $G$ is a graph whose edges are the union of (i) a graph on a vertex set $K$ of size $j$, and (ii) a multi-forest $F$ with no edge between two vertices in $K$ ($F$ can contain edges not in $G$). The set $K$ is called the \emph{core} of the $j$-tree. (The core may not be unique.)
\ED

\BT[\cite{Madry10}]\label{thm:madry}
For any multigraph $G=(V,E)$ and parameter $t$, there exists a distribution $\m D$ on $t$ multigraphs $G_1,\lds,G_t$ of $\tO(m/t)$-trees of $G$ such that for any subset $C\s V$:
 \BE
 \im For all $i\in[t]$, $|E_G(C,V\sm C)| \le |E_{G_i}(C,V\sm C)|$
 \im $\E_{G_i\sim\m D}[| E_{G_i}(C,V\sm C) |] \le \al\, |E_G(C,V\sm C)|$ for $\al=O(\log^2 n \log\log n)$
 \EE
Moreover, we can find such a distribution in $\tO(|E|\cd t)$ time.
\ET
The only randomized procedure in \Cref{thm:madry} is the construction of low-stretch spanning trees, but this can be replaced with a deterministic construction with slightly worse parameters~\cite{ElkinEST08}. In other words, \Cref{thm:madry} can be made entirely deterministic with the same guarantees up to $O(\log^2 n \log\log n)$ factors.
\BD[Canonical $j$-tree]\label{def:canon}
Fix a multigraph $G=(V,E,w)$ and a parameter $j$, and fix a vertex set $K\s V$ of size $j$ and an unweighted forest $F$ in $V$ with no edge between two vertices in $K$ ($F$ can contain edges not in $G$). A \emph{canonical $j$-tree $H$ with core $K$ and forest $F$} is a $j$-tree constructed as follows: For each vertex $v\in V$, consider the tree $T$ in $F$ containing $v$, and let $r(v)$ be the (unique) vertex of $T$ that is also in $K$ (the ``root'' of $T$). For each edge $e=(u,v)\in E$:
\BE
\im If $r(u)=r(v)$, then for each edge $e'$ in the path from $u$ to $v$ in $F$, add an edge $e'$ to $H$.
\im Otherwise, $r(u)\ne r(v)$. For each edge $e'$ on the path from $u$ to $r(u)$ in $F$, add an edge $e'$ to $H$. Do the same for each edge $e'$ on the path from $v$ to $r(v)$ in $F$. Finally, add the edge $(r(u),r(v))$ to $H$.
\EE
\ED

\BL\label{lem:canon}
Fix a graph $G=(V,E)$ and a parameter $j$, and let $H$ be a $j$-tree of $G$ with core $K$ such that $G \le H\le\al\,G$. Then, there exists a canonical $j$-tree $H'$ with core $K$ such that $G \le H'\le  H$. Furthermore, given the core $K$ of $H$, we can compute in $\tO(|V|+|E|)$ time a weighted graph $H''$ where each edge $(u,v)\in H''$ has weight equal to the number of (parallel) edges between $u$ and $v$ in $H'$. 
\EL
\BP
The direction $H'\ge G$ follows from the embedding of the canonical $j$-tree, so we focus on the other direction $H'\le H$. Consider an embedding of $G$ into $H$; we will use this embedding to construct an embedding from $H'$ into $H$. Let $F$ be the forest of $H$, and for each vertex $v\in V$, let $r(v)$ be defined as in \Cref{def:canon} for $K$ and $F$. For each edge $e=(u,v)$ in $G$, by the $j$-tree structure of $H$, the embedded path $P_e$ of $e$ in $H$ must consist of either (1) the path from $u$ to $v$ in $F$ if $r(u)=r(v)$, or (2) the path from $u$ to $r(u)$ in $F$, the path from $v$ to $r(v)$ in $F$, and some path from $r(u)$ to $r(v)$ in $H$. Now consider the edges in $H'$ that resulted from edge $e$ in $G$. In case (1), for each edge $e'$ on the path from $u$ to $v$ in $F$, we added an edge $e'$ in $H'$; embed each edge $e'$ onto the edge in $P_e$ with the same endpoints of $e'$. This gives an embedding of the edges in $H'$ constructed by $e$ onto the edges of $P_e$. In case (2), we can do the same for the edges $e'$ in the path from $u$ to $r(u)$ and from $v$ to $r(v)$. For the edge $(r(u),r(v))$ in $H'$, we embed it along the path from $r(u)$ to $r(v)$ in $P_e$. Altogether, we also embed every edge constructed by $e$ onto the edges of $P_e$. Doing this for every $e$ in $G$, we obtain an embedding of $H'$ into $H$.

The computation of $H''$ is straightforward. First, for each edge $(u,v)\in G$ with $r(u)=r(v)$, we can compute the lowest common ancestor $a(u,v)$ of $u$ and $v$ in the corresponding tree in $F$ rooted at $r(u)=r(v)$~\cite{aho1976finding}. For each edge $(u,v)$ with $r(u)\ne r(v)$, we can find the vertices $r(u)$ and $r(v)$. Next, for each tree in $F$, we can perform a simple traversal to determine, for each edge $e$ in the tree, the number of pairs $(u,r(u))$ or $(u,a(u,v))$ whose path between the two vertices passes through $e$. Lastly, we can easily keep track of the edges $(r(u),r(v))$ for each edge $(u,v)$ with $r(u)\ne r(v)$.
\EP
From \Cref{thm:madry} and \Cref{lem:canon}, we obtain the following corollary:
\BC
For any multigraph $G=(V,E)$ and parameter $t$, there exists a distribution on $t$ graphs $G_1,\lds,G_t$ of $\tO(n/t)$-trees of $G$ such that for any subset $C\s V$:
 \BE
 \im For all $i\in[t]$, $w(\pt_{G_i}C) \le w(\pt_GC)$
 \im $\E_i[w(\pt_{G_i}C)] \le \al\, w(\pt_GC)$
 \EE
Moreover, we can find such a distribution in $\tO(|E|\cd t)$ time.
\EC
\BP
Apply \Cref{thm:madry} to find the $j$-trees $G_1,\lds,G_t$, and then apply \Cref{lem:canon} onto each one. The two properties hold by the properties of an embedding.
\EP

\subsection{Most-Balanced Sparse Cut on a $j$-tree}

We will actually compute most-balanced \emph{sparse} cut, defined as follows:
\BD[Most-Balanced Sparse Cut]\label{def:most-balanced-s}
Given $G=(V,E)$ and parameters $\phi,\phi^*,c$, a set $S\s V$ with $\vol(S)\le m$ is a \emph{$(\phi^*,c)$-most-balanced $\phi$-sparse cut} if it satisfies:
 \BE
 \im $\sigma(S) \le \phi$.
 \im Let $S^*\s V$ be the set with maximum $\vol(S^*)$ out of all sets $S'$ satisfying $\sigma(S')\le\phi ^*$ and $|S'|\le n/2$. Then, $|S|\ge |S^*|/c$.
 \EE
\ED

\begin{algorithm}[H]
\caption{MostBalancedSparseCut$(G,H,K,F,\phi)$}
Assumption: $H=(V,E)$ is a $|K|$-tree of $G$ with core $K$ and multi-forest $F$.
\\Output: 
\\Runtime: $\tO(m+|K|^\om)$
\begin{algorithmic}[1]

\State $H_K\gets H[K]$ with the following additional \emph{self-loops}: for each edge $(u,v)$ in $G$, if $r(u)=r(v)$, then add a self-loop at $r(u)\in K$.\label{line:self}

\State $S_K\gets$ (recursive) $(f(\phi),\be)$-approximate most-balanced low-conductance cut on $H_K$ \label{line:a1}

\State Construct a vertex-weighted multi-tree $T=(V_T,E_T,w_T)$ as follows: Starting with $H$, contract $K$ into a single vertex $k$ with weight $|K|$. All other vertices have weight $1$. (The vertices of $T$ have total weight $n$.) \label{line:a2}
\State Root $T$ at a vertex $r\in V_T$ such that every subtree rooted at a child of $r$ has total weight at most $n/2$. \Comment{A subtree of $T$ rooted at $u$ is the set of vertices $v\in V$ whose path to $r$ includes $u$.}\label{line:a3}
\State $S_{T}\gets $ RootedTreeMostBalancedSparseCut$(T,r,\phi)$ (Algorithm~\ref{alg:tree})\label{line:a4}

\State It is guaranteed that one of $S_K$ and $S_T$ has conductance $\le O(\De\sr{\phi\log m})$. Of the (one or two) cuts satisfying this property, output the one with highest volume. \label{line:a5}

\end{algorithmic} \label{alg:most-balanced}
\end{algorithm}

\begin{algorithm}[H]
\caption{RootedTreeMostBalancedSparseCut$(T=(V,E,w),r,\phi)$}
Assumption: $T$ is a weighted multi-tree with weight function $w:V\to\N$ (so that $w(v)$ is the weight of vertex $v\in V$). The tree is rooted at a root $r$ such that every subtree $V_u$ rooted at a vertex $u\in V\sm\{r\}$ has total weight $w(V_u)\le w(V)/2$.
\\Output: a set $S\s V$ satisfying the conditions of \Cref{lem:tree-alg}.
\\Runtime: $\tO(n^\om)$
\begin{algorithmic}[1]

\State Find all vertices $u\in V\sm\{r\}$ such that if $V_u$ is the vertices in the subtree rooted at $u$, then $w(E[V_u,V\sm V_u])/|V_u|\le2\phi$. Let this set be $X$.
\State Let $X^\uparrow$ denote all vertices $u\in X$ without an ancestor in $X$ (that is, there is no $v\in X\sm\{u\}$ with $u\in T_v$).
\State Starting with $S=\emptyset$, iteratively add the vertices $V_u$ for $u\in X^\uparrow$. If $w(S)\ge n/4$ at any point, then terminate immediately and output $S$. Otherwise, output $S$ at the end.\label{line:t3}

\end{algorithmic} \label{alg:tree}
\end{algorithm}

\BL\label{lem:tree-alg}
Algorithm~\ref{alg:tree} can be implemented to run in $O(|V|)$ time.
The set $S$ output by Algorithm~\ref{alg:tree} satisfies $|E[S,V\sm S]|/\min\{w(S),w(V\sm S)\}\le6\phi$.
Moreover, for any set $S^*$ with $|E[S^*,V\sm S^*]|/w(S^*)\le \phi$ and  $w(S^*)\le2w(V)/3$, and which is composed of vertex-disjoint subtrees rooted at vertices in $T$, we have $\min\{w(S),w(V\sm S)\}\ge w(S^*)/3$.
\EL
\BP
Clearly, every line in the algorithm can be implemented in linear time, so the running time follows. We focus on the other properties.

Every set of vertices $V_u$ added to $S$ satisfies $|E[V_u,V\sm V_u]|/w(V_u)\le2\phi$. Also, the added sets $V_u$ are vertex-disjoint, so $|E[S,V\sm S]| = \sum_{V_u\s S}|E[V_u,V\sm V_u]|$. This means that Algorithm~\ref{alg:tree} outputs $S$ satisfying $|E[S,V\sm S]|/w(S)\le2\phi$. Since every set $V_u$ has total weight at most $w(V)/2$, and since the algorithm terminates early if $w(S)\ge w(V)/4$, we have $w(S)\le 3w(V)/4$. This means that $\min\{w(S),w(V\sm S)\}\ge w(S)/3$, so $|E[S,V\sm S]|/\min\{w(S),w(V\sm S)\} \le 3|E[S,V\sm S]|/w(S)\le6\phi$.

Suppose first that the algorithm terminates early. Then, as argued above, $\min\{w(S),w(V\sm S)\}\ge w(V)/4$, which is at least $(2w(V)/3)/3\ge w(S^*)/3$, so  $\min\{w(S),w(V\sm S)\}\ge w(S^*)/3$.

Now suppose that $S$ does not terminate early. Let $S^*_1,\lds,S^*_\el$ be the vertices in the (vertex-disjoint) subtrees that together compose $S^*$, that is, $\bigcup_iS^*_i=S^*$. Note that $E[S^*_i,V\sm S^*_i]$ is a single edge in $E$ for each $i$. Suppose we reorder the sets $S^*_i$ so that $S^*_1,\lds, S^*_q$ are the sets that satisfy $|E[S^*_i,V\sm S^*_i]|/w(S^*_i)\le2\phi$. Since $|E[S^*,V\sm S^*]|/w(S^*)\le \phi$, by a Markov's inequality-like argument, we must have $\sum_{i\in[q]}w(S^*_i) \ge (1/2)\sum_{i\in[\el]}w(S^*_i)=w(S^*)/2$. Observe that by construction of $X^\uparrow$, each of the subsets $S^*_1,\lds, S^*_q$ is inside $V_u$ for some $u\in X^\uparrow$. Therefore, the set $S$ that Algorithm~\ref{alg:tree} outputs satisfies $w(S)\ge\sum_{i\in[q]}w(S^*_i)\ge w(S^*)/2$. The bound on $|E[S,V\sm S]|/\min\{w(S),w(V\sm S)\}$ follows as before.
\EP

\BL\label{lem:alg-balanced-r}
Algorithm~\ref{alg:most-balanced} runs in time $O(m)$ plus the recursive call (line~\ref{line:a1} of Algorithm~\ref{alg:most-balanced}).
\EL
\BP
Lines~\ref{line:a2}~and~\ref{line:a3}, can be easily implemented in linear time. By \Cref{lem:tree-alg}, line~\ref{line:a4} also takes linear time.
\EP

\BCL\label{clm:6.8}
Let $H$ be a $|K|$-tree of a connected graph $G$ with core $K$ and multi-forest $F$, and let $\De$ be the maximum degree in $G$.
Fix a subset $S\s K$, and let $S_F$ be the vertices in the trees in $F$ intersecting $S$ (note that $S_F\supseteq S$). Then,
\BE
\im $|E_{H_K}(S,V_K\setminus S)| = |E_H(S_F,V\setminus S_F)|$, and
\im $|S_F|\le\vol_{H_K}(S)\le\De\cd|S_F|$.
\EE
\ECL

\BP
For (1), observe that every edge in $E_H(S_F,V\sm S_F)$ must belong in $H[K]$: the only difference in the edges of $H[K]$ and $H_K$ are self-loops, which never appear in $E_{H_K}(S,V_K\sm S)$;
this proves property (1).

For (2), we first show the $|S_F|\le\vol_{H_K}(S)$ direction. For a given vertex $u\in S_F$, let $e=(u,v)$ be an edge in $G$ incident to $u$, which must exist since $G$ is connected. By the construction of $H_K$, $e$ corresponds to the edge $(r(u),r(v))$ in $H_K$: either $r(u)\ne r(v)$ and the edge $(r(u),r(v))$ was added in the construction of $H$ (see \Cref{def:canon}), or $r(u)=r(v)$ and it was added in line~\ref{line:self} of Algorithm~\ref{alg:most-balanced}. Let us charge the vertex $u$ to the endpoint $r(u)$ of the edge $(r(u),r(v))$ in $H_K$. Since no endpoint of any edge is charged more than once, and since the number of endpoints in $S$ of edges in $H[K]$ is exactly $\vol_{H_K}(S)$, we have $|S_F|\le\vol_{H_K}(S)$.

We now show the $\vol_{H_K}(S)\le\De\cd|S_F|$ direction. Consider an endpoint $u\in S$ of edge $(u,v)$ in $H_K$ (possibly a self-loop). This edge resulted from an edge $(u',v')$ in $G$ with $r(u')=u$ and $r(v')=v$; note that $u'\in S_F$. Let us charge the endpoint $u$ of edge $(u,v)$ in $H_K$ to the endpoint $u'\in S_F$ of edge $(u',v')$ in $G$. Every vertex in $|S_F|$ can be charged at most $\De$ times, since its degree in $G$ is at most $\De$. Therefore, the quantity $\vol_{H_K}(S)$, which equals the number of endpoints $u\in S$ of edges in $H_K$, is at most $\De\cd|S_F|$.
\EP

The next lemma shows that we can transform the optimal cut $S^*\s V$ into either a cut $S^*_T\s V_T$ or a cut $S^*_K\s H_K$ without losing too much in sparsity and volume. Note that $S^*_T$ corresponds to a cut in $H$ that only cuts the forest edges in $F$, and $S^*_K$ corresponds to a cut in $H$ that only cuts edges in the core $H[K]$. Then, in the proof of \Cref{lem:alg-balanced}, we will show that the cuts $S_T$ and $S_K$ computed by Algorithm~\ref{alg:most-balanced} approximate $S^*_T$ and $S^*_K$ respectively.

\BL\label{lem:either}
Let $H$ be a $|K|$-tree of a connected graph $G=(V,E)$ with core $K$ and multi-forest $F$, and let $S^*\s V$ be any cut with $|S^*|\le n/2$ (where $n:=|V|$). Let $T$ and $H_K$ be defined as in Algorithm~\ref{alg:most-balanced}. One of the following must hold:
\BE
\im There exists a cut $S^*_T$ in $T$ satisfying $|E_T(S^*_T,V\sm S^*_T)|\le|E_H(S^*,V\sm S^*)|$ and $|S^*|/2\le w( S^*_T)\le 2n/3$, and $S^*_T$ is the disjoint union of subtrees of $T$ rooted at $r$.
\im There exists a cut $S^*_K$ in $H_K$ satisfying $|E_{H_K}(S^*_K,K\sm S^*_K)|\le|E_H(S^*,V\sm S^*)|$ and \linebreak $\min \{ \vol_{H_K}(S^*_K) , \vol_{H_K}(K\sm S^*_K)\} \ge|S^*|/3$.
\EE
\EL
\BP
Let $S^*\s V$ the set as described in \Cref{def:most-balanced-s} ($\sigma_H(S^*)\le\phi$). Let $U$ be the vertices $u\in V$ whose (unique) path to $r(u)$ in $F$ contains at least one edge in $E[S^*,V\sm S^*]$ (see Figure~\ref{fig:j}). Note that $U\cap K=\emptyset$ and $E_H(U,V\sm U)\s E_H(S^*,V\sm S^*)$. Moreover, suppose we first root the tree $T$ at $k$ (not $r$); then, for each vertex $u\in U$, its entire subtree is contained in $U$. Therefore, $U$ is a union of subtrees of $T$ rooted at $k$.

\begin{figure}
\centering \includegraphics[scale=.7]{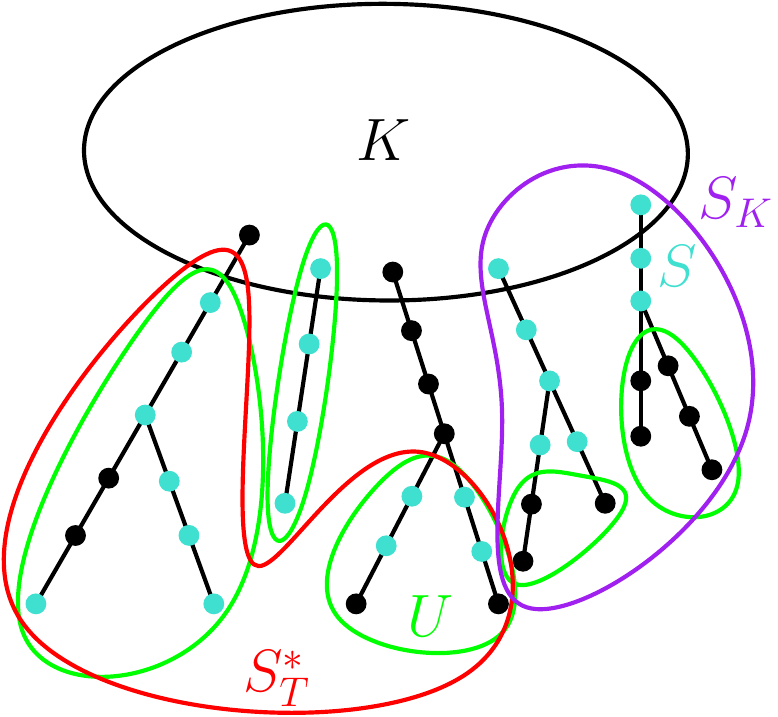}
\caption{
Cases~2a and 2b of \Cref{lem:either}. The set $S$ is the cyan vertices.
}
\label{fig:j}
\end{figure}

\para{Case 1: $r\in U$.}
In this case, let $U'\s U$ be the vertices in the subtree containing $r$. Then, if we now re-root $T$ at vertex $r$, then the vertices in $V_T\sm U'$ now form a subtree. Define $S^*_T\s V_T$ as $S^*_T:=V_T\sm U'$ (Figure~\ref{fig:j} shows a different case). By our selection of $r$, $w(S^*_T) =w(V_T\sm U')\le n/2$. Moreover, $w(S^*_T)=w(V_T\sm U')\ge |S^*|$ and $E_T(S^*_T,V\sm S^*_T)\s E_H(U,V\sm U)\s E_H(S^*,V\sm S^*)$, so the conditions $|S^*|\le w(S^*_T)\le n/2$ and $|E_T(S^*_T,V\sm S^*_T)|\le|E_H(S^*,V\sm S^*)|$ are satisfied, fulfilling condition~(1).

\para{Case 2a: $r\notin U$ and $|U|\ge n/6$.}
Since $r\notin U$, every subtree in $U$ has weight at most $n/2$. Let $U'$ be a subset of these subtrees of total weight in the range $[n/6,n/2]$.
Define $S^*_T:=U'$, which satisfies $n/2\ge|S^*_T|\ge n/6\ge|S^*|/3$ and $E_T(S^*_T,V\sm S^*_T)\s  E_H(U,V\sm U)\s E_H(S^*,V\sm S^*)$, fulfilling condition~(1).

\para{Case 2b: $r\notin U$ and $|U|<n/6$.}
In this case, let $S:=S^*\cup U$, which satisfies $|S^*|\le|S|\le|S^*|+|U|\le|S^*|+n/6\le2n/3$ and $E_H(S,V\sm S)\s E_H(S^*,V\sm S^*)$. Next, partition $S$ into $S_K$ and $S_T^*$, where $S_K$ consists of the vertices of all connected components of $H[S]$ that intersect $K$, and $S_T^*:=S\sm S_K$ is the rest. Clearly, we have $|E_H(S_K,V\sm S_K)|\le E_H(S^*,V\sm S^*)$ and  $|E_H(S_T^*,V\sm S_T^*)|\le E_H(S^*,V\sm S^*)$. Also, observe that for each tree $T$ in $F$, either $V(T)\s S_T^*$ or $V(T)\cap S_T^*=\emptyset$.

We have that either $|S_T^*|\ge|S|/2$ or $|S_K|\ge|S|/2$. If the former is true, then the set $S^*_T$ satisfies condition~(1). Otherwise, assume the latter. Since $E_H(S_{K},V\sm S_{K})$ does not contain any edges in $F$, there exists a set $S^*_K:=\{r(v):v\in S_{K}\}\s K$ such that $S_K$ is the vertices in the trees in $F$ intersecting $S^*_K$. This also means that $V\sm S_K$ is the vertices in the trees of $F$ intersecting $K\sm S^*_K$. By \Cref{clm:6.8} on $S^*_K$ and $K\sm S^*_K$, we have $|E_{H_K}(S^*_K,K\sm S^*_K)|=|E_H(S_K,V\sm S_K)|\le|E_H(S^*,V\sm S^*)|$, and 
\[ \min \{ \vol_{H_K}(S^*_K) , \vol_{H_K}(K\sm S^*_K)\} \ge \min \{ |S_K|,|V\sm S_K|\} .\]
It suffices to show that $\min\{|S_K|,|V\sm S_K|\}\ge|S^*|/3$. We have $|S_K|\le|S^*|+n/4\le3n/4$, so $\min \{ |S_K|,|V\sm S_K|\}\ge|S^*|/3$, and this along with $|S_K|\ge|S^*|$ shows that $\min\{|S_K|,|V\sm S_K|\}\ge|S_K|/3\ge|S^*|/3$.
\EP

\BL\label{lem:alg-balanced}
For a $|K|$-tree $H=(V,E)$ with core $K$ and multi-forest $F$, and for any parameter $\phi$, Algorithm~\ref{alg:most-balanced} outputs a $(\phi,O(\De\be))$-most-balanced $\max\{f(3\phi),6\phi\}$-sparse cut for $H$.
\EL

\BP
Let $S^*\s V$ the set as described in \Cref{def:most-balanced-s} ($\sigma_H(S^*)\le\phi$). We need to show that Algorithm~\ref{alg:most-balanced} outputs a cut $S$ ($|S|\le n/2$) satisfying $\sigma(S)\le \max\{f(3\phi),6\phi\}$ and $|S|\ge  |S^*|/O(\De\be)$.

First, suppose that condition~(1) of \Cref{lem:either} holds. Then, by \Cref{lem:tree-alg}, line~\ref{line:a4} of Algorithm~\ref{alg:most-balanced} returns a cut $S_T\s V_T$ with $|E[S^*,V_T\sm S^*]|/w(S^*)\le6\phi$ and $\min\{w(S_T),w(V_T\sm S_T)\}\ge w(S^*_T)/3\ge |S^*|/6$. Let $S_H:=S_T\sm \{k\} \cup K$ if $k\in S_T$ and $S_H:=S_T$ otherwise. Then, $\min\{|S_H|,|V\sm S_H|\}=\min\{w(S_T),w(V_T\sm S_T)\}\ge|S^*|/6$ and
\[ \sigma(S_H)=\f{|E_H(S_H,V\sm S_H)|}{\min\{|S_H|,|V\sm S_H|\}} = \f{|E_H(S_T,V_T\sm S_T)|}{\min\{w(S_T),w(V_T\sm S_T)\}} \le 6\phi .\]

Otherwise, suppose that condition~(2) of \Cref{lem:either} holds. Then, the cut $S^*_K$ satisfies 
\[ \Phi_{H_K}(S^*_K)=\f{|E_{H_K}(S^*_K,K\sm S^*_K)|}{\min \{ \vol_{H_K}(S^*_K) , \vol_{H_K}(K\sm S^*_K)\}}\le3\phi .\]
 By the guarantee of the recursive call in line~\ref{line:a1}, the set $S_K$ that it outputs satisfies $\vol(S_K)\le m$ and
\[\Phi_{H_K}(S_K)\le f(\Phi_{H_K}(S_K^*))\le f(3\phi)\]
and
\begin{align*}
\vol_{H_K}(S_K)&\ge (1/\be)\cd\min\{\vol_{H_K}(S^*_K),\vol_{H_K}(K\sm S^*_K)\}.
\end{align*}
Let $S_F$ be the vertices in the trees in $F$ intersecting $S_K$. By \Cref{clm:6.8} on $S_K$ and $K\sm S_K$, we have $|E_{H_K}(S_K,K\setminus S_K)| = |E(S_F,V\setminus S_F)|$ and
  $|S_F|\ge\vol(S_K)/\De$ and $|V\sm S_F|\ge\vol(K\sm S_K)/\De\ge \vol(S_K)/\De$, so 
\begin{align}
\min\{|S_F|,|V\sm S_F|\}&\ge \f1\De\vol_{H_K}(S_K) \label{eq:SF}
\\&\ge \Om\lp\f1{\De\be}\rp\cd\min\{\vol_{H_K}(S^*_K),\vol_{H_K}(K\sm S^*_K)\}\nonumber
\\ &\ge \Om\lp\f1{\De\be}\rp\cd|S^*|, \nonumber
\end{align}
where the last inequality follows by property~(2) of \Cref{lem:either}.
\EP

\subsection{Most Balanced Cut on $G$}

\BL
\label{lem:JTreeMain}
Let $G$ be a graph with $n$ vertices and $m$ edges, and fix a parameter $j$ depending on $n$.
Suppose we have an $(f(\phi),c)$-approximate balanced cut algorithm
(which was invoked on line~\ref{line:a1} of Algorithm~\ref{alg:most-balanced})
that takes time $T(\nhat, \mhat)$ on an input with $\nhat$ vertices and $\mhat$ edges.
Then there is a
$(\max\{f(3\al\De\phi),6\al\De\phi\}, O(\De^2 c))$-approximate balanced cut
algorithm that for any graph $G$ with $m$ edges and any parameter $j$,
runs in time
\[
\tO\left( \frac{m}{j} \left( T\left(j, m\right)+ m \right) \right).
\]
\EL

\BP
Let $S^*\s V$ be the set as described in \Cref{def:most-balanced-s} ($\Phi(S^*)\le\phi$). Since $\vol(S^*)\le \De|S^*|$ and $\vol(S^*)\le \vol(V\sm S^*)\le\De|V\sm S^*|$, we have $\min\{|S^*|,|V\sm S^*|\}\ge\vol(S^*)/\De$, so $\sigma(S^*)\le\De \Phi(S^*)\le\De\phi$. 

 Invoke \Cref{thm:madry} with a parameter $t:=\tO(m/j)$, computing $t$ many $j$-trees $G_1,\lds, G_t$ with $j=\tO(m/t)$. Let $\al$ be the parameter specified in \Cref{thm:madry}. 

For each $j$-tree $G_i$, run Algorithm~\ref{alg:most-balanced} with parameter $\al\De\phi$ on each of the $j$-trees. By \Cref{lem:alg-balanced-r}, Algorithm~\ref{alg:most-balanced} takes $O(m)$ time plus one call to $\m A$ for each of $t$ many $j$-trees, which is a total of $\tO(tm)=\tO(m^2/j)$ time plus $\tO(m/j)$ calls to $\m A$. %

By property~(2) of \Cref{thm:madry}, there exists a $j$-tree $G_i$ such that  $|E_{G_i}(S^*,V\sm S^*)| \le \al\, |E_G(S^*,V\sm S^*)|$, so $\sigma_{G_i}(S^*)\le\al\sigma_G(S^*)=\al\De\phi$.
By \Cref{lem:alg-balanced}, Algorithm~\ref{alg:most-balanced} returns a $(\al\De\phi,O(\De/c))$-most-balanced $\max\{f(3\al\De\phi),6\al\De\phi\}$-sparse cut $S$ for $G_i$.

Since $\vol(S)\ge|S|$ and $\vol(V\sm S)\ge|V\sm S|\ge|S|$, we have
\begin{gather} \min\{\vol(S),\vol(V\sm S)\}\ge|S|. \label{eq:volS} \end{gather}
 By property~(1) of \Cref{thm:madry}, this cut $S$ satisfies $|E_G(S,V\sm S)\le|E_{G_i}(S,V\sm S)|$, which means that
\[ \Phi_G(S) \stackrel{(\ref{eq:volS})}\le \sigma_G(S)\le\sigma_{G_i}(S)\le \max\{f(3\al\De\phi),6\phi\} . \]

By property~(2) of most-balanced sparse cut (\Cref{def:most-balanced-s}), the set $S^*$ with $\sigma_{G_i}(S^*)\le\al\De\phi$ ensures that $|S|\ge|S^*|/O(\De c)$. \[ \min\{\vol(S),\vol(V\sm S)\}\stackrel{(\ref{eq:volS})}\ge|S|\ge\f{|S^*|}{O(\De c)}\ge\f{\vol(S^*)}{O(\De^2 c)} .\]
Thus, $S$ is a $(\phi,O(\De^2 c))$-most-balanced $\max\{f(3\al\De\phi),6\al\De\phi\}$-conductance cut.
\EP

Recall that $\alpha = O(\log^2 n \log\log n)$. By setting the parameter $k=n/j$, \Cref{lem:JTreeMain} says that given an $(f(\phi), c)$-approximate balanced-cut routine with
running time $T(m)$ on an input with $n$ vertices and $m$ edges,
we can obtain an $(f(O(\phi \log^3 n)), O( c))$-approximate balanced-cut
routine on a graph with maximum degree $\Delta = O(1)$ with running time:
\[
\tilde O\left( k \left(m + T\left(\frac{m}{k} , m\right) \right) \right).
\]
We would have proven \Cref{thm:MadryMain} if there is no restriction on the maximum degree $\Delta = O(1)$. Fortunately, the following lemma show that for the most balanced lower conductance cut problem, we can in fact assume with out loss of generality that $\Delta = O(1)$ (see \Cref{sec:expander_split} for the proof):

\BL
Given a graph $G$ with $n$ vertices and $m$ edges, we can compute in linear time a graph $G'$ such that
 \BE
 \im $G'$ has $O(m)$ vertices with constant maximum degree.
 \im $\Phi(G) =\Phi(G')$.
 \im Given a $(\phi, c)$-most balanced $\al$-conductance cut in $G'$, we can transform it into
 a $(\phi,O(c))$-most balanced $O(\al)$-conductance cut in $G'$ in linear time.
 \EE
\EL

This allows us to invoke \Cref{lem:JTreeMain} with $\Delta = O(1)$,
and thus obtaining the guarantees of \Cref{thm:MadryMain}.

\section{Recursive Sparsification}
\label{sec:recursion}

\RecursionMain*

\newcommand\ww{\boldsymbol{\mathit{w}}}

The general idea here is to incorporate one of the most
well-known applications of graph partitioning:
graph sparsification.
This routine essentially allow one to transform any
graph into a sparse one by repeatedly computing
balanced separators on it, and then replacing the
expanding pieces with expanders.
Then we use another divide and conquer scheme to avoid
calling sparsify on the initial, possibility dense,
graph ith $m$ edges.
Instead, we will partition the vertices into $b$ equal
sized parts, and recursively sparsify each of the
$O(b^2)$ subgraphs on about $2n / b$ vertices.
The resulting graphs are then combined with another call
to sparsification.
This scheme essentially trades the multiplicative accumulation
of errors from repeatedly calling sparsification
with the smaller sizes throughout these calls.

\begin{definition}
	\label{def:Approx}
	A weighted graph $H=(V,E^{(H)}, w^{(H)})$ is a
	$\kappa$-approximation of another weighted graph
	$G=(V,E^{(G)},w^{(H)})$ on the same set of vertices $V$
	if for any cut $S\subseteq V$, we have
	\[
	\frac{1}{\kappa} \cdot \ww^{\left(G\right)}
	\left(E^{\left(G\right)}\left(S,\overline{S}\right)\right)
	\leq
	\ww^{\left(H\right)}	\left(E^{\left(H\right)}\left(S,\overline{S}\right)\right)
	\leq
	\kappa \cdot \ww^{\left(G\right)}
	\left(E^{\left(G\right)}\left(S,\overline{S}\right)\right).
	\]
\end{definition}

This notion is a simplified notion of graph sparsification,
and we can show modify the existing sparsification literature
to produce $n^{o(1)}$-approximations deterministically,
when given an approximate balanced cut procedure.

\begin{lemma}
	\label{lem:SparsifyHandWave}
	Given any $(f(\phi), \beta)$-approximate balanced-cut routine
$\textsc{ApproxBalCut}$ in time $m^\theta$ for some $1<\theta\le 2$ such that
$f(\phi) \leq \phi^{-\xi} n^{o(1)}$ for some absolute constant $\xi > 0$ (as specified
	in Theorem~\ref{thm:RecursionMain}),
	
	we can construct
	a deterministic sparsification algorithm $\textsc{DeterministicSparsify}$ that takes any 
	weighted graph $G$ as input, and outputs in
	deterministic $\oh(m^\theta)$ time a sparse graph $H$ with
	$\oh(n)$ edges that $n^{o(1)}$-approximates $G$.
\end{lemma}

\begin{proof}
	By the minimum spanning tree based hierarchical invocation
	given in the weighted sparsifier section of~\cite{SpielmanT11},
	it suffices to give such an algorithm for unweighted graphs.
	
	On such graphs, the equivalence between finding approximate
	balanced cuts and almost-expanders~\cite{SaranurakW19}
	as stated in \Cref{thm:det cond}
	means we also have a routine that either finds an $1/2$-balanced
	cut, or a expander of size at least half the graph.
	Then by repeatedly invoking this partition routine, we
	obtain in $\oh(m^\theta)$ time a partition
	of the vertices into expanders with conductance at
	least $n^{-o(1)}$, so that at most
	\[
		f\left(\phi \right) m \log{n}
		\leq \frac{m}{2}
	\]
	edges are between the pieces.
	Repeating this process $O(\log{n})$ iterations then
	puts all the edges into expanders,
	and thus gives a total number of vertices of $O(n \log^2{n})$.
	
	Then the construction of weighted expanders
	from Appendix J of~\cite{KyngLPSS16} gives that each of
	these expanders can be $n^{o(1)}$-approximated by
	a graph with average degree $O(1)$.
	These constructions in turn rely on the explicit
	expander constructions by either Margulis~\cite{Margulis88},
	or by Lubotzky, Phillips, and Sarnak~\cite{LubotzkyPS88},
	both of which are determinsitic.
	Such a replacement, however, incurs an error equaling to
	conductance (expanders have constant conductace),
	which in turn goes into the overall approximation factor.
\end{proof}

This notion of approximation composes under summation
of graphs, as well as compositions.
\begin{fact} (see e.g. Section 2 of~\cite{KyngLPSS16})
\label{fac:approx_sum}
\begin{itemize}
	\item	If $G_1$ $\kappa_1$-approximates $H_1$,
	and $G_2$ $\kappa_2$-approximates $H_2$,
	then $G_1 + G_2$ $\max\{\kappa_1, \kappa_2\}$-approximates
	$H_1 + H_2$.
	\item If $G_1$ $\kappa_1$-approxiamtes $G_2$,
	and $G_2$ $\kappa_2$-approximates $G_3$,
	then $G_1$ $\kappa_1\kappa_2$-approximates $G_3$.
\end{itemize}
\end{fact}

Pseudocode of our routine is then given in \Cref{alg:4way}.

\begin{algorithm}[H]
	\caption{$\textsc{RecursiveSparsify}(G)$}
	\label{alg:4way}
	
	\begin{algorithmic}[1] 
		\State Let $m$ be the number of edges of $G$\;
		\State Let $n$ be the number of vertices of $G$\;
		\If{$m\le b\cdot n$}
			\State Return $\textsc{DeterministicSparsify}(G)$\label{lin:4}\;
		\Else
			\State Partition $V(G)$ into $b$ parts $V_1,\ldots,V_b$ such that their sizes differ by at most $1$\label {lin:6}\;
			\State Decompose $G$ into $\frac{b(b-1)}{2}$ subgraphs: $G_{i,j}=(V_i\cup V_j, E(G)(V_i,V_j))$\label {lin:7}\;
			\For{$1\le i\le j\le b$}
				\State $H_{i,j}$=\textsc{RecursiveSparsify}($G_{i,j}$)\label{lin:9}\;
			\EndFor
			\State Return $\textsc{DeterministicSparsify}(\sum_{ij} H_{i, j})$\label{lin:10}\;
		\EndIf
	\end{algorithmic}
	
\end{algorithm}

\paragraph{Proof of Theorem~\ref{thm:RecursionMain}.}
	 We may assume that $G$ is simple. Let $n_0$ be the number of vertices and edges in the original graph $G_0$ on which we call $\textsc{RecursiveSparsify}$ (\Cref{alg:4way}). It suffices to show that \Cref{alg:4way} runs in $O(n^{\frac{(1+c)(2\theta-2)}{c}}m^{2-\theta})$ time on any graph $G$ with $n$ vertices and $m$ edges such that $n\le n_0$ if $b=3n_0^{\frac{1}{c}}$ and returns a $n^{o(1)}$ approximation of $G$ with $\oh(n)$ edges, since we can then call the given $\textsc{ApproxBalCut}$ routine on the returned graph.\\
	
	The time cost of Line \ref{lin:6} and \ref{lin:7} and passing function arguments are bounded by $O(m)=O(n^{2\theta - 2}m^{2-\theta})$. So we focus on bounding the total running time of recursive calls (Line \ref{lin:4}, \ref{lin:9} and \ref{lin:10}), denoted by $T(n, m)$. 
	We will prove by induction that $T(n, m)$ is no more than $(nb)^{2\theta-2}m^{2-\theta}100^{\log_{b/3} (n)}$. 
		
	As base case, if $m\le b\cdot n$, by Lemma \ref{lem:SparsifyHandWave}, $\textsc{DeterministicSparsify}$ computes an $n^{o(1)}$-approximation of any graph in $\oh(m^\theta)$ time. So \Cref{alg:4way} runs in $m^\theta$ time which is no more than $(nb)^{2\theta-2}m^{2-\theta}$.
	
	Otherwise, $m>b\cdot n\ge 3n$. This implies $n\ge 6$.

	If $m>b\cdot n$, 
	\begin{align*}
	T(n, m)\le& (b^2(n/b))^\theta + \sum_{1\le i\le j\le b} T(2n/b+1, m_{i,j})100^{\log_{b/3} (2n/b+1)}\\
	\le& (b^2(n/b))^\theta+\sum_{1\le i\le j\le b} (9n^{2\theta-2}m_{i,j}^{2-\theta})100^{\log_{b/3} (2n/b+1)}\\
	&\text{(by $((2n/b+1)b)^{2\theta-2}\le (3n)^{2\theta-2}$)}\\
	\le& (b^2(n/b))^\theta+9\sum_{1\le i\le j\le b} (n^{2\theta-2}(2m/b/(b-1))^{2-\theta})100^{\log_{b/3} (2n/b+1)}\\
	=& (b^2(n/b))^\theta+9b^2(n^{2\theta-2}(2m/b/(b-1))^{2-\theta})100^{\log_{b/3} (2n/b+1)}\\
	\le& (b^2(n/b))^\theta+36 b^2(n^{2\theta-2}(m/(b^2))^{2-\theta})100^{\log_{b/3} (2n/b+1)}\\
	=& (b^2(n/b))^\theta+36 b^{2\theta-2}(n^{2\theta-2}m^{2-\theta})100^{\log_{b/3} (2n/b+1)}\\
	\le& (b^2(n/b))^\theta+36 b^{2\theta-2}(n^{2\theta-2}m^{2-\theta})100^{-1+\log_{b/3} (n)}\\
	\le& (b^2(n/b))^\theta+0.36b^{2\theta-2}n^{2\theta-2}m^{2-\theta}100^{\log_{b/3} (n)}\\
	\le& b^{2\theta-2}n^{2\theta-2}m^{2-\theta}100^{\log_{b/3} (n)}
	\end{align*}
	for $n\ge 6$. The last inequality is by $(nb)^\theta<m^\theta=m^{\theta-1}m^{2-\theta}\le n^{2\theta-2}m^{2-\theta}$.\\
	
	The number of edges returned by Algorithm \ref{alg:4way} is equal to that number of $\textsc{DeterministicSparsify}$, which is $\oh(n)$.\\
	
	Since there are at most $\log_{b/3}n=c$ layers of recursion and each layer computes a $n^{o(1)}$-approximation (by $\textsc{DeterministicSparsify}$) of the graph returned by the previous layer (Fact \ref{fac:approx_sum}), Algorithm \ref{alg:4way} returns a $\left(n^{o(1)}\right)^c=n^{o(1)}$-approximation of $G$.

Note however that this accumulation of errors is due to the
approximation of a piece with $\phi$ conductance with a
regular expander.
We believe this lack of better deterministic approximations
of expanders is an inherent gap in the construction of
deterministic sparsification / partitioning tools.
While it does not affect our overall final performances,
it is a question worth investigation on its own.

\section*{Acknowledgement}

        This project has received funding from the European Research
        Council (ERC) under the European Union's Horizon 2020 research
        and innovation programme under grant agreement No
        715672 and 759557. Nanongkai was also partially supported by the Swedish
        Research Council (Reg. No. 2015-04659.).
        Yu Gao was partially supported by the National Science Foundation under Grant No. 1718533.

\appendix 
\section{Notations}
\label{sec:notations}

We use $G = (V,E)$ to denote an \textit{undirected} graph,
and $|V| = n$ and $|E| = m$ to denote the number of edges and vertices respectively.
We assume that $G$ is connected because otherwise there is a trivial separator
of $0$ vertices/edges.
We also assume that $G$ does not have parallel edges: otherwise,
we can remove duplicate edges without changing vertex connectivity.  

\begin{definition} [$\deg,\vol,N$] For any vertex $v$ on graph $G$, and any subset of vertex $U \subseteq V$, 
	\begin{itemize}[nolistsep, noitemsep]
		\item $\deg_G(v) = $ number of edges incident to $v$. 
		\item $\vol_G(U) = \sum_{v\in U}\deg_G(v)$. 
		\item $N_G(v) = \{u \colon (u,v) \in E\}$, i.e., $N_G(v)$ is the set of neighbors of $v$.
		\item $N_G(U) = \bigcup_{v \in U} N_G(v) \setminus U$.  Note that $U$ is excluded. 
	\end{itemize}
	We omit subscription when the graph that we refer to is clear from the context.
\end{definition}

\begin{definition}[Subgraphs] For any set of vertices $U \subseteq V$, we denote $G[U]$ as an \textit{induced} subgraph by $U$. For any vertex set $U$ and edge set $F$, we denote 
	\begin{itemize} [nolistsep,noitemsep]
		\item $G - U = (V \setminus U, E)$, and
		\item $G - F = (V, E \setminus F)$
	\end{itemize}
\end{definition}

\begin{definition} [Edge-cuts and vertex-cuts]
	Let $x,y$ be any distinct vertices. We call any edge-set $C \subset E$ (respectively any vertex-set $U \subset V$):%
	\begin{itemize} [nolistsep,noitemsep]
		\item an $(x,y)$-\textit{edge-cut} (respectively an $(x,y)$-\textit{vertex-cut}) if there is no path from $x$ to $y$ in $G - C$ (respectively if there is no path from $x$ to $y$ in $G - U$ and $x,y \not \in U$),
		\item an \textit{edge-cut} (respectively an \textit{vertex-cut}) if it is an $(s,t)$-edge-cut (respectively $(s,t)$-vertex-cut) for some distinct vertices $s$ and $t$. 
	\end{itemize}
\end{definition}

\begin{definition} [Separation triple] \label{def:sep-triple}
	A \textit{separation triple} $(L,S,R)$ is an order triplet of sets forming a partition of $V$ where $L$ and $R$ are non-empty, and there is no edge between $L$ and $R$.
\end{definition}

Note that $S$ is an $(x,y)$-vertex-cut for any $x \in L$ and $y \in R$.

\begin{definition} [Edge set]
	We denote  $E(S,T) = \{(u,v)\colon u \in S, v \in T, \text{ and } v \in E \}$.
\end{definition}

\begin{definition} [Vertex connectivity $\kappa$]
	Vertex connectivity of a graph $G$, denoted as $\kappa_G$, is the minimum cardinality vertex-cut or $n-1$ if no vertex-cut exists. For any vertices $x,y \in V$, we denote $\kappa_G(x,y)$ as the smallest cardinality $(x,y)$-vertex-cut or $n-1$ if $(x,y)$-vertex-cut does not exist. 
\end{definition}

Observe that $\kappa_G = \min\{ \kappa_G(x,y) \colon x,y \in V, x \not = y \}$.

\section{Split Vertex Connectivity} 
\label{sec:splitvc}

\begin{theorem} \label{thm:split-vc-runtime}
There is a deterministic SplitVC algorithm (\Cref{def:splitvc}) that runs in $O(mk(|S|+k^2))$ time.
\end{theorem}

We prove \Cref{thm:split-vc-runtime} by giving an algorithm, and analyzing its correctness and running time. 

\subsection{Algorithm}
 
\begin{algorithm}[H]
\caption{SplitVC$(G,S,k)$}
Input: Graph $G = (V,E)$, a vertex-cut $S$, and a positive integer $k$ \\
Assumptions: $|S| \geq k$. \\%$G$ has aboricity $k$. $n$ is the number of vertices of the original graph, and  $k < n$. \\%,%
Output: An $(x,y)$-vertex-cut of size $< k$ for some $x \in S$ and $y
\in S$ or the symbol $\perp$ certifying that $\kappa_G(x,y) \geq k $ for all $x \in S$ and $y \in S$. %
\begin{algorithmic}[1] \label{alg:split-vc}
\State Let $X$ be any subset of size $k$ from $S$.\label{line:X} \Comment{$X$ exists
 since $|S| \geq  k$. }
\If{$\min_{x \in X, y \in X} \kappa_G(x,y) < k $} \label{line:first-loop} %
\State \Return the corresponding $(x,y)$-vertex-cut in $G$. \label{line:pair-exists}
\EndIf
\State Let $G'$ be a graph obtained from $G$ by adding a new vertex
$s$, and edges $(s,v)$ for all $v \in X$.  \label{line:newg}
\If{$\min_{v \in S} \kappa_{G'}(s,v) < k$ } \label{line:check-newg}%
\State \Return  A $(u,v)$-vertex-cut in $G$ where $u \in S$, $v \in S$.  \label{line:vcut2} \Comment{See \Cref{pro:get-vcut}.}
\EndIf
\State \Return $\perp$.  \label{line:last-line}
\end{algorithmic}
\end{algorithm}

\subsection{Analysis} 

We show that \Cref{alg:split-vc} is correct.  Let $X$ be the set as defined in \Cref{alg:split-vc} (line~\ref{line:X}).  Recall that $s$ is the new vertex in $G'$.  Let $\kappa' = \min_{v \in S}
\kappa_{G'}(s,v)$. 

\begin{proposition} \label{pro:get-vcut}
If $\kappa' < k$, then the corresponding $(s,v)$-vertex-cut in $G'$ is also a $(u,v)$-vertex-cut in $G$ for some $u \in S$ and some $v \in S$.
\end{proposition}
\begin{proof}
 Let $(L',S',R')$ be a separation triple such that
$s \in L'$ and $|S'| = \kappa'$ and $v \in R' \cap S$ (such separation triple exists by \Cref{alg:split-vc} line~\ref{line:check-newg}). Since $s \in L'$ and there cannot be an edge between $L'$ and $R'$, we have $N_{G'}(s) \subseteq L' \sqcup S'$. Since $|S'| < k$
, but $|N_{G'}(s)| = |X| = k$, there is some vertex $u$ in $X$ that is
also in $L'$. That is, $L' \setminus \{ s \} \not  =
\emptyset$. Also, $R'$ contains a vertex $v \in S$. Hence, we get a new separation triple in $G$ by removing $s$ from $L'$  Therefore, $S'$ is a $(u,v)$-vertex-cut in $G$ where $u \in S$ and $v \in S$.
\end{proof}

\begin{lemma} \label{lem:split-vc-correct1} Suppose there exist $x$ and $y$ such that $x \in S, y
  \in S$ and $\kappa_G(x,y) < k $. Then, \Cref{alg:split-vc}
  (line~\ref{line:pair-exists} or line~\ref{line:vcut2}) returns a $(u,v)$-vertex-cut of size at most $k$ where $u\in S$ and $v \in S$. 
\end{lemma} 
\begin{proof}
Let  $(L^*, S^*, R^*)$ be a separation triple such that $x \in L^*$ and
  $y \in R^*$ and $|S^*| = \kappa_G(x,y)$.  We have that $S \setminus
  S^*$ has two components $S \cap L^*$ and $S \cap R^*$.%

If $X  \cap L^* \not = \emptyset$ and $X  \cap R^* \not =
  \emptyset$, then there exist $u \in X \cap L^*$ and $v \in
  X  \cap R^*$. Hence, $u \in  L^*$ and $v \in
  R^*$. Therefore, $\kappa_G(u,v) = |S^*| < k$.  Since $u \in X$ and
  $v \in X$, \Cref{alg:split-vc}(line~\ref{line:pair-exists})  returns
  a vertex-cut of size at most $k$. 

Otherwise, $X   \cap L^* = \emptyset$ or $X  \cap
R^* = \emptyset$. Since $|S^*| < k,$ and $|X| = k$, either  $X
\cap L^* = \emptyset$ or $X \cap R^* = \emptyset$. Now, we
assume WLOG that $X  \cap R^* = \emptyset$.  This means   $X \subseteq L^*
\sqcup S^*$. Let $G'$ be the graph as defined in
\Cref{alg:split-vc} (line~\ref{line:newg}). Let $\kappa' = \min_{v \in S}
\kappa_{G'}(s,v)$. 

We claim that  $\kappa' \leq  \kappa_G(x,y)$.  Recall that  $(L^*, S^*, R^*)$ is a separation triple such that $x \in S \cap L^*, y \in  S \cap R^*$ and $|S^*| = \kappa_G(x,y)$.  We show that $S^*$ is an
  $(s,y)$-vertex-cut in $G'$. Note that $s$ is a new vertex in $G'$,
  and $y \in S$.    Since  $N_{G'}(s) = X \subseteq L^* \sqcup S^*$,   the new edges do not join $L^*$ and $R^*$. Also, $y \in R^*$. Hence, $(L^* \cup \{ s \}, S^*, R^*)$
  is a separation triple in $G'$ where $s$ and $y$ belong to different
  partitions. Therefore, $S^*$ is an  $(s,y)$-vertex-cut in $G'$, and
  we have $\kappa_{G'}(s,y) \leq |S^*| = \kappa_G(x,y)$. Therefore,
  $\kappa' \leq  \kappa_G(x,y) < k$.

Therefore, by \Cref{pro:get-vcut}, \Cref{alg:split-vc}(line~\ref{line:vcut2}) outputs correctly a $(u,v)$-vertex-cut in $G$ of size $< k $ where $u \in S$ and $v\in S$. %

\end{proof}

We show the last part. 
\begin{lemma} \label{lem:split-vc-correct2}
Suppose $\kappa_G(x,y) \geq k$ for all $x \in S$ and $y \in S$. \Cref{alg:split-vc} returns the symbol $\perp$ at line~\ref{line:last-line}. 
\end{lemma}
\begin{proof}
Clealy, \Cref{alg:split-vc} never returns a vertex-cut at line~\ref{line:pair-exists}.   Let $\kappa' = \min_{v \in S}
\kappa_{G'}(s,v)$. It remains to show that $\kappa' \geq k$. Suppose $\kappa'\ < k$. Let $S'$ be the corresponding $(s,v)$-vertex-cut in $G'$ where $v \in S$. By \Cref{pro:get-vcut}, $S'$ is also a $(u,v)$-vertex-cut in $G$ where $u \in S$ and $v \in S$. Therefore, $\kappa_G(u,v) < k$, which is a contradiction. 
Therefore,  \Cref{alg:split-vc} never returns a vertex-cut at line~\ref{line:vcut2}, and correctly returns the symbol $\perp$ at line~\ref{line:last-line}. 

\end{proof}

\begin{lemma}  \label{lem:split-vc-runtime}
\Cref{alg:split-vc} terminates in $O(mk(|S|+k^2))$ time. 
\end{lemma}
\begin{proof}
Given $x,y$, we can compute an $(x,y)$-vertex-cut such that  $\kappa_G(x,y) < k$ or certify that $\kappa_G(x,y) \geq k$ in $O(mk)$ time using Ford-Fulkerson algorithm.  We run at most $|X|^2 = k^2$ calls of $\kappa(x,y)$ at line~\ref{line:first-loop}, and at most $|S|$ calls of $\kappa(x,y)$ at line~\ref{line:check-newg}. Each call of $\kappa(x,y)$ takes $O(mk)$. Therefore, the running time for \Cref{alg:split-vc} follows.%
\end{proof}

\begin{proof}[Proof of \Cref{thm:split-vc-runtime}]
This follows from \Cref{alg:split-vc} is correct by \Cref{lem:split-vc-correct1}, and \Cref{lem:split-vc-correct2} for the two possible cases. The running time follows from \Cref{lem:split-vc-runtime}. 
\end{proof}

\section{Expander Split}
\label{sec:expander_split}
\begin{definition}
[Expander Split]\label{def:exp split}Let $G=(V,E)$ be a graph.
The \emph{expander split} graph $G'$ of $G$ is obtained from $G$
by the following operations 
\begin{itemize}
\item For each node $u\in V$, we replace $u$ by a constant-degree expander
$X_{u}$ with $\deg(u)$ nodes. We call $X_{u}$ a \emph{super-node}
in $G'$.
\item Let $E_{u}=\{e_{u,1},\dots,e_{u,\deg(u)}\}$ denote the set of edges
in $G$ incident to $u$. For each $e=(u,v)$, suppose $e=e_{u,i}=e_{v,j}$,
we create add an edge between the $i$-th of node $X_{u}$ and the
$j$-th node of $X_{v}$.
\end{itemize}
\end{definition}

\begin{proposition}
\label{prop:exp split}For any $m$-edge graph $G=(V,E)$, the expander
split graph $G'=(V',E')$ of $G$ has the following properties
\begin{enumerate}
\item $G'$ has $O(m)$ vertices with constant maximum degree, and can be
obtained from $G$ in $O(m)$ time.
\item $\Phi_{G'}=\Theta(\Phi_{G})$.
\item Given a $\beta$-balanced cut $S$ in $G'$ where $\Phi_{G'}(S)\le\epsilon$
for some small enough constant $\epsilon<1$, then we can obtain in
$O(m)$ time a $\Omega(\beta)$-balanced cut $T$ in $G$ where $\Phi_{G}(T)=O(\Phi_{G'}(S))$.
\end{enumerate}
\end{proposition}

\Cref{prop:exp split} allows us to assume that we only work with a
graph with constant degree.

\subsection{Proof of \Cref{prop:exp split}}

We prove \Cref{prop:exp split} here. This part should be skipped in
the first read.
\begin{proposition}
\label{prop:explicit expander}
[Fast explicit expanders]Given any number $n$, there is a deterministic
algorithm with running time $O(n)$ that constructs a graph $H_{n}$
with $n$ vertices such that each vertex has degree at most 16, and
the conductance $\Phi_{H_{n}}=\Omega(1)$. 
\end{proposition}

\begin{proof}
We assume that $n\ge10$, otherwise $H_{n}$ can be constructed in
constant time. The expander construction by Margulis, Gabber and Galil
is as follows. For any number $k$, the $H'_{k^{2}}$ is a vertex
set $\mathbb{Z}_{k}\times\mathbb{Z}_{k}$ where $\mathbb{Z}_{k}=\mathbb{Z}/k\mathbb{Z}$.
For each vertex $(x,y)\in\mathbb{Z}_{k}\times\mathbb{Z}_{k}$, its
eight adjacent vertices are $(x\pm2y,y),(x\pm(2y+1),y),(x,y\pm2x),(x,y\pm(2x+1))$.
In \cite{GabberG81}, it is shown that $\Phi_{H'_{k^{2}}}=\Omega(1)$.

Let $k$ be such that $(k-1)^{2}<n\le k^{2}$. As $n\ge10$, so $k\ge4$,
and so $(k-1)^{2}\ge k^{2}/2$. So we can contract disjoint pairs
of vertices in $H'_{k^{2}}$ and obtain a graph $H_{n}$ with $n$
nodes where each node has degree between $8$ and $16$. Note that
$\Phi_{H_{n}}\ge\Phi_{H'_{k^{2}}}$. It is clear that the construction
takes $O(n)$ time.
\end{proof}
Let $\cP=\{V_{1},\dots V_{k}\}$ be a partition of $V$. We say that
a cut $S$ \emph{respects} $\cP$ if for each $i$, either $V_{i}\subseteq S$
or $V_{i}\cap S=\emptyset$ (i.e. no overlapping). Let $\Phi_{G}^{out}=\min_{S\text{ respects }\text{\ensuremath{\P}}}\Phi_{G}(S)$.
Let $\Phi_{G}^{in}=\min_{i}\Phi_{G[V_{i}]}$. We say that $V_{i}$
is a \emph{clump }if, for each $u\in V_{i}$, $\deg_{G[V_{i}]}(u)=\Theta(\deg_{G}(u))$.
In particular, for every $S\subset V_{i}$, we have $\vol_{G[V_{i}]}(S)=\Theta(\vol_{G}(S))$.%

\begin{lemma}
\label{prop:conductance preserve}Suppose that $\cP=\{V_{1},\dots V_{k}\}$
is a partition of $V$ where each $V_{i}$ is a clump. We have
\begin{enumerate}
\item $\Phi_{G}=\Omega(\Phi_{G}^{out}\cdot\Phi_{G}^{in})$.
\item Given a $\beta$-balanced cut $S$ where $\Phi_{G}(S)\le\epsilon\Phi_{G}^{in}$
for some small enough constant $\epsilon<1$, then we can obtain in
$O(\vol_{G}(V))$ time a $\Omega(\beta)$-balanced cut $T$ respecting
 where $\Phi_{G}(T)=O(\Phi_{G}(S)/\Phi_{G}^{in})$.
\end{enumerate}
\end{lemma}

\begin{proof}
Below, we write $a\lesssim b$ to denote $a=O(b)$. Consider any $\beta$-balanced
cut $(S,V-S)$ in $G$ where $\vol_{G}(S)\le\vol_{G}(V-S)$. We will
prove that either 1) $\Phi_{G}(S)=\Omega(\Phi_{G}^{in})$, otherwise
we can obtain in $O(\vol_{G}(V))$ time a $\Omega(\beta)$-balanced
cut $(T,V-T)$ respecting $\cP$ such that $\Phi_{G}(T)=O(\Phi_{G}(S)/\Phi_{G}^{in})$.
Observe that this implies both the first part of the lemma, i.e. $\Phi_{G}=\Omega(\Phi_{G}^{out}\cdot\Phi_{G}^{in})$,
and also the second part.

Let $\cP'=\{V_{i}\mid0<\vol_{G}(S\cap V_{i})<2\vol_{G}(V_{i}-S)\}$
and $\cP''=\{V_{i}\mid\vol_{G}(S\cap V_{i})>2\vol_{G}(V_{i}-S)\}$.
Note that as $V_{i}$ is a clump, $\vol_{G[V_{i}]}(S\cap V_{i})=\Theta(\vol_{G}(S\cap V_{i}))$
and $\vol_{G[V_{i}]}(V_{i}-S)=\Theta(\vol_{G}(V_{i}-S))$. So $\delta_{G[V_{i}]}(V_{i}\cap S)\ge\Omega(\Phi_{G}^{in}\cdot\vol_{G[V_{i}]}(V_{i}\cap S))$
for each $V_{i}\in\cP'$ and $\delta_{G[V_{i}]}(V_{i}-S)\ge\Omega(\Phi_{G}^{in}\cdot\vol_{G[V_{i}]}(V_{i}-S))$
for each $V_{i}\in\cP''$. Let $T=\bigcup_{V_{i}\in\cP''}V_{i}$, $\overline{T}=V-T$
and $\overline{S}=V-S$. There are two cases.

In the first case, suppose $\vol_{G}(S-T)\ge\vol_{G}(S)/2$. Then,
we have
\begin{align*}
\delta_{G}(S) & \ge\sum_{V_{i}\in\cP'}\delta_{G[V_{i}]}(S\cap V_{i}).\\
 & \ge\Omega(\Phi_{G}^{in})\cdot\sum_{V_{i}\in\cP'}\vol_{G[V_{i}]}(S\cap V_{i})\\
 & \ge\Omega(\Phi_{G}^{in})\cdot\sum_{V_{i}\in\cP'}\vol_{G}(S\cap V_{i})\\
 & =\Omega(\Phi_{G}^{in})\cdot\vol_{G}(S-T)\\
 & \ge\Omega(\Phi_{G}^{in})\cdot\vol_{G}(S)/2
\end{align*}
So $\Phi_{G}(S)=\Omega(\Phi_{G}^{in})$.

In the second case, suppose $\vol_{G}(S-T)\le\epsilon\vol_{G}(S)$.
We will show that (1) $\delta_{G}(T)=O(\delta_{G}(S)/\Phi_{G}^{in})$,
(2) $\vol_{G}(T)=\Omega(\vol_{G}(S))$, and (3) $\vol_{G}(\overline{T})=\Omega(\vol_{G}(\overline{S}))$.
This would imply 
\begin{align*}
\Phi_{G}(T) & =\frac{\delta_{G}(T)}{\min\{\vol_{G}(T),\vol_{G}(\overline{T})}\\
 & \lesssim\frac{\delta_{G}(S)/\Phi_{G}^{in}}{\min\{\vol_{G}(S),\vol_{G}(\overline{S})\}}\\
 & =\Phi_{G}(S)/\Phi_{G}^{in}.
\end{align*}
and that $T$ is a $\Omega(\beta$)-balanced cut. Now, it remains
to prove the three claims.
\begin{claim}
We have the following:
\end{claim}

\begin{itemize}
\item $\delta_{G}(T)=O(\delta_{G}(S)/\Phi_{G}^{in})$, 
\item $\vol_{G}(S)=O(\vol_{G}(T))$, and
\item $\vol_{G}(\overline{S})=O(\vol_{G}(\overline{T}))$. 
\end{itemize}
\begin{proof}
It is convenient to bound $\vol_{G}(T-S)$ and $\vol_{G}(S-T)$ first.
We have 
\begin{align*}
\vol_{G}(T-S) & =\sum_{V_{i}\in\cP''}\vol_{G}(V_{i}-S)\\
 & \lesssim\sum_{V_{i}\in\cP''}\vol_{G[V_{i}]}(V_{i}-S)\\
 & \lesssim\sum_{V_{i}\in\cP''}\delta_{G[V_{i}]}(V_{i}-S)/\Phi_{G}^{in}\\
 & \le\delta_{G}(V-S)/\Phi_{G}^{in}=\delta_{G}(S)/\Phi_{G}^{in}.
\end{align*}
Next, 
\begin{align*}
\vol_{G}(S-T) & =\sum_{V_{i}\in\cP'}\vol_{G}(S\cap V_{i})\\
 & \lesssim\sum_{V_{i}\in\cP'}\vol_{G[V_{i}]}(S\cap V_{i})\\
 & \lesssim\sum_{V_{i}\in\cP'}\delta_{G[V_{i}]}(S)/\Phi_{G}^{in}\\
 & \le\delta_{G}(S)/\Phi_{G}^{in}.
\end{align*}
Both of the two bounds above exploit the fact that $V_{i}$ is a clump.
From this, we obtain the first part of the claim:

\begin{align*}
\delta_{G}(T) & \le\delta_{G}(S)+\vol_{G}(T-S)+\vol_{G}(S-T)\\
 & =O(\delta_{G}(S)/\Phi_{G}^{in}).
\end{align*}
For the second part, we have 
\begin{align*}
\vol_{G}(S) & \le\vol_{G}(T)+\vol_{G}(S-T)\\
 & \le\vol_{G}(T)+\vol_{G}(S)/2
\end{align*}
and so $\vol_{G}(S)=O(\vol_{G}(T))$. For the last part, first observe
that
\begin{align*}
\vol_{G}(T-S) & =\sum_{V_{i}\in\cP''}\vol_{G}(V_{i}-S)\\
 & <\sum_{V_{i}\in\cP''}\vol_{G}(V_{i}\cap S)/2\\
 & =\vol_{G}(T\cap S)/2.
\end{align*}
So we have
\begin{align*}
\vol_{G'}(\overline{S}) & =\vol_{G}(\overline{T})+\vol_{G}(T-S)\\
 & <\vol_{G}(\overline{T})+\vol_{G}(T\cap S)/2\\
 & \le\vol_{G}(\overline{T})+\vol_{G}(S)/2\\
 & \le\vol_{G}(\overline{T})+\vol_{G}(\overline{S})/2,
\end{align*}
and so $\vol_{G}(\overline{S})=O(\vol_{G}(\overline{T}))$.
\end{proof}
\end{proof}
Now we are ready to prove \Cref{prop:exp split}.
\begin{proof}
[Proof of \Cref{prop:exp split}]For 1), property follow immediately
from the definition of expander split graph and \Cref{prop:explicit expander}.
For 2), $\Phi_{G'}\le\Phi_{G}$ by the construction. To show that
$\Phi_{G'}=\Omega(\Phi_{G})$, let $\cP$$=\{X_{u}\}_{u\in V}$ be
a partition of $V'$. For any cut $S'$ in $G'$ respecting $\cP$,
there is a corresponding cut $S$ in $G$. Note that $\delta_{G'}(S')=\delta_{G}(S)$
and $\vol_{G'}(S')=\Theta(\vol_{G}(S))$. So $\Phi_{G'}^{out}=\Theta(\Phi_{G})$.
By \Cref{prop:explicit expander}, $\Phi_{G'}^{in}=\Theta(1)$. Note
that each node $u$ in $G'$ is such that $\deg_{G'}(u)=\Theta(\deg_{X_{u}}(u))$.
In particular, each super-node $X_{u}$ is a clump in $G'$. By \Cref{prop:conductance preserve}
we have $\Phi_{G'}=\Omega(\Phi_{G'}^{out}\Phi_{G}^{in})=\Omega(\Phi_{G})$.
For 3), this follows from \Cref{prop:conductance preserve} as well.
\end{proof}

\newcommand{\etalchar}[1]{$^{#1}$}

\end{document}